\documentclass[conference]{IEEEtran}
\usepackage{amssymb}
\usepackage{graphicx}

\usepackage{url}
\urldef{\mailsa}\path|{alfred.hofmann, ursula.barth, ingrid.haas, frank.holzwarth,|
\urldef{\mailsb}\path|anna.kramer, leonie.kunz, christine.reiss, nicole.sator,|
\urldef{\mailsc}\path|erika.siebert-cole, peter.strasser, lncs}@springer.com|

\usepackage{etex}
\usepackage{mathtools,xspace}
\usepackage{amsmath}
\usepackage{amsfonts}
\usepackage{amssymb}
\usepackage{amsthm}

\usepackage[all]{xy}
\SelectTips{cm}{}  
\xyoption{v2}
\xyoption{curve}
\xyoption{2cell}
\UseAllTwocells
\SilentMatrices

\usepackage{wrapfig}
\usepackage{floatflt}
\usepackage{tikz}

\usepackage{cancel}

\usepackage{stmaryrd}

\newcommand{\pow}{\mathcal{P}}
\newcommand{\dist}{\mathcal{D}}

\newcommand{\place}{\underline{\phantom{n}}\,} 
\newcommand{\Sets}{\mathbf{Sets}}
\newcommand{\Ord}{\mathbf{Ord}}
\newcommand{\Meas}{\mathbf{Meas}}

\newcommand{\op}{\mathrm{op}}

\newcommand{\uniquefp}[1]{(\!|#1|\!)}

\newcommand{\bigbullet}{\raisebox{.8mm}{\;\circle*{6}}\!\!}
\newcommand{\middlebullet}{\raisebox{1mm}{\,\circle*{4}}\!\!}

\newcommand{\Reach}{\mathrm{Reach}}

\newcommand{\defarrow}{\stackrel{\text{def.}}{\Leftrightarrow}}

\newcommand{\id}[0]{\mathrm{id}}

\newcommand{\sem}[1]{\llbracket#1\rrbracket}

\newcommand\Acc{\mathsf{Acc}}

\newcommand{\Tree}[0]{\mathsf{Tree}}
\newcommand{\TreeSigma}[0]{\mathsf{Tree}_\Sigma}

\newcommand{\Treefin}[0]{\mathsf{Tree}_{\mathrm{fin}}}

\newcommand{\empseq}[0]{\langle\rangle}

\newcommand{\branch}{\mathsf{branch}}

\newcommand{\RunTree}{\mathrm{RunTree}}
\newcommand{\AccTree}{\mathrm{AccTree}}

\newcommand{\myparagraph}[1]{\noindent{\bf #1.\ }}

\newcommand{\supp}{\mathrm{supp}}

\newcommand{\Setsto}[2]{{#2}^{#1}}

\newcommand{\treeprefix}{\preceq}
\newcommand{\treepostfix}{\succeq}

\newcommand{\Ninf}{\mathbb{N}_\infty}

\newcommand{\tpg}{\mathsf{g}}
\newcommand{\pts}{\mathsf{p}}
\newcommand{\btr}{\mathsf{b}}
\newcommand{\tree}{\mathsf{t}}

\newcommand{\nonnegreals}{[0,\infty)}
\newcommand{\nonnegrealsinf}{[0,\infty]}
\newcommand{\unitreals}{[0,1]}

\newcommand{\Ftpg}{F_\tpg}
\newcommand{\Omegatpg}{\{0,1\}}
\newcommand{\sigmatpg}{{\sigma_\tpg}}

\newcommand{\rtpgz}{r_{\tpg,\mathfrak{z}}}
\newcommand{\qtpgz}{q_{\tpg,\mathfrak{z}}}
\newcommand{\Rtpgz}{\Ord_{\leq\mathfrak{z}}}

\newcommand{\Fpts}{F_\pts}

\newcommand{\Omegapts}{\unitreals}
\newcommand{\sigmapts}{{\sigma_\pts}}
\newcommand{\rpts}{r_\pts}
\newcommand{\rptspe}{r'_{\pts,\varepsilon}}
\newcommand{\rptspa}{r'_{\pts\mathsf{m},\alpha}}
\newcommand{\rptswg}{r_{\mathsf{nc},\gamma}}
\newcommand{\qpts}{q_\pts}
\newcommand{\qbtr}{q_\btr}
\newcommand{\qptsp}{q'_{\pts}}
\newcommand{\qptsw}{q_{\mathsf{nc}}}
\newcommand{\ppts}{p_\pts}
\newcommand{\pptsp}{p'_\pts}
\newcommand{\Rpts}{\dist\Ninf}

\newcommand{\Rptsp}{{\nonnegrealsinf}}
\newcommand{\Rptspm}{{[\alpha\delta,\infty]}}
\newcommand{\Rptsw}{{\unitreals}}

\newcommand{\botRtpgz}{\bot_{\Rtpgz}}
\newcommand{\botRpts}{\bot_{\Rpts}}
\newcommand{\topRpts}{\top_{\Rpts}}

\newcommand{\topRptsp}{\top_{\Rptsp}}
\newcommand{\botRptsw}{\bot_{\Rptsw}}

\newcommand{\leqOmegatree}{\mathrel{\leq}}
\newcommand{\Omegatree}{\{0,1\}}
\newcommand{\sigmatree}{{\sigma_\tree}}
\newcommand{\rtree}{r_\tree}

\newcommand{\Rtree}{\Treefin^{\bot}}
\newcommand{\botRtree}{\bot_\tree}
\newcommand{\qtree}{q_\tree}
\newcommand{\FtreeSigma}{F_{\tree,\Sigma}}
\newcommand{\Ftree}{F_{\tree}}
\newcommand{\leqRtree}{\mathrel{\sqsubseteq_{\Rtree}}}

\newcommand{\treecomb}{\mathsf{combine}}

\newcommand{\leqOmegatpg}{\mathrel{\leq}}
\newcommand{\leqOmegapts}{\mathrel{\leq}}
\newcommand{\leqptsw}{\mathrel{\leq}}
\newcommand{\leqRtpgz}{\mathrel{\sqsubseteq_\Ord}}
\newcommand{\leqRpts}{\mathrel{\sqsubseteq_{\Rpts}}}
\newcommand{\geqRpts}{\mathrel{\sqsupseteq_{\Rpts}}}
\newcommand{\leqRptsp}{\mathrel{\sqsubseteq_{\Rptsp}}}
\newcommand{\geqRptsp}{\mathrel{\sqsupseteq_{\Rptsp}}}
\newcommand{\leqRptspm}{\mathrel{\sqsubseteq_{\Rptspm}}}

\newcommand{\leqRptsw}{\mathrel{\leq}}

\newcommand{\bpts}{b}

\newcommand{\bptsw}{b}

\newcommand{\nadd}{\mathrel{\mathrm{\widehat{+}}}}

\newcommand{\proofmark}{\noindent\textit{Proof.}\,\;}

\makeatletter
\newcommand{\mypushright}[1]{\ifmeasuring@#1\else\omit\hfill$\displaystyle#1$\fi\ignorespaces}
\newcommand{\mypushleft}[1]{\ifmeasuring@#1\else\omit$\hspace*{1mm}\displaystyle#1$\hfill\fi\ignorespaces}
\makeatother

\newlength{\currentparindent}

\newif\ifignore 
\ignorefalse
\newcommand{\auxproof}[1]{
\ifignore\mbox{}\newline
\textbf{BEGIN: AUX-PROOF} \dotfill\newline
{#1}\mbox{}\newline
\textbf{END: AUX-PROOF}\dotfill\newline
\fi}

\theoremstyle{plain}
\theoremstyle{definition}
\newtheorem{mydefinition}{Definition}[section]
\theoremstyle{plain}
\newtheorem{mylemma}[mydefinition]{Lemma}
\newtheorem{myproposition}[mydefinition]{Proposition}
\newtheorem{mytheorem}[mydefinition]{Theorem}
\newtheorem{mycorollary}[mydefinition]{Corollary}
\newtheorem{mysublemma}[mydefinition]{Sublemma}
\theoremstyle{definition}
\newtheorem{myremark}[mydefinition]{Remark}

\newtheorem{myexample}[mydefinition]{Example}

\newtheorem{mynotation}[mydefinition]{Notation}
\newtheorem{myassumption}[mydefinition]{Assumption}

\theoremstyle{remark}
\def\myqed{\qed}

\begin{document}
%
\title{Categorical Liveness Checking\\by Corecursive Algebras}

\author{
\IEEEauthorblockN{Natsuki Urabe\IEEEauthorrefmark{1}\quad Masaki Hara
}
\IEEEauthorblockA{
Dept.\ Computer Science,
The University of Tokyo, Japan\\
\{urabenatsuki, qnighy\}@is.s.u-tokyo.ac.jp\\
\IEEEauthorrefmark{1}JSPS Research Fellow}
\and
\IEEEauthorblockN{Ichiro Hasuo
}
\IEEEauthorblockA{National Institute of Informatics, Japan\\
i.hasuo@acm.org}
}



\maketitle

\pagestyle{plain}

\begin{abstract}
\emph{Final coalgebras} as ``categorical greatest fixed points'' play a central role in the theory of coalgebras. Somewhat analogously, most proof methods studied therein have focused on \emph{greatest} fixed-point properties like safety and bisimilarity. Here we make a step towards categorical proof methods for \emph{least} fixed-point properties over dynamical systems modeled as coalgebras.  Concretely, we seek a categorical axiomatization of well-known proof methods for liveness, namely \emph{ranking functions} (in nondeterministic settings) and \emph{ranking supermartingales} (in probabilistic ones). We find an answer in a suitable combination of coalgebraic simulation (studied previously by the authors) and \emph{corecursive algebra} as a classifier for (non-)well-foundedness.
\end{abstract}


\section{Introduction}\label{sec:intro}

\subsection{Backgrounds}
\label{subsec:backgrounds}
Verification of liveness, much like that of safety,
is a prototypical problem that underlines verification of more complex alternating fixed-point specifications.
 Liveness means that something ``good'' eventually occurs,
 while safety means that anything ``bad'' never occurs.

\begin{wrapfigure}[4]{r}{2.6cm}
\vspace{-.65cm}
\small
\begin{equation}\label{eq:introfig1}
\hspace{-1.4cm}
\begin{xy}
(0,0)*=<2mm>[Fo]++!U{x_0} = "x0",
(0,8)*=<2mm>[]++!D{x_1}*{\bigbullet} = "x1",
(8,0)*=<2mm>[Fo]++!U{x_2} = "x2",
(8,8)*=<2mm>[Fo]++!D{x_3} = "x3",
(16,0)*=<2mm>[Fo]++!U{x_4} = "x4",
(16,8)*=<2mm>[Fo]++!D{x_5} = "x5",
\ar @{->} ^{} (-5,0)*+{};"x0"*=<3mm>{}
\ar @{->} ^{} "x0"*+{};"x2"*=<3mm>{}
\ar @{->} ^{} "x2"*+{};"x3"*=<3mm>{}
\ar @{->} ^{} "x3"*+{};"x1"*=<3mm>{}
\ar @{->} ^{} "x1"*+{};"x0"*=<3mm>{}
\ar @{->} ^{} "x5"*+{};"x3"*=<3mm>{}
\ar @{->} ^{} "x4"*+{};"x5"*=<3mm>{}
\ar @{->} ^{} "x2"*+{};"x4"*=<3mm>{}
\end{xy}
\hspace{-.5cm}
\end{equation}
\end{wrapfigure}
\vspace{1mm}
\subsubsection{Ranking Functions}
As an example, 
suppose that we are given a transition system as in the figure (\ref{eq:introfig1}).
Here 
$x_1$ 
is an \emph{accepting} state that represents a good event.
The reachability problem---a typical example of liveness checking problems---asks 
the following:
``Does there exist a path from the initial state $x_0$ to $x_1$?''
%
The answer is yes: $x_1$ is reachable by the path $\to\!x_0\!\to\! x_2\!\to\! x_3\!\to\! x_1$.
Note that 
the path does
 not refer to
the states $x_4$ and $x_5$.

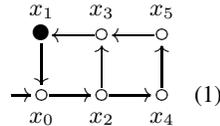
\begin{wrapfigure}[4]{r}{2.8cm}
\vspace{-.65cm}
\small
\begin{equation}\label{eq:introfig2}
\hspace{-1.4cm}
\begin{xy}
(0,0)*=<2mm>[Fo]++!U{x_0} = "x0",
(0,8)*=<2mm>[]++!D{x_1}*{\bigbullet} = "x1",
(8,0)*=<2mm>[F]++!U{x_2} = "x2",
(8,8)*=<2mm>[Fo]++!D{x_3} = "x3",
(16,0)*=<2mm>[Fo]++!U{x_4} = "x4",
(16,8)*=<2mm>[Fo]++!D{x_5} = "x5",
\ar @{->} ^{} (-5,0)*+{};"x0"*=<3mm>{}
\ar @{->} ^{} "x0"*+{};"x2"*=<3mm>{}
\ar @{->} ^{} "x2"*+{};"x3"*=<3mm>{}
\ar @{->} ^{} "x3"*+{};"x1"*=<3mm>{}
\ar @{->} ^{} "x1"*+{};"x0"*=<3mm>{}
\ar @{->} ^{} "x5"*+{};"x3"*=<3mm>{}
\ar @{->} ^{} "x4"*+{};"x5"*=<3mm>{}
\ar @{->} ^{} "x2"*+{};"x4"*=<3mm>{}
\end{xy}
\hspace{-.5cm}
\end{equation}
\end{wrapfigure}
In the example above, we assumed that the system is controlled in an \emph{angelic} manner:
we can choose the next state 
to eventually reach a good state.
However, real-world systems often contain \emph{demonic} branching, too, where the next state is chosen 
to avoid a good state.
Such a system can be modeled as a \emph{two-player game} played by angelic and demonic players.
The figure 
(\ref{eq:introfig2})
illustrates an example.
At the state $x_2$ 
the next move is chosen 
by the demonic player.
The answer to the 
reachability problem
 is again yes:
 no matter if $x_3$ or $x_4$ is chosen as the successor of $x_2$,
the angelic player can force reaching $x_1$ (by $\to x_3\to x_1$ and $\to x_4\to x_5\to x_3\to x_1$).

Numerous methods are known for such liveness checking problems (e.g.\ \cite{schuppanB06livenesschecking,claessenS12livenesschecking,nallaGMBK14effectiveliveness}).
A well-known method is the one using a \emph{ranking function}~\cite{Floyd1967Flowcharts}. 
%
For a two-player game, a ranking function is typically defined as a function
$b:X\to \Ninf$, from the state space $X$ to the set $\Ninf=\mathbb{N}\cup\{\infty\}$,
that satisfies
the following conditions: 
(i) for each non-accepting state $x$ of the angelic player,
there exists a successor state $x'$ such that $b(x)\ge b(x') + 1$;
and (ii) for each non-accepting state $y$ of the demonic player,
we have $b(y)\ge b(y') + 1$
for each successor state $y'$ of $y$.
It is known that \emph{soundness} holds:
existence of a ranking function $b$ such that $b(x)<\infty$ implies that,
regardless of the demonic player's choice, 
the angelic player can 
construct a path $x=x_0\!\to\!x_1\!\to\!\cdots$ that eventually reaches 
an accepting state. 
The well-foundedness of $\mathbb{N}$ is crucial here: 
 we have $b(x_i)>b(x_{i+1})$ for each $i$ before an accepting state is reached; and 
an infinite descending chain is impossible in $\mathbb{N}$. 
For example, in the two-player game in (\ref{eq:introfig2}), the ranking function 
$b=[x_0\mapsto 5, x_1\mapsto 0, x_2\mapsto 4,x_3\mapsto 1,x_4\mapsto 3,x_5\mapsto 2]$ 
 ensures that $x_1$ is reachable from $x_0$.
Intuitively, 
the value $b(x)$ bounds the number of steps from $x$ to 
an accepting state.

%
\begin{wrapfigure}[4]{r}{2.0cm}
\vspace{-.5cm}
\small
\begin{equation}\label{eq:introfig3}
\hspace{-1.4cm}
\begin{xy}
(0,0)*=<2mm>[Fo]++!L{x_0} = "x0",
(0,8)*=<2.5mm>++!L{x_1}*{\bigbullet} = "x1",
%
%
\ar @{->} ^{} (0,-5)*+{};"x0"*=<3mm>{}
\ar @{->} _{\frac{1}{2}} "x0"*+{};"x1"*=<3mm>{}
\ar @{->} @(ul,dl)_(.2){\frac{1}{2}} "x0"*+{};"x0"*=<3mm>{}
\end{xy}
\hspace{-.5cm}
\end{equation}
\end{wrapfigure}
\vspace{1mm}
\subsubsection{Ranking Supermartingales}
One can consider 
 liveness checking problems also for
probabilistic systems. 
A typical example is the \emph{almost-sure reachability problem}:
let us consider
the \emph{probabilistic transition system} (PTS) as in the figure (\ref{eq:introfig3}).
In the almost-sure reachability problem, we want to know if the accepting state $x_1$ is reached with probability $1$.
In the PTS in (\ref{eq:introfig3}), the answer is yes,
though there exists a path that does not visit any accepting state at all 
(namely $\to x_0\to x_0\to\cdots$, but this occurs with probability $0$).

A notion analogous to that of
ranking function is also known for probabilistic systems, 
namely \emph{ranking supermartingales}~\cite{chakarovS13probabilisticprobram,fioritiH15probterm}.
%
For a fixed positive real $\varepsilon>0$, 
an ($\varepsilon$-additive) ranking supermartingale is a function 
$b':X\to \nonnegrealsinf$, from the state space $X$ to the set $\nonnegrealsinf$ of extended non-negative real numbers,
that satisfies the following condition.
\begin{displaymath}
{\textstyle\forall x\in X\setminus \Acc.\;\;b'(x)\,\geq\,\Bigl(\sum_{x'\in X} \mathrm{Prob}(x\!\to\!x')\cdot b'(x')\Bigr)+\varepsilon}
\end{displaymath}
Here $\mathrm{Prob}(x\!\to\!x')$ denotes the probability with which the system makes a transition from $x$ to $x'$.
This means that
for each state $x\in X$, 
the expected value of $b'$ decreases by at least $\varepsilon$ after a transition.
The existence of a ranking supermartingale $b'$ such that $b'(x)<\infty$ implies that 
the expected value of the number of steps from $x$ to an accepting state is finite (specifically it is no bigger than 
${b'(x)/\varepsilon}$).
From this it easily follows that 
an accepting state is visited almost surely.

\vspace{1mm}
\subsubsection{Coalgebras and Algebras}
This paper aims to understand, 
 in the categorical terms of \emph{(co)algebra}, 
essences of liveness checking methods like ranking functions and ranking supermartingales.
Coalgebras are commonly used for modeling state-based dynamics in the categorical
language
(see e.g.~\cite{jacobs16CoalgBook,Rutten00a}).
Formally, 
for an endofunctor $F$ over a category $\mathbb{C}$, an \emph{$F$-coalgebra} is an arrow $c$ of the type $c:X\to FX$.
We can regard $X$ as a state space, $F$ as a specification of the branching type,
and $c:X\to FX$ as a transition function.
By changing the functor $F$  we can represent various kinds of transition types
(see Fig.~\ref{fig:variousFunctors}).
%
%
It is also known that, using  coalgebras, we can generalize various automata-theoretic notions and techniques (such as
behavioral equivalence~\cite{sokolova11probSysSurvey}, bisimulation~\cite{aczelM89finalcoalgebra} and
simulation~\cite{hasuo06genericforward})
to various systems (e.g.\ nondeterministic, probabilistic, and weighted ones).

A dual notion, i.e.\ an arrow of type $a:FX\to X$, is known as an \emph{$F$-algebra}. 
In this paper, it is used to capture \emph{properties} (or \emph{predicates}) over a system  represented as a coalgebra.
\begin{figure}[t] 
\small
\vspace{-.3cm}
\begin{center}
\begin{tabular}{c|c}
functor $F$ & $c:X\to FX$ represents \\ \hline\hline
$(\place)^\Sigma\times \{0,1\}$ & deterministic automaton \\ \hline
$\pow^2(\place)\times\{0,1\}$ & two-player game \\ \hline
$\dist(\place)\times\{0,1\}$ & probabilistic transition system (PTS)
\end{tabular}
\end{center}
\vspace{-2mm}
\caption{Coalgebraic representations of transition systems. Here $\pow$ and $\dist$ denote the powerset and the distribution 
functors respectively (Def.~\ref{def:concreteFunctors}).}
\vspace{-3mm}
\label{fig:variousFunctors}
\end{figure}

\subsection{Contributions}
\label{subsec:contributions}
We contribute 
 a categorical axiomatization of ``ranking functions'' that is behind the well-known methods that we have sketched. It combines: \emph{corecursive algebras} as value domains (that are, like $\Ninf$, suited to detect well-foundedness) and lax homomorphisms (like in \emph{coalgebraic simulations}~\cite{hasuo06genericforward,urabeH15CALCO}).
Based on the axiomatization we develop a general theory; our main result  is \emph{soundness}, i.e.\ that existence of a categorical ranking function indeed witnesses liveness (identified with a least fixed-point property). 
We also exploit our general theory and derive two new notions of ``ranking functions'' as instances. 
The two concrete definitions are new to the best of our knowledge.

We shall now briefly sketch our general theory, illustrating key notions and the backgrounds from which we derive them. 

\vspace{1mm}
\subsubsection{Corecursive Algebras for (Non-)Well-Foundedness} 
In the (conventional) definition of a ranking function $b:X\to \Ninf$, 
well-foundedness of $\mathbb{N}= \Ninf\setminus\{\infty\}$
plays an important role as
it ensures that no path can continue infinitely (without hitting an accepting state).
Similarly, for a ranking supermartingale, it is crucial that $[0,\infty) = [0,\infty]\setminus\{\infty\}$ has no infinite sequence
that decreases everywhere at least by $\varepsilon>0$.
 In unifying the two notions,
we need to categorically capture well-foundedness. 

\begin{wrapfigure}[3]{r}{2cm}
\small
\vspace{-6mm}
\hspace{-2mm}
 \begin{xy}
 \xymatrix@R=1.3em@C=1.6em{
 {F X} \ar@{}[dr]|{=} \ar@{-->}[r]^{F \uniquefp{c}_r} 
 & {F R}  \ar[d]_{r} \\
 {X} \ar[u]_{c}  \ar@{-->}[r]^{\uniquefp{c}_r}  & {R} } 
 \end{xy}
\end{wrapfigure}
Our answer comprises suitable use of  \emph{corecursive
algebras}~\cite{caprettaUV09corecursivealgebra}.  An $F$-algebra $r:FR\to R$ is
said to be corecursive if from an arbitrary coalgebra $c:X\to FX$ there 
exists a unique coalgebra-algebra homomorphism $\uniquefp{c}_r$ (see the
diagram).  Corecursive algebras have been previously used to
describe general structured
corecursion~\cite{caprettaUV09corecursivealgebra} (see also  Rem.~\ref{rem:relatedWorkCorecursiveAlgebra}). 
Our use of them in this paper seems novel:
$r$ being corecursive means that
the function $\Phi_{c,r}\colon f\mapsto r\circ Ff\circ c$ has a unique fixed point; in particular its least and greatest fixed points coincide\footnote{Examples abound in computer science---especially in domain theory---where similar coincidences play important roles. They include: limit-colimit coincidence~\cite{SmythP82} and initial algebra-final coalgebra coincidence~\cite{Freyd91,hasuo07generictrace}.}; we find this feature of corecursive algebras suited for their use as categorical ``classifiers'' for (non-)well-foundedness.


\vspace{1mm}
\subsubsection{Modalities and Least Fixed-Point Properties}\label{subsubsec:introModality}
Liveness properties such as reachability and termination are all instances of \emph{least fixed-point properties}: once a proper \emph{modality} $\heartsuit_{\sigma}$ is fixed, the property in question is described by a least fixed-point formula $\mu u.\, \heartsuit_{\sigma} u$. The way we categorically formulate these constructs, as shown below, is nowadays standard (see e.g.~\cite{Hasuo15CMCSJournVer,HinoKH016}). As the base category $\mathbb{C}$ we use $\Sets$ in this paper (although extensions e.g.\ to $\Meas$ would not be hard). 
\begin{itemize}
 \item We fix a domain $\Omega\in \mathbb{C}$ of \emph{truth values} (e.g.\ $\Omega=\{0,1\}$), and a \emph{property} over  $X\in \mathbb{C}$ is an arrow $u\colon X\to \Omega$. 
 \item A (state-based, dynamical) system is a coalgebra $c\colon X\to FX$ for a suitable functor $F\colon\mathbb{C}\to\mathbb{C}$. 
 \item A modality $\heartsuit_{\sigma}$ is interpreted as an $F$-algebra\footnote{
 We use the same functor $F$ for  
 coalgebras (systems) and algebras (modalities).
 This characterization is used in~\cite{Hasuo15CMCSJournVer,HinoKH016}
 and also found in many coalgebraic modal logic papers (e.g.~\cite{SchroederP09}). 
 However, for some examples, it comes more natural to use  functors $F$ and $G$ together with a natural transformation $\alpha:F\Rightarrow G$, and to
 model a system and a modality as $c:X\to FX$ and $\sigma:G\Omega\to\Omega$ respectively. 
 This modeling induces our current one as 
 $\sigma$ and $\alpha$
 together induce 
 an $F$-algebra $F\Omega\mathrel{\raisebox{-0.7mm}{$\overset{\alpha_\Omega}{\to}$}}G\Omega
 \mathrel{\raisebox{-0.7mm}{$\overset{\sigma}{\to}$}}\Omega$ 
 (cf.\ \S{}\ref{sec:nonTotalRD}).
  }
  $\sigma\colon F\Omega\to \Omega$ over $\Omega$
 (see Example~\ref{example:concreteDefTwoPlayer} and Prop.~\ref{prop:constBehSituationPTS} for examples). 
 \item Assuming some syntax is given, we should be able to derive the interpretation $\sem{\heartsuit_{\sigma}\varphi}_{c}$ of a modal formula from $\sem{\varphi}_{c}$. In the current (purely semantical) framework this goes as follows. 
Given a property $u\colon X\to \Omega$, we define the property $\Phi_{c,\sigma}(u)\colon X\to \Omega$ by the composite
 \begin{displaymath}
\Phi_{c,\sigma}(u) 
 \;=\;
 \bigl(\,
  X\stackrel{c}{\to} FX\stackrel{Fu
}{\to} F\Omega\stackrel{\sigma}{\to}
 \Omega
\,\bigr)\enspace.
 \end{displaymath}
 \item Assuming a suitable order structure $\sqsubseteq_{\Omega}$ on $\Omega$ and additional monotonicity requirements, the correspondence 
   \begin{math}
    \Phi_{c,\sigma}\colon \Omega^{X}\to\Omega^{X}
   \end{math}
   has the least fixed point. It is denoted by $\sem{\mu\sigma}_{c}\colon X\to \Omega$; intuitively it is the interpretation $\sem{\mu u.\, \heartsuit_{\sigma}u}_{c}$ of the  formula $\mu u.\, \heartsuit_{\sigma}u$ in the system $c$. 
\end{itemize}
 Concrete examples are in~\S{}\ref{subsec:sfpSem}. 
Another standard categorical modeling of a modality (see e.g.~\cite{SchroederP09}) is by a \emph{predicate lifting}, i.e.\ a natural transformation $\sigma_{X}\colon
\Omega^{X}\Rightarrow \Omega^{FX}$. It corresponds to our  modeling 
via the Yoneda lemma; see e.g.~\cite{Hasuo15CMCSJournVer}.

 \begin{wrapfigure}[4]{r}{2.2cm}
\small
\vspace{-5mm}
\begin{equation}\label{eq:introLFP}
\vspace{-0mm}
\hspace{-13mm}
 \begin{xy}
 \xymatrix@R=1.6em@C=1.8em{
 {F X} \ar@{}[dr]|{=_\mu} \ar@{->}[r]^{F \llbracket\mu\sigma\rrbracket_c} 
 & {F \Omega}  \ar[d]_{\sigma} \\
 {X} \ar[u]_{c}  \ar@{->}[r]_{\llbracket\mu\sigma\rrbracket_c}
  & {\Omega}} 
 \end{xy}
\hspace{-6mm}
\end{equation}
\end{wrapfigure}
\vspace{1mm}
\subsubsection{Ranking Functions, Categorically}\label{subsubsec:introRankingFuncCat}
Our modeling is summarized  on the right: the liveness property $\sem{\mu\sigma}_{c}$ in question is the \emph{least} arrow (with respect to the order on $\Omega$) that makes the square commute (note the subscript $=_{\mu}$). 

The liveness checking problem is then formulated as follows: given an arrow $h\colon X\to \Omega$, we would like to decide if $h\sqsubseteq_{\Omega} \sem{\mu\sigma}_{c}$ holds. Here $\sqsubseteq_{\Omega}$ denotes the pointwise extension of the order on $\Omega$. For example, let's say we want to check the assertion that a specific state $x_{0}\in X$ satisfies the liveness property $\sem{\mu\sigma}_{c}$.
In this case we would define the above ``assertion'' $h\colon X\to \Omega$, where $\Omega=\{0\sqsubseteq_{\Omega} 1\}$, by: $h(x_{0})=1$ and $h(x)=0$ for all $x\neq x_{0}$. 

\begin{wrapfigure}[5]{r}{3.2cm}
\vspace{-.3cm}
\small
 \begin{xy}
 \xymatrix@R=1.6em@C=1.8em{
 {F X} \ar@{}[dr]|{\rotatebox{90}{\scriptsize$\sqsubseteq$}} \ar[r]_{F b} 
 & {F R} \ar[d]_{r} \ar@{}[dr]|{\sqsubseteq} 
 \ar[r]_{F q}  
 & {F \Omega} \ar[d]_{\sigma} & \\
 {X} \ar[u]_{c}  \ar[r]^{b} 
\ar@{}[rr]_{\raisebox{-.1cm}{\rotatebox{270}{\scriptsize$\sqsupseteq$}}}
 \ar@/_1.8em/[rr]_(.85){
  h 
} 
 & {R} \ar[r]^{q} & {\Omega}  \ar@{}[ru]|{\begin{minipage}{0.1\hsize}
\begin{equation}\label{eq:diagramIntroRD}
\end{equation}
\end{minipage}} & 
 }
 \end{xy}
\end{wrapfigure}
Our categorical framework of  ranking function-based verification goes as follows.
\begin{itemize}
 \item We fix a \emph{ranking domain}---the value domain for ranking functions---to be an algebra $r\colon FR\to R$ together with a lax homomorphism $q\colon R\to \Omega$ (from $r$ to $\sigma$, in the right square in~(\ref{eq:diagramIntroRD})). The latter $q$ identifies  $r\colon FR\stackrel{}{\to} R$ as a ``refinement'' of the modality $\sigma\colon F\Omega\stackrel{}{\to}\Omega$. A crucial  requirement is that $r$ is \emph{corecursive}, making it suited for detecting well-foundedness. 
\end{itemize}
\phantom{hoge}

\vspace{-1.4em}
 \begin{itemize}
  \item  A (categorical) \emph{ranking arrow} for a coalgebra $c\colon X\stackrel{}{\to} FX$ is then defined to be a lax homomorphism $b\colon X\to R$, in the sense shown in the left square in~(\ref{eq:diagramIntroRD}). 
  \item Our soundness theorem says: given an assertion  $h\colon X\to \Omega$,  in order to establish $h\sqsubseteq\sem{\mu\sigma}_{c}$, it suffices to find a ranking arrow $b\colon X\to R$ such that $h\sqsubseteq q\circ b$ (see~(\ref{eq:diagramIntroRD})) .
 \end{itemize}

 This way the problem of verifying a least fixed-point property is reduced to finding a witness $b$. Note that  the requirement on the ranking arrow $b$---namely $b\sqsubseteq r\circ Fb\circ c$---is \emph{local} (it only involves one-step transitions) and hence easy to check. 

Our main technical contribution is a proof for soundness.  There we argue in terms of inequalities between arrows---much like in~\cite{hasuo06genericforward,urabeH15CALCO}---relying on  fundamental order-theoretic results on  fixed points (Knaster--Tarski and Cousot--Cousot, see~\S{}\ref{subsec:lfpgfp}). Corecursiveness of $r$ is crucial there.

\vspace{1mm}
\subsubsection{Concrete Examples}
 Ranking functions for two-player games (\S{}\ref{subsec:backgrounds}) are easily seen to be an instance of our categorical notion, for suitable $F,\Omega$ and $R$.

Our second example in~\S{}\ref{subsec:backgrounds}---additive ranking supermartingales in the probabilistic setting---is itself not an example, however.  Analyzing its reason we are led to a few variations of the definition, among which some seem new. We discuss these variations: their relationship, advantages and disadvantages. 

\subsection{Organization of this Paper}
In \S{}\ref{sec:prelim}
we introduce preliminaries on: our running examples (two-player games and PTSs); liveness checking methods for them; (co)algebras; and least and greatest fixed points. 
Our main contribution  is in \S{}\ref{sec:catTheoryRF}, where our categorical developments are accompanied (for illustration) by  concrete examples from two-player games. The entailments of our general framework in the probabilistic setting are described in~\S{}\ref{subsec:examplePTS}. 
Finally
 in~\S{}\ref{sec:conclusions} we conclude.

Some details and  proofs are deferred to the appendices. 
%

\section{Preliminaries}\label{sec:prelim}

%
In this paper $\nonnegreals$ and $\nonnegrealsinf$ denote
the sets $\{a\in\mathbb{R}\mid a\geq 0\}$ and $\{a\in\mathbb{R}\mid a\geq 0\}\cup\{\infty\}$ respectively.
We extend the ordinary order $\leq$ over $\mathbb{R}$ to 
$\nonnegrealsinf$
by regarding $\infty$ as the greatest element.
We write $\Ninf$ for $\mathbb{N}\cup\{\infty\}$.
For a function $\varphi:X\to[0,1]$, its \emph{support} $\{x\in X\mid \varphi(x)>0\}$ is
denoted by $\supp(\varphi)$.

\subsection{Two-player Games and Ranking Functions}\label{subsec:TPGRF}
Our \emph{two-player games} are
 played by an angelic player $\max$ and a demonic player $\min$.

\begin{mydefinition}[two-player game]
A \emph{(two-player) game structure} is
a triple $G=(X_{\max},X_{\min},\tau)$ 
 of 
a set $X_{\max}$ of states of the player $\max$,
a set $X_{\min}$ of states of the player $\min$,
and 
a transition relation $\tau\subseteq X_{\max}\times X_{\min}\cup X_{\min}\times X_{\max}$.

\fussy
A \emph{strategy of the player $\max$} 
is a partial function $\alpha:X_{\max}\times (X_{\min}\times X_{\max})^*\rightharpoonup X_{\min}$
such that
for each $x_0y_1x_1\ldots \allowbreak y_{n} x_n\in X_{\max}\times (X_{\min}\times X_{\max})^*$, 
if $\alpha(x_0y_1x_1\ldots y_{n} x_n)$ is defined then $(x_n,\alpha(x_0y_1x_1\ldots y_{n} x_n))\in \tau$.
%
A \emph{strategy of the player $\min$}
is a partial function 
$\beta:X_{\max}\times(X_{\min}\times X_{\max})^*\times X_{\min}\rightharpoonup X_{\max}$ 
such that 
for each $x_0y_1x_1\ldots y_{n-1} x_{n-1} y_n\in X_{\max}\times (X_{\min}\times X_{\max})^*\times X_{\min}$, 
if $\beta(x_0y_1x_1\ldots y_{n-1} x_{n-1} y_n)$ is defined then  $(y_n,\beta(x_0y_1x_1\ldots y_{n-1} x_{n-1} y_n))\in \tau$.

For $x\in X_{\max}$ and a pair of strategies $\alpha$ and $\beta$ of the players $\max$ and $\min$, respectively,
the \emph{run induced by $\alpha$ and $\beta$ from $x$} is 
a possibly infinite sequence 
\begin{math}
 \rho^{\alpha,\beta,x}= x_0y_1x_1y_2x_2\ldots
\end{math}
that is an element of the set
\begin{displaymath}
\begin{array}{l}
 (X_{\max}\times (X_{\min}\times X_{\max})^{*}\times\{\bot_{\max}\})\\
\cup\; (X_{\max}\times (X_{\min}\times X_{\max})^{*}\times X_{\min}\times\{\bot_{\min}\})\\
\cup\; (X_{\max}\times X_{\min})^\omega
\end{array}
\end{displaymath}
and is inductively defined as follows:
for $n=0$, $x_0=x$; and 
for $n>0$,
\begin{align*}
y_n&=\small
\begin{cases}
\bot_{\max} & \hspace{-2.0cm}
(\text{$\alpha(x_0y_1x_1\ldots y_{n-1} x_{n-1})$ is undefined}) \\
\alpha(x_0y_1x_1\ldots y_{n-1} x_{n-1}) &  \text{(otherwise), and} 
\end{cases} \\
%
%
x_n &= \small\begin{cases}
\bot_{\min} & \hspace{-1.3cm}
(\text{$\beta(x_0y_1x_1\ldots x_{n-1} y_n)$ is undefined}) \\
\beta(x_0y_1x_1\ldots x_{n-1} y_n) & (\text{otherwise})\,.
\end{cases}
\end{align*}
The symbol $\bot_{\max}$ (resp.\ $\bot_{\min}$) represents the end of the run
at $\max$'s (resp.\ $\min$'s) turn: it means the player 
got stuck.
\end{mydefinition}
%
%
Once an initial state $x$ and strategies $\alpha$ and $\beta$ of the players $\max$ and $\min$ are fixed, 
a run $\rho^{\alpha,\beta,x}$ is determined.
There are different ways to determine the ``winner'' of a run, including:\ $\max$ wins if $\min$ gets stuck; 
$\max$ wins if he does not get stuck;
$\max$ wins if some specified states are visited infinitely many times (the \emph{B\"uchi condition}), etc.
In this paper where we focus on liveness, we choose the following (rather simple) winning condition:
the player $\max$ wins if an \emph{accepting state} is reached or the player $\min$ gets stuck.
Studies of 
more complex conditions  (like the B\"uchi condition) are left as future work.

\begin{mydefinition}[reaching set]\label{def:winnerTPGNew}

Let $G=(X_{\max},X_{\min},\tau)$ 
be a two-player game structure.
We fix a set $\Acc\subseteq X_{\max}$ of \emph{accepting states}.
%
%
%
%
A run 
$
\rho^{\alpha,\beta,x}=
x_0y_1x_1\ldots$ on $(X_{\max},X_{\min},\tau)$ is 
\emph{winning} with respect to $\Acc$ for the player $\max$ if 
\begin{itemize}
\item 
$x_n\in\Acc$ for some $n$; or

\item 
$\rho^{\alpha,\beta,x}$ is a finite sequence whose
last letter is $\bot_{\min}$\,.
\end{itemize}
%
We define the \emph{reaching set} $\Reach_{G,\Acc}\subseteq X_{\max}$ 
by: 
\begin{equation} \label{eq:def:winnerTPG3}
\small
\Reach_{G,\Acc}=
\left\{
x\in X\;\middle|\;
\begin{aligned}
&\exists \alpha:\text{strategy of $\max$}.\;\\
&\quad\forall \beta:\text{strategy of $\min$}.\, \\
& \quad\quad\text{$\rho^{\alpha,\beta,x}$ is winning wrt.\ $\Acc$}
\end{aligned}
\right\}\,.
\end{equation}
\end{mydefinition}

\begin{wrapfigure}[5]{r}{2.3cm}
\vspace{-.5cm}
\scriptsize
\begin{xy}
(0,0)*=<2mm>[Fo]++!L{x_0} = "x0",
(12,0)*=<2mm>[Fo]++!L{x_1} = "x1",
(0,16)*=<2.5mm>++!L{x_2}*{\bigbullet} = "x2",
(9,11)*=<2mm>[Fo]+!DR{x_3} = "x3",
(15,11)*=<2mm>[Fo]++!L{x_4} = "x4",
(-3,8)*=<2mm>[F]++!L{y_0} = "y0",
(3,8)*=<2mm>[F]++!L{y_1} = "y1",
(12,5)*=<2mm>[F]++!L{y_2} = "y2",
(12,16)*=<2mm>[F]++!L{y_3} = "y3",
\ar @{->} ^{} "x0"*+{};"y0"*=<3mm>{}
\ar @{->} ^{} "x0"*+{};"y1"*=<3mm>{}
\ar @{->} @/_1.4mm/^{} "x1"*+{};"y1"*=<3mm>{}
\ar @{->} ^{} "x1"*+{};"y2"*=<3mm>{}
\ar @{->} ^{} "x3"*+{};"y3"*=<3mm>{}
\ar @{->} ^{} "y0"*+{};"x2"*=<3mm>{}
\ar @{->} @/_1.4mm/^{} "y1"*+{};"x1"*=<3mm>{}
\ar @{->} ^{} "y1"*+{};"x2"*=<3mm>{}
\ar @{->} ^{} "y2"*+{};"x3"*=<3mm>{}
\ar @{->} ^{} "y2"*+{};"x4"*=<3mm>{}
\end{xy}
\end{wrapfigure}
\vspace{2mm}
\noindent
\begin{minipage}{0.25\hsize}
\begin{myexample}\label{example:TPGSem}
\end{myexample}
\end{minipage}\!\!
\fussy
We define a game structure $G=(X_{\max},X_{\min},\tau)$ by 
$X_{\max}=\{x_0,x_1,x_2,x_3,x_4\}$, 
$X_{\min}=\{y_0,y_1,y_2,\allowbreak y_3\}$ and
$\tau=\{(x_0,y_0),(x_0,y_1),\allowbreak(x_1,y_1),\allowbreak(x_1,y_2),\allowbreak(x_3,y_3),\allowbreak(y_0,x_2),(y_1,x_1),(y_1,x_2),\allowbreak(y_2,x_3),(y_2,x_4)\}$.
Let $\Acc=\{x_2\}$.
The situation is shown on the right.
%
Then 
 the reaching set is $\Reach_{G,\Acc}=\{x_0,x_2,x_3\}$.\!\!\!
\vspace{1mm}\!

\vspace{1mm}
\myparagraph{Ranking functions}
Suppose that we are given a game structure $G=(X_{\max},X_{\min},\tau)$ and a set $\Acc\subseteq X_{\max}$ of accepting states,
and want to prove that 
a state $x$ is included in $\Reach_{G,\Acc}$. 
%
A \emph{ranking function} 
is a standard proof method in such a setting.
%
There are variations in the definition of ranking function~\cite{Floyd1967Flowcharts,urbanM14abstractdomain,mannaP84adequateproof}:
in this paper we use the following.

\begin{mydefinition}[ranking function]\label{def:rankFunc}
Let $G=(X_{\max},X_{\min},\tau)$ be a game structure and $\Acc\subseteq X_{\max}$.
We fix an ordinal $\mathfrak{z}$ 
and let $\Rtpgz=\{\mathfrak{n}\mid\mathfrak{n}\leq\mathfrak{z}\}$ be the set of ordinals smaller than or equal to $\mathfrak{z}$.
A function $b:X_{\max}\to \Rtpgz$ is called
a \emph{ranking function} (for $G$ and $\Acc$) 
if it satisfies
\begin{displaymath}
{\textstyle
\min_{y\colon (x,y)\in\tau}\sup_{x'\colon (y,x')\in\tau} b(x')\nadd 1\;\leq\; b(x)
}
\end{displaymath}
 for each $x\in X_{\max}\setminus\Acc$,
where $b(x')\nadd 1= 
\min\{ b(x')+1,\mathfrak{z}\}$ denotes addition truncated at $\mathfrak{z}$.
\end{mydefinition}

The following well-known theorem states soundness, i.e.\ 
that a ranking function witnesses reachability.

\begin{mytheorem}[soundness, see e.g.~\cite{Floyd1967Flowcharts}]
\label{thm:soundnessRankConv}
Let $\mathfrak{z}$ be an ordinal, and let
 $b:X\to \Rtpgz$ 
be a ranking  function
for $G$ and $\Acc$.
Then 
$b(x)<\mathfrak{z}$ 
(i.e.\ $b(x)\neq\mathfrak{z}$) implies
$x\in\Reach_{G,\Acc}$.
%
\qed
\end{mytheorem}

\begin{myexample}\label{example:rankFuncConv}
For the game in Example~\ref{example:TPGSem},
let $\mathfrak{z}=\omega$ and 
define a function $b:X\to \Rtpgz$ by 
$b(x_0)=1$, $b(x_2)=0$ and $b(x_1)=b(x_3)=b(x_4)=\omega$.
Then $b$ is a ranking function.
Hence by Thm.~\ref{thm:soundnessRankConv},
we have $x_0\in\Reach_{G,\Acc}$.
\end{myexample}

\emph{Completeness} (the converse of Thm.~\ref{thm:soundnessRankConv}) does not hold.
A counterexample is given later in Example~\ref{example:incompletetpg}.

\begin{myremark}\label{rem:posStratSuffices}
A strategy $\alpha:X_{\max}\times (X_{\min}\times X_{\max})^*\rightharpoonup X_{\min}$
is said to be \emph{positional} if its outcome depends only on the last state of the input, i.e.\
$x_n=x'_{n'}$ implies $\alpha(x_0y_1\ldots \allowbreak y_{n} x_n)=\alpha(x'_0y'_1\ldots \allowbreak y'_{n'} x'_{n'})$\,.
%
It is known that a positional strategy suffices as long as  we  consider reaching sets, i.e.\
the set $\Reach_{G,\Acc}$ in Def.~\ref{def:winnerTPGNew} is unchanged if we replace ``$\exists \alpha:\text{strategy of $\max$}$'' in (\ref{eq:def:winnerTPG3}) with 
``$\exists \alpha:\text{positional strategy of $\max$}$'' (see e.g.\ \cite{gradelW06positionalpriorities}).

A ranking function allows us to synthesize such a positional strategy.
Let $x\in X_{\max}$ and $b:X_{\max}\to\Rtpgz$  be a ranking function s.t.\ $b(x)<\infty$.
We define a strategy $\alpha$ for $\max$ by
\[
\alpha(x_0y_1\ldots y_n x_n)={\textstyle
\mathrm{arg\,min}_{y\colon (x_{n},y)\in\tau}\sup_{x'\colon (y,x')\in\tau} b(x')}\,.
\]
Then it is a positional strategy such that for each strategy $\beta$ of $\min$, 
the run $\rho^{\alpha,\beta,x}$ is winning wrt.\ $\Acc$ for $\max$.
\end{myremark}

\subsection{Probabilistic Transition Systems and Ranking Supermartingales}\label{subsec:PTSRF}

\begin{mydefinition}[PTS]\label{def:PTSNew}
A \emph{probabilistic transition system} (PTS) is
a pair $M=(X,\tau)$
 of a set $X$ and
a transition function $\tau: X\to \dist X$. 
Here
$\dist X=\{d:X\to[0,1]\mid \sum_{x\in X}d(x)= 1\}$
 is the set of  probability distributions over $X$.
%
\end{mydefinition}

\begin{mydefinition}[reachability probability]\label{def:winnerMDPNew}
Let $M=(X,\tau)$ be a PTS.
We fix a set $\Acc\subseteq X$ of \emph{accepting states}.
%
For each $x\in X$ and $n\in\mathbb{N}$, 
we define a value $f_n(x)\in\unitreals$ by:
\begin{align*}
&f_n(x)=\\
&
\sum\left\{
\prod_{i=0}^{k-1}
\tau(x_i)(x_{i+1})
\,\middle|\,
{\small
\begin{aligned}
&k\leq n-1, x_0,\ldots,x_{k}\in X, \\
& x_0=x,  x_i\notin \Acc \text{ ($\forall i\in [0,k-1])$,}\;\\
&\qquad\text{and}\;x_k\in\Acc
\end{aligned}
}
\right\}\,.
\end{align*}
Note that if $x\in\Acc$ then $f_n(x)=1$.
%
As the sequence $\bigl(f_n(x)\bigr)_{n\in\mathbb{N}}$
is increasing for each $x\in X$, 
we can define a function $f:X\to\unitreals$ by: 
\begin{displaymath}
f(x)=
\lim_{n\to\infty}
f_n(x) \,.
\end{displaymath}
%
The function $f$ is called the \emph{reachability probability function} 
with respect to $M$ and $\Acc$, and
is denoted by $\Reach_{M,\Acc}$.
\end{mydefinition}

Here
 the value $f_n(x)\in[0,1]$ is the probability with which an accepting state is reached within $n$ steps from $x$.

\begin{wrapfigure}[6]{r}{2.0cm}
\vspace{-.5cm}
\small
\begin{xy}
(0,2)*=<2mm>[Fo]++!UL{x_0} = "x0",
(-5,9)*=<2mm>[Fo]++!L{x_1} = "x1",
(-5,18)*=<2.5mm>++!R{x_2}*{\bigbullet} = "x2",
(5,9)*=<2mm>[Fo]++!R{x_3} = "x3",
\ar @{->} ^{\frac{1}{2}} "x0"*+{};"x1"*=<3mm>{}
\ar @{->} _{\frac{1}{2}} "x0"*+{};"x3"*=<3mm>{}
\ar @{->} @(ul,dl)_(.2){\frac{1}{2}} "x1"*+{};"x1"*=<3mm>{}
\ar @{->} _{\frac{1}{2}} "x1"*+{};"x2"*=<3mm>{}
\ar @{->} @(ur,dr)^(.5){1} "x2"*+{};"x2"*=<3mm>{}
\ar @{->} @(ur,dr)^(.2){\frac{1}{2}} "x3"*+{};"x3"*=<3mm>{}
\end{xy}
\end{wrapfigure}
\vspace{1mm}
\noindent
\begin{minipage}{0.28\hsize}
\begin{myexample}\label{example:PTSNew}
\end{myexample}
\end{minipage}\!\!
We define a PTS $M=(X,\tau)$ by 
$X=\{x_0,x_1,x_2,x_3\}$, and
$\tau(x_0)=[x_1\mapsto \frac{1}{2},x_3\mapsto \frac{1}{2}]$,
$\tau(x_1)=[x_1\mapsto \frac{1}{2},x_2\mapsto \frac{1}{2}]$,
$\tau(x_2)=[x_2\mapsto 1]$, and
$\tau(x_3)=[x_3\mapsto 1]$.
Let $\Acc=\{x_2\}$.
The situation is as shown on the right.
Then $\Reach_{M,\Acc}:X\to [0,1]$ assigns $\frac{1}{2}$ to $x_0$, $1$ to $x_1$ and $x_2$,  and $0$ to $x_3$.
%
\vspace{1mm}


Let us consider 
the \emph{almost-sure reachability problem} for PTS. 
Given a PTS $M=(X,\tau)$, a set $\Acc\subseteq X$ of accepting states
and an (initial) state $x\in X$,
we want to prove that $\Reach_{M,\Acc}(x)=1$.
For this problem,
a ranking function-like notion
called \emph{ranking supermartingale}~\cite{chakarovS13probabilisticprobram} is known.
There are several variations in
its definition.
We follow the definition in~\cite{fioritiH15probterm};
a variation can be found in~\cite{chakarov2016deductive}.

\begin{mydefinition}[$\varepsilon$-additive ranking supermartingale]\label{def:multAddSMPTS}
Let $M=(X,\tau)$ be a PTS and $\Acc\subseteq X$ be the set of accepting states. 
Let $\varepsilon >0$ be  a real number. 
A function $b':X\to \Rptsp$ is 
an \emph{$\varepsilon$-additive ranking supermartingale (for $M$ and $\Acc$)}  
if 
\begin{displaymath}
{\textstyle
\bigl(\,\sum_{x'\in \supp(\tau(x))} \tau(x)(x')\cdot b'(x')\,\bigr)
+\varepsilon\;\leq\; b'(x)}
\end{displaymath}
 holds for each $x\in X\setminus \Acc$.
\end{mydefinition}
\noindent
Intuitively an $\varepsilon$-additive ranking supermartingale $b'$ bounds the expected number of steps to accepting states: specifically it is no bigger than $b'(x)/\varepsilon$.

\begin{mytheorem}[\cite{fioritiH15probterm}]\label{thm:soundnessRSMConvPTS}
Let $b':X\to \Rptsp$ be an $\varepsilon$-additive ranking supermartingale for $M$ and $\Acc$.
Then $b'(x)<\infty$ implies $\Reach_{M,\Acc}(x)=1$.
\qed
\end{mytheorem}

\begin{myexample}\label{example:RSMConvPTS}
For the PTS in Example~\ref{example:PTSNew},
we define $b':X\to \Rptsp$ by 
$b'(x_0)=\infty$, $b'(x_1)=2$, $b'(x_2)=0$ and $b'(x_3)=\infty$.
Then $b'$ is a 1-additive ranking supermartingale. 
Hence by Thm.~\ref{thm:soundnessRSMConvPTS},
we have $\Reach_{M,\Acc}(x_1)=1$.
\end{myexample}

\subsection{Categorical Preliminaries}\label{subsec:corecAlg}
%
We assume that  readers are familiar with basic categorical notions.
For more details, see e.g.~\cite{mac98categoriesfor,jacobs16CoalgBook}.
%

\begin{mydefinition}[(co)algebra]\label{def:algCoalg}
Let $F:\mathbb{C}\to\mathbb{C}$ be an endofunctor on a category $\mathbb{C}$.
An \emph{$F$-coalgebra} is a pair $(X,c)$ of an object $X$ in $\mathbb{C}$ and an arrow $c$ of the type $c:X\to FX$.
An \emph{$F$-algebra} is a pair $(A,a)$ of an object $A$ in $\mathbb{C}$ and an arrow $a$ of the type $a:FA\to A$.
\end{mydefinition}

In this paper we exclusively use the category $\Sets$ of sets and functions as the base category $\mathbb{C}$ (although extensions e.g.\ to $\Meas$ would not be hard). We would be interested in endofunctors composed by the following.

\begin{mydefinition}[$\pow$, $\dist$ and $(\place)\times C$]\label{def:concreteFunctors}
The \emph{powerset functor}  $\pow:\Sets\to\Sets$ is such that:
\begin{itemize}
\item for each $X\in\Sets$, $\pow X=\{A\subseteq X\}$; and
\item for each $f:X\to Y$ and $A\in\pow X$, $(\pow f)(A)=\{y\in Y\mid \exists x\in A.\, f(x)=y\}$. 
\end{itemize}
The \emph{(discrete) distribution functor}  $\dist:\Sets\to\Sets$  is:
\begin{itemize}
\item for each $X\in\Sets$, $\dist X=\{d:X\to[0,1]\mid \sum_{x\in X}d(x)= 1\}$; and
\item for each $f:X\to Y$, $\delta\in \dist X$ and $y\in Y$, 
$(\dist f)(\delta)(y)=\sum_{x\in f^{-1}(y)}\delta(x)$.
\end{itemize}
For  $C\in\Sets$, the functor $(\place)\times C:\Sets\to\Sets$ is:
\begin{itemize}
\item for  $X\in\Sets$, $X\times C=\{(x,c)\mid x\in X,c\in C\}$; and
\item for each $f:X\to Y$, $x\in X$ and $c\in C$, $(f\times C)(x,c)=(f(x),c)$.
\end{itemize}
\end{mydefinition}
We combine these functors for modeling transition types of various kinds of systems (Fig.~\ref{fig:variousFunctors}). For two-player games we use the functor
 $\Ftpg=\pow^2(\place)\times \{0,1\}$.
It works as follows.
\begin{itemize}
\item For each set $X$, $\Ftpg X=\{\Gamma\subseteq \pow X\}\times \{0,1\}$.
\item For each function $f:X\to Y$, $\Ftpg f:\Ftpg X\to \Ftpg Y$ is defined by $\Ftpg f(\Gamma,t)=(\{\{f(x)\mid x\in A\}\mid A\in\Gamma\},t)$.
\end{itemize}
The correspondence between $\Ftpg$-coalgebras and two-player games will be spelled out in Def.~\ref{def:TPGcoalg}.

The key idea in this paper is to use a \emph{corecursive algebra} as a classifier for (non-)well-foundedness.

\begin{wrapfigure}[3]{r}{2cm}
\small
\vspace{-6mm}
\hspace{-2mm}
 \begin{xy}
 \xymatrix@R=1.3em@C=1.6em{
 {F X} \ar@{}[dr]|{=} \ar@{-->}[r]^{F f} 
 & {F R}  \ar[d]_{r} \\
 {X} \ar[u]_{c}  \ar@{-->}[r]^{f}  & {R} } 
 \end{xy}
\end{wrapfigure}
\vspace{2mm}
\noindent
\begin{minipage}{0.74\hsize}
\begin{mydefinition}[corecursive algebra, \cite{caprettaUV09corecursivealgebra}]\label{def:corecAlg}
\end{mydefinition}
\end{minipage}\!\!
An $F$-algebra $r:FR\to R$ is  \emph{corecursive} if given an arbitrary coalgebra $c:X\to FX$,
there exists a unique arrow $f:X\to R$ such that $f=r\circ Ff\circ c$.

\begin{myremark}\label{rem:relatedWorkCorecursiveAlgebra}

 The connection between corecursive algebras and (non-)well-foundedness
 has been hinted by some existing results. For example, 
 for set functors preserving monos and inverse image diagrams,
 \emph{recursive
 coalgebras}---the categorical dual of corecursive algebras used for
 general structured recursion in~\cite{OSIUS197479}---are known to
 coincide with  \emph{well-founded coalgebras}%
~\cite{taylor99practicalFoundation}, where well-foundedness
 is categorically modeled in terms of ``inductive components.'' 
 For more general categories,
 it is known that if a functor preserves monos then well-foundedness implies recursiveness,
 but its converse does not necessarily hold~\cite{adamekMMS13wpc}.
 The dual
 of this result, between corecursive and \emph{anti-founded} algebras,
 is pursued in~\cite{caprettaUV09corecursivealgebra} but with limited
 success. 
 
 In~\cite{Levy15} the notion of \emph{co-founded part} of an algebra is introduced, with a main theorem that the co-founded part of an injectively structured corecursive algebra carries a final coalgebra. The result is used for characterizing a final coalgebra as that of suitable modal formulas. Despite its name, co-founded parts have little to do with our current view of corecursive algebras here as well-foundedness classifiers.

Discussions on other works on corecursive algebra are found in~\S{}\ref{sec:nonTotalRD}  in the appendix.
\end{myremark}

\subsection{Verification of Least/Greatest Fixed-Point Properties}\label{subsec:lfpgfp}
The following results are fundamental in the studies of fixed-point specifications.
\begin{mytheorem}\label{thm:KTCC}
 Let $L$ be a  complete lattice, 
 and  $f\colon L\to L$ be a monotone
 function. 
\begin{enumerate}
 \item {} (Knaster--Tarski) The set of prefixed points (i.e.\ those $l\in L$ such that $f(l)\sqsubseteq l$) forms a complete lattice. Moreover its least element is (not only a prefixed but) a fixed point, that is, the least fixed point $\mu f$. 
 \item {} (Cousot--Cousot~\cite{cousotC79})
Consider the  (transfinite) sequence
 \begin{math}
  \bot\sqsubseteq f(\bot) \sqsubseteq \cdots \sqsubseteq f^{\mathfrak{a}}(\bot)\sqsubseteq \cdots
 \end{math}
 where, for a limit ordinal $\mathfrak{a}$, we define
 $f^{\mathfrak{a}}(\bot)=\bigsqcup_{\mathfrak{b}<\mathfrak{a}}f^{\mathfrak{b}}(\bot)$. The
 sequence eventually stabilizes and its limit is the least fixed point
 $\mu f$. 
 \myqed
\end{enumerate}
\end{mytheorem}
\noindent
For the greatest fixed point $\nu f$ we have the dual results. From these four results---Knaster--Tarski and Cousot--Cousot, for $\mu$ and $\nu$---we derive the following four ``proof principles.''
\begin{mycorollary}\label{cor:KTCCmunu}
Under the conditions of Thm.~\ref{thm:KTCC}:
 \begin{itemize} \setlength{\itemindent}{1cm}
 \item[(KT$\mu$)] $f(l)\sqsubseteq l$ implies $\mu f\sqsubseteq l$.
 \item[(KT$\nu$)] $l\sqsubseteq f(l)$ implies $l\sqsubseteq \nu f$.
 \item[(CC$\mu$)] $f^{\mathfrak{a}}(\bot)\sqsubseteq \mu f$ for each ordinal $\mathfrak{a}$.
 \item[(CC$\nu$)] $\nu f\sqsubseteq f^{\mathfrak{a}}(\top)$ for each ordinal $\mathfrak{a}$. \myqed
 \end{itemize}
\end{mycorollary}
Among these four, however, only two are applicable in verification: 
our goal is to show that an assertion $h$ is \emph{below} a fixed point (see~\S{}\ref{subsubsec:introRankingFuncCat}); the rules (KT$\nu$) and (CC$\mu$) are for \emph{under-approximation} and thus serve our goal; but the other two are for \emph{over-approximating} the fixed point in question. 

It is these order-theoretic principles behind (namely CC and KT) that cause the difference between the proof methods for liveness (lfp's) and safety (gfp's).
 The role of
 ordinals $\mathfrak{a}$---equivalence classes of well-ordered sets---in (CC$\mu$) can be
 discerned in the definitions of  ranking functions/supermartingales. These proof methods for liveness are in a sharp
 contrast with those for safety, in which
  finding an \emph{invariant} (i.e.\ a post-fixed point $l$ in (KT$\nu$))
 suffices.

The basic idea behind the current contribution---liveness checking by combination of coalgebraic simulation and corecursive algebra---can be laid out as follows. For verification it is convenient if we can rely on \emph{certificates} whose constraints are locally checkable. Their examples include invariants, various notions of (bi)simulation and a general notion of coalgebraic simulation; they are all postfixed points in a suitable sense. They should thus be able to witness only gfp's (not lfp's) in view of 
Cor.~\ref{cor:KTCCmunu}. Here we leverage the lfp-gfp coincidence in corecursive algebras to make coalgebraic simulations witness lfp's too. The lfp-gfp coincidence might seem a serious restriction but it is a common phenomenon in many ``interesting'' structures in computer science (as we discussed at the end of \S{}\ref{subsubsec:introRankingFuncCat})


\section{Categorical Ranking Functions}\label{sec:catTheoryRF}
Here we present our general categorical framework for ranking function-based liveness checking.



\subsection{Running Example: Two-Player Games}\label{subsec:runningexample}
In this section, in order to provide abstract notions with intuitions, we use two-player games (\S{}\ref{subsec:TPGRF}) as 
a running example.
We use the functor $\Ftpg=\pow^2(\place)\times \{0,1\}:\Sets\to\Sets$ to model them as coalgebras (\S{}\ref{subsec:corecAlg}, here $\tpg$ stands for ``game'').

\begin{mydefinition}\label{def:TPGcoalg}
Given a $\Ftpg$-coalgebra $c:X\to \Ftpg X$,
we define 
a game structure $G^c=(X^c_{\max},X^c_{\min},\tau^c)$ and a set $\Acc^c\subseteq X_{\max}$ of 
accepting states as follows: $X^c_{\max}=X$, $X^c_{\min}=\pow X$, 
$\tau^c=\{(x,A)\mid x\in X,  A\in c_1(x)\}\cup \{(A,x')\mid A\in \pow X, x'\in A\}$ and 
$\Acc^c=\{x\in X\mid c_2(x)=1\}$.
Here we write $c(x)=(c_1(x),c_2(x))\in \pow^2 X\times\{0,1\}$ for every $x$.

Conversely, given a game structure $G=(X_{\max},X_{\min},\tau)$ and a set $\Acc\subseteq X_{\max}$,
we define an $\Ftpg$-coalgebra $c^{G,\Acc}:X\to \Ftpg X$ as follows: $X=X_{\max}$ and
$c^{G,\Acc}(x)=(\{\{x'\in X_{\max}\mid (y,x')\in \tau \} \mid y\in X_{\min},(x,y)\in \tau\}, t)$ where $t$ is $1$ if $x\in \Acc$ and $0$ otherwise.
\end{mydefinition}
\noindent
 The  above two transformations 
constitute an embedding-projection pair: games and $\Ftpg$-coalgebra are almost equivalent; the former have additional freedom (in the choice of the set $X_{\min}$) that is however inessential.

Throughout the rest of this section, each categorical notion is accompanied by a concrete example in terms of two-player games.  
For readability, the details of these examples (they are all straightforward) are
 deferred 
to \S{}\ref{subsec:appendixTPG} in the appendix.
The other running example (PTSs) will be discussed later in 
\S{}\ref{subsec:examplePTS}.

\subsection{Modalities and Least Fixed-Point Properties, Categorically}\label{subsec:sfpSem}
Towards a categorical framework in which a soundness theorem is proved on the categorical level of abstraction, we need categorical modeling of modalities and least fixed-point properties. 
Our modeling here follows~\cite{Hasuo15CMCSJournVer,HinoKH016};
it has been sketched in~\S{}\ref{subsubsec:introModality}.


The following function is heavily used in our developments.

\vspace{2mm}
%
\noindent
\begin{minipage}{0.40\hsize}
\begin{mydefinition} [$\Phi_{c,a}$]
\label{def:Phicsigma}
\end{mydefinition}
\end{minipage}\!\!
%
Let $F:\Sets\to\Sets$,
$c:X\to FX$ 
be a coalgebra
and  $a:FA\to A$ be an algebra.
We define a function $\Phi_{c,a}:\Setsto{X}{A}\to\Setsto{X}{A}$ by
 $\Phi_{c,a}(f)=a\circ Ff\circ c$, that is,
\vspace{-2mm}
\begin{displaymath}
\begin{matrix}
\left(
 \begin{xy}
 \xymatrix@R=1.6em@C=3.0em{
 {X}\ar[r]^{f}   & {A} 
 }
 \end{xy}
 \right)
& 
\stackrel{\Phi_{c,a}}{\longmapsto }
& 
\Biggl(\!\!
\raisebox{4mm}{\small
 \begin{xy}
 \xymatrix@R=.8em@C=2.2em{
 {F X} 
 \ar[r]^{F f} 
 & {F A} \ar[d]_{a}  \\
 {X} \ar[u]_{c}   & {A} 
 }
 \end{xy}
 }\!
 \Biggr)
\end{matrix}\enspace.
\end{displaymath}
\noindent
Then corecursiveness (Def.~\ref{def:corecAlg}) is rephrased as follows:
 $r:FR\to R$ is corecursive
if and only if
the function $\Phi_{c,r}$ has a unique fixed point for each 
$c:X\to FX$.

Our categorical modeling of \emph{modality} is as follows.
%
\begin{mydefinition}[a truth-value domain and an $F$-modality]\label{def:modality}
A \emph{truth-value domain} is 
a poset $(\Omega,\sqsubseteq_\Omega)$.
If 
the order 
is clear from the context 
we simply write $\Omega$.
For a functor $F:\Sets\to\Sets$,
an \emph{$F$-modality} over the truth-value domain $\Omega$ is 
an $F$-algebra $\sigma:F\Omega\to\Omega$.
\end{mydefinition}


\begin{myexample}\label{example:concreteDefTwoPlayer}
For
two-player games (i.e.\  $\Ftpg$-coalgebras)
a natural  truth-value domain is given by $(\Omegatpg,\leqOmegatpg)$ where
$1$ stands for ``true.''
On top of this domain a natural
$\Ftpg$-modality $\sigmatpg:\Ftpg\Omegatpg\to\Omegatpg$ is given as follows.
%
\begin{displaymath}
\sigmatpg(\Gamma,t)=
\begin{cases}
1 & (t=1) \\
\max_{A\in\Gamma}\min_{a\in A}a & (\text{otherwise})
\end{cases}
\end{displaymath}
%
Here, in $(\Gamma,t)\in \Ftpg X = \pow^2 X\times \{0,1\}$, $t\in\{0,1\}$ indicates if the current state is accepting or not ($t=1$ if yes). The second case in the above definition of $\sigmatpg(\Gamma,t)$ reflects the intention that, in $\Gamma\in \pow(\pow X)$, the first $\pow$ is for the angelic player $\max$'s choice while the second $\pow$ is for the demonic $\min$'s.
\end{myexample}
%


Using an $F$-modality $\sigma$, liveness is categorically characterized as a least fixed-point property.

\begin{wrapfigure}[4]{r}{2cm}
\small
\vspace{-5mm}
\hspace{-2mm}
 \begin{xy}
 \xymatrix@R=1.6em@C=1.8em{
 {F X} \ar@{}[dr]|{=_\mu} \ar@{->}[r]^{F \llbracket\mu\sigma\rrbracket_c} 
 & {F \Omega}  \ar[d]_{\sigma} \\
 {X} \ar[u]_{c}  \ar@{->}[r]_{\llbracket\mu\sigma\rrbracket_c}
& {\Omega}} 
 \end{xy}
\end{wrapfigure}
\vspace{2mm}
\noindent
\begin{minipage}{0.42\hsize}
\begin{mydefinition}[$\sem{\mu\sigma}_c$]\label{def:lfpsem}
\end{mydefinition}
\end{minipage}\!\!
Let $(\Omega,\sqsubseteq_\Omega)$ be a truth-value domain and $\sigma:F\Omega\to\Omega$ be a modality.
We say 
that $\sigma$  
\emph{has least fixed points}
if for each $c:X\to FX$, the least fixed point 
of $\Phi_{c,\sigma}:\Setsto{X}{\Omega}\to\Setsto{X}{\Omega}$ (Def.~\ref{def:Phicsigma})---with respect to the pointwise extension 
of the order $\sqsubseteq_{\Omega}$---exists.
The least fixed point is called the \emph{(coalgebraic) least fixed-point property} in $c$ specified by $\sigma$, 
and is denoted by $\sem{\mu\sigma}_c:X\to\Omega$.\!
\vspace{2mm}

\begin{myexample}\label{example:concreteDefTwoPlayerSem}
The $\Ftpg$-modality $\sigmatpg:\Ftpg\Omegatpg\to\Omegatpg$ in Example~\ref{example:concreteDefTwoPlayer} has 
least fixed points (this follows from Prop.~\ref{prop:sigmaLFP} later). 
 For each coalgebra $c:X\to\Ftpg X$,
the least fixed-point property $\sem{\mu\sigmatpg}_{c}:X\to\{0,1\}$ is concretely described by:
\begin{displaymath}
\sem{\mu\sigmatpg}_{c}(x)=
\begin{cases}
1 & (\text{$x\in\Reach_{G^c,\Acc^c}$}) \\
0 & (\text{otherwise}).
\end{cases}
\end{displaymath}
Conversely, for each pair of a game structure $G$ and a set $\Acc$,
$\Reach_{G,\Acc}$ is described by $\sem{\mu\sigmatpg}_{c^{G,\Acc}}$.
See Prop.~\ref{prop:concreteDefTwoPlayer} for a proof. This way we characterize reachability in two-player games in categorical terms.
%
%
%
\end{myexample}


\subsection{Ranking Domains and Ranking Arrows}
As we described in~\S{}\ref{subsubsec:introRankingFuncCat}, we understand \emph{liveness checking} as the task of determining if $h\sqsubseteq \sem{\mu\sigma}_{c}$, for a given assertion $h\colon X\to \Omega$. Here we introduce our categorical machinery for providing witnesses to such satisfaction of liveness.


For simplicity of arguments we assume the following.
%
\begin{myassumption}\label{asm:behDomainrank}
Let $F:\Sets\to\Sets$.
We assume that a truth-value domain $(\Omega,\sqsubseteq_\Omega)$ and an $F$-modality
$\sigma:F\Omega\to\Omega$ over $\Omega$ satisfy the following conditions.
%
\begin{enumerate}
%

\item\label{asm:behDomain4and3}
The poset $(\Omega,\sqsubseteq_\Omega)$ is a complete lattice.

%
%
%
%
%

\item\label{asm:behDomain1and2}
For each $F$-coalgebra $c:X\to FX$, the function $\Phi_{c,\sigma}:\Setsto{X}{\Omega}\to\Setsto{X}{\Omega}$ in
Def.~\ref{def:lfpsem} is monotone with respect to the pointwise extension of $\sqsubseteq_{\Omega}$.

\end{enumerate}
\end{myassumption}
%
\noindent
These assumptions are mild. For example, Cond.~\ref{asm:behDomain1and2} is satisfied if: $F\Omega$ has an order structure; $\sigma\colon F\Omega\to \Omega$ is monotone; and the action $F_{X,\Omega}\colon \Omega^X\to (F\Omega)^{FX}$ of $F$ on arrows is monotone, too. 
Cond.~\ref{asm:behDomain4and3} in the above implies that $\Setsto{X}{\Omega}$ is a complete lattice.
%
Thus 
we can  construct a transfinite sequence 
$\bot_\Omega
\sqsubseteq
\Phi_{c,\sigma}(\bot_\Omega) 
\sqsubseteq
  \cdots
\sqsubseteq
  \Phi_{c,\sigma}^{\mathfrak{a}}(\bot_\Omega)
\sqsubseteq
  \cdots$ 
as in Thm.~\ref{thm:KTCC}.2, to
obtain the least fixed point of $\Phi_{c,\sigma}$ as its
limit.

\begin{myproposition}\label{prop:sigmaLFP}
Under the conditions in Asm.~\ref{asm:behDomainrank},
 $\sigma$ has least fixed points (in the sense of Def.~\ref{def:lfpsem}).
\qed
\end{myproposition}


\begin{myexample}\label{example:FsigmaSatisfyAsmTPG}
The data $\Ftpg,\sigmatpg$ for two-player games satisfy the assumptions:
%
see Prop.~\ref{prop:FsigmaSatisfyAsmTPG} (in the appendix) for a proof.
\end{myexample}




We are ready to introduce the key notions.

\vspace{2mm}
\noindent
\begin{minipage}{0.65\hsize}
\begin{mydefinition}[ranking domains]\label{def:rankingdom}
\end{mydefinition}
\end{minipage}\!\!
We assume 
Asm.~\ref{asm:behDomainrank}.
Let $r:FR\to R$ be an $F$-algebra,
$q:R\to \Omega$ 
be an arrow, and
%
$\sqsubseteq_R$ be a partial order on $R$.
Note that for each set $X$, the order $\sqsubseteq_\Omega$ (resp.\ $\sqsubseteq_R$) extends to the one between functions $X\to \Omega$ (resp.\ $X\to R$)
%
\begin{wrapfigure}[4]{r}{4.0cm}
\vspace{-3mm}
\mbox{
\small
 \begin{xy}
 \xymatrix@R=1.6em@C=1.9em{
 {F X} \ar@{}[dr]|{\rotatebox{90}{$\sqsubseteq$}} \ar[r]_{F b} 
 & {F R} \ar[d]_{r} \ar@{}[dr]|{\sqsubseteq} 
 \ar[r]_{F q}  & {F \Omega} \ar[d]_{\sigma} &\\
 {X} \ar[u]_{c}  \ar[r]^{b} 
 & {R} \ar[r]^{q} & {\Omega} 
	\ar@{}[ru]|(.9){\hspace{-5mm}\begin{minipage}{0.1\hsize}
	\begin{equation}\label{eq:diagramRD}
	\end{equation}
	\end{minipage}} 
 &
 }
 \end{xy}
 }
\end{wrapfigure}
 in a pointwise manner.

A triple $(r,q,\sqsubseteq_R)$ is called 
a \emph{ranking domain} for $F$ and $\sigma$ if the following conditions are satisfied.
\begin{enumerate}
\item\label{item:def:rankingdom2}
We have $q\circ r\sqsubseteq_{\Omega} \sigma\circ Fq$ between arrows $FR\to \Omega$
(the square on the right in~(\ref{eq:diagramRD})).

\item\label{item:def:rankingdom4}
The same conditions as in Asm.~\ref{asm:behDomainrank} 
hold for $r$, i.e.\ 
%
\begin{enumerate}


\item\label{item:def:rankingdom40and3}
the poset $(R,\sqsubseteq_R)$ 
is a complete lattice; and

%
%
%
%

\item\label{item:def:rankingdom41and2}
for each $c:X\to FX$, the function $\Phi_{c,r}:\Setsto{X}{R}\to\Setsto{X}{R}$ 
(Def.~\ref{def:Phicsigma}) is monotone.

\end{enumerate}

\item\label{item:def:rankingdom3}
The function $q:R\to \Omega$ is monotone (i.e.\ $a\sqsubseteq_R b\,\Rightarrow q(a)\sqsubseteq_\Omega q(b)$),
strict (i.e.\ $q(\bot_R)=\bot_\Omega$) and
 continuous (i.e.\ for each 
  subset $K\subseteq R$,
we have $q(\bigsqcup_{a\in K} a)=\bigsqcup_{a\in K}q(a)$). 

\item\label{item:def:rankingdom7}
The algebra $r:FR\to R$ is corecursive.
\end{enumerate}
%

\vspace{2mm}


Cond.~\ref{item:def:rankingdom4} in the definition ensures that the least fixed point of
$\Phi_{c,r}$ 
arises from the approximation sequence in Thm.~\ref{thm:KTCC}.2.
Cond.~\ref{item:def:rankingdom3} ensures that
this least fixed point
is preserved by $q$. In particular we insist on strictness---this is much like in domain theory~\cite{Plotkin83PisaNotes}.
The most significant in Def.~\ref{def:rankingdom} 
is the corecursiveness of $r$ (Cond.~\ref{item:def:rankingdom7}): it makes $r$ a refinement of $\sigma$ that is suited for detecting well-foundedness.

\begin{mydefinition}[ranking arrows]\label{def:rankingarrow}
Let $(r,q,\sqsubseteq_R)$ be a ranking domain for $F$ and $\sigma$; and $c:X\to FX$ be a coalgebra. An arrow $b:X\to R$ is called a \emph{(coalgebraic) ranking arrow}
for $c$ with respect to $(r,q,\sqsubseteq_R)$
if it satisfies 
$b\sqsubseteq_{R} \Phi_{c,r}(b)=r\circ Fb\circ c$ (the square  on the left
in~(\ref{eq:diagramRD})).
\end{mydefinition}

Now we give a soundness theorem for  (categorical) ranking arrows. This is the main theorem of this paper;
its proof 
demonstrates the role of the corecursiveness assumption.

\vspace{2mm}
\noindent
\begin{minipage}{.49\hsize}
\begin{mytheorem}[soundness]\label{thm:soundnessranking}
\end{mytheorem}
\end{minipage}\!\!
Let $(r,q,\sqsubseteq_R)$ be a ranking do-
%
\begin{wrapfigure}[7]{r}{2.0cm}
\vspace{-4mm}
\small
 \begin{xy}
 \xymatrix@R=1.8em@C=2.3em{
 {F X} \ar@{}[dr]|{\rotatebox{90}{$\sqsubseteq$}} \ar[r]_{F b} 
 \ar@/^.8em/[rr]^{F\,\llbracket \mu\sigma\rrbracket_c}
 & {F R} \ar[d]_{r} \ar@{}[dr]|{\sqsubseteq} 
 \ar[r]_{F q}  & {F \Omega} \ar[d]_{\sigma} \\
 {X} \ar[u]_{c}  \ar[r]^{b} \ar@/_.8em/[rr]_{\llbracket \mu\sigma\rrbracket_c} & {R} \ar[r]^{q} & {\Omega} 
 }
 \end{xy}
\end{wrapfigure}
main.
Let $c:X\to FX$ be an $F$-coalgebra and $b:X\to R$ be a ranking arrow for $c$ (i.e.\ $b\sqsubseteq r\circ Fb\circ c$).
Then we have: 
\begin{displaymath}
q\circ b\sqsubseteq_{\Omega} \llbracket \mu\sigma\rrbracket_c\,.
\end{displaymath}

Thus for liveness checking (i.e.\ for proving $h\sqsubseteq \sem{\mu\sigma}_{c}$) it suffices to find a ranking arrow $b$ such that $h\sqsubseteq q\circ b$. 
In the proof of the theorem we use the following generalization of Thm.~\ref{thm:KTCC}.2. It starts from a post-fixed point $l$ (not from $\bot$).

\begin{mylemma}\label{lem:cousotCousotExt}
 Assume the conditions in Thm.~\ref{thm:KTCC}, and let $l$ be a post-fixed point of $f$, i.e.\ $l\sqsubseteq f(l)$.
 Then we can define a transfinite sequence
  \begin{math}
  l\sqsubseteq f(l) \sqsubseteq \cdots \sqsubseteq f^{\mathfrak{a}}(l)\sqsubseteq \cdots
 \end{math} 
 in a similar manner to Thm.~\ref{thm:KTCC}.2.
 The sequence eventually stabilizes and its limit is a (not necessarily least) fixed point of $f$.
 \qed
\end{mylemma}

\vspace{1mm}
\noindent\textit{Proof\;(Thm.\;\ref{thm:soundnessranking}).}
%
By Cond.~\ref{item:def:rankingdom40and3}  in Def.~\ref{def:rankingdom}, the poset $(\Setsto{X}{R},\sqsubseteq_{R})$ is 
a complete lattice.
%
Moreover, by its definition, $b:X\to R$ is a post-fixed point of $\Phi_{c,r}$.
Hence together with Cond.~\ref{item:def:rankingdom41and2}, 
we can construct a transfinite sequence 
$b\sqsubseteq_{R} \Phi_{c,r}(b) \sqsubseteq_{R}  \cdots \sqsubseteq_{R}  \Phi_{c,r}^{\mathfrak{a}}(b)\sqsubseteq_{R}  \cdots$ 
as in Lem.~\ref{lem:cousotCousotExt}.
By Lem.~\ref{lem:cousotCousotExt},
there exists an ordinal $\mathfrak{m}$ such that 
$\Phi_{c,r}^\mathfrak{m}(b)$ is a fixed point of $\Phi_{c,r}$.
By its definition, we have $b\sqsubseteq_{R}\Phi_{c,r}^{\mathfrak{m}}(b)$.

Note here that $r$ is assumed to be a corecursive algebra (Cond.~\ref{item:def:rankingdom7} in Def.~\ref{def:rankingdom}).
Hence $\Phi_{c,\sigma}$ has a unique fixed point;
it is denoted by $\uniquefp{c}_r:X\to R$. 
Then we have:
\begin{equation}
b\sqsubseteq_{R}\Phi_{c,r}^{\mathfrak{m}}(b)=\uniquefp{c}_r\,.
\label{eq:thm:soundnessranking1}
\end{equation}

\begin{wrapfigure}[6]{r}{2.0cm}
\vspace{-5mm}
\small
 \begin{xy}
 \xymatrix@R=1.8em@C=2.3em{
 {F X} 
 \ar@/_.6em/[r]_(.4){F b} \ar[r]^(.65){F\uniquefp{c}_r} 
 \ar@/^1.6em/[rr]^(.8){F\,\llbracket \mu\sigma\rrbracket_c}
 & {F R} \ar[d]_{r} \ar@{}[dr]|{\sqsubseteq} 
 \ar[r]_{F q}  & {F \Omega} \ar[d]_{\sigma} \\
 {X} \ar[u]_{c}  \ar@/^.6em/[r]^(.6){b} \ar[r]_(.7){\uniquefp{c}_r}  \ar@/_1.6em/[rr]_(.8){\llbracket \mu\sigma\rrbracket_c} & {R} \ar[r]^{q} & {\Omega} 
 }
 \end{xy}
\end{wrapfigure}
%
By Cond.~\ref{item:def:rankingdom40and3}  in Def.~\ref{def:rankingdom}, $\Setsto{X}{R}$ is a complete lattice.
Hence we can also define $\Phi_{c,r}^\mathfrak{a}(\bot_{R}):X\to R$ for each ordinal $\mathfrak{a}$ (here $\bot_R$ denotes the least element in $\Setsto{X}{R}$),
and by 
Thm.~\ref{thm:KTCC}.2,
there exists $\mathfrak{m}'$ such that $\Phi_{c,r}^{\mathfrak{m}'}(\bot_R)$ is also a fixed point of $\Phi_{c,r}$.
Hence $\Phi_{c,r}^{\mathfrak{m}'}(\bot_R)=\uniquefp{c}_r$.

%
%
%
%
%

%

We shall now prove that $q\circ\Phi_{c,r}^\mathfrak{a}(\bot_{R})\sqsubseteq \llbracket \mu\sigma\rrbracket_c$ holds for each ordinal $\mathfrak{a}$. This is
by transfinite induction on $\mathfrak{a}$. 

For $\mathfrak{a}=0$, we have:
\begin{align*}
q\circ\Phi_{c,r}^\mathfrak{a}(\bot_{R}) 
&= \;\;\quad q\circ \bot_{R} & (\text{by definition}) \\
&=\;\;\quad \bot_{\Omega} & (\text{by Cond.~\ref{item:def:rankingdom3} in Def.~\ref{def:rankingdom}}) \\
&\sqsubseteq_{\Omega} \llbracket \mu\sigma\rrbracket_c\,. 
\end{align*}

For a successor ordinal $\mathfrak{a}+1$, we have:
\begin{align*}
&q\circ\Phi_{c,r}^{\mathfrak{a}+1}(\bot_{R}) \\
& =\;\;\quad q\circ r\circ F(\Phi_{c,r}^\mathfrak{a}(\bot_{R}))\circ c
 & (\text{by definition}) \\
&\sqsubseteq_{\Omega} \sigma\circ  F(q\circ \Phi_{c,r}^\mathfrak{a}(\bot_{R}))\circ c
& (\text{by Cond.~\ref{item:def:rankingdom2} in Def.~\ref{def:rankingdom}}) \\
&\sqsubseteq_{\Omega} \sigma\circ F(\llbracket \mu\sigma\rrbracket_c)\circ c
 & (\text{by IH and Asm.~\ref{asm:behDomainrank}.\ref{asm:behDomain1and2}}) \\
&=\;\;\quad \llbracket \mu\sigma\rrbracket_c\enspace.
& (\text{$\llbracket \mu\sigma\rrbracket_c$ is a fixed point})
\end{align*}

For a limit ordinal $\mathfrak{l}$, we have:
\begin{align*}\textstyle
&q\circ\Phi_{c,r}^{\mathfrak{l}}(\bot_{R}) \\
\textstyle
& \textstyle = q \bigl(\bigsqcup_{\mathfrak{a}<\mathfrak{l}}\Phi_{c,r}^{\mathfrak{a}}(\bot_{R}) \bigr) 
& (\text{by definition})) \\
\textstyle
&= \textstyle \bigsqcup_{\mathfrak{a}<\mathfrak{l}}\bigl(q\circ \Phi_{c,r}^{\mathfrak{a}}(\bot_{R}) \bigr)
& (\text{by Cond.~\ref{item:def:rankingdom3} in Def.~\ref{def:rankingdom}}) \\
\textstyle
&\sqsubseteq_{\Omega}  \llbracket \mu\sigma\rrbracket_c & (\text{by IH})\,.
\end{align*}
As $\Phi_{c,r}^\mathfrak{m'}(\bot_{R})=\uniquefp{c}_r$, 
the last fact yields
\begin{math}
q\circ\uniquefp{c}_r\;\sqsubseteq_{}\; \llbracket \mu\sigma\rrbracket_c
\end{math}. 
Combining with~(\ref{eq:thm:soundnessranking1}) and the monotonicity of $q$, we obtain
\begin{math}
q\circ b\,\sqsubseteq_{}\,q\circ\uniquefp{c}_r\,\sqsubseteq_{}\, \sem{\mu\sigma}_c
\end{math}.
This concludes the proof.
\qed
\vspace{2mm}


\begin{myremark}
Note that the requirement on ranking arrows---$b\sqsubseteq_{R} r\circ Fb\circ c$ (Def.~\ref{def:rankingarrow})---is a \emph{local} one: it only involves one-step transitions by $c$ and hence is easy to check. 

The condition asserts that $b$ is a suitable \emph{post-fixed point}. In view of the order-theoretic foundations in~\S{}\ref{subsec:lfpgfp} this might seem strange: we are using an invariant-like construct $b$ to witness a \emph{least} fixed point, not a greatest. We are allowed to do so thanks to the corecursiveness of $r\colon FR\to R$---here the least and greatest fixed points for $\Phi_{c,r}$ coincide. It is also crucial that $q$ preserves least fixed points, being strict and continuous.
\end{myremark}

%
\begin{myexample}\label{example:rankDomTwoPlayerGame}
For two-player games as $\Ftpg$-coalgebras, we can define a ranking domain
\begin{displaymath}
\Bigl(
\rtpgz:\Ftpg \Rtpgz\to \Rtpgz,\,
\qtpgz:\Rtpgz\to \Omegatpg,\,
\leqRtpgz
\Bigr)
\end{displaymath}
as follows.

%
%
\begin{enumerate}
\item[1)] 
$\Rtpgz=\{\mathfrak{a}\mid \mathfrak{a}\;\text{is an ordinal s.t.}\;\mathfrak{a}\leq\mathfrak{z}\}$, 
and
\begin{displaymath}
\rtpgz(\Gamma,t)=\begin{cases}
0 & (t=1) \\
\min_{A\in\Gamma}\sup_{\mathfrak{a}\in A}(\mathfrak{a}\nadd 1) & (\text{otherwise})\,.
\end{cases}
\end{displaymath}

\item[2)] 
$\qtpgz(\mathfrak{z})=0$, and 
$\qtpgz(\mathfrak{a})=1$ for any $\mathfrak{a}$ such that
 $\mathfrak{a}<\mathfrak{z}$; 

\item[3)] 
$\mathfrak{a}\leqRtpgz\mathfrak{b}\,\defarrow\,\mathfrak{a}\geq\mathfrak{b}$ (note the directions of inequalities).

\end{enumerate}
Recall that $\mathfrak{a}\nadd 1$ denotes $\min\{\mathfrak{a}+1,\mathfrak{z}\}$. 
%
The triple
$(\rtpgz,\qtpgz,\leqRtpgz)$ 
is indeed a ranking domain (see Prop.~\ref{prop:rankDomTwoPlayerGame} in the appendix).
One can think of the above data as a classifier for (non-)well-foundedness: all the ordinals $\mathfrak{a}<\mathfrak{z}$ are for ``well-founded'' and the maximum ordinal $\mathfrak{z}$ is for ``non-well-founded.'' Observe that the map $\qtpgz$ acts accordingly.

We indeed have the following correspondences.
%
%
%
\begin{enumerate}
\renewcommand{\labelenumi}{\alph{enumi})}
\item\label{item:prop:convRankFuncCoalg1Example}
 $b:X\to \Rtpgz$ is a (categorical) ranking arrow (Def.~\ref{def:rankingarrow}) for an $\Ftpg$-coalgebra 
 $c:X\to\Ftpg X$ wrt.\ $(\rtpgz,\qtpgz,\leqRtpgz)$ iff $b$ is a ranking function for $G^c$ and $\Acc^c$ 
 (in the conventional sense of Def.~\ref{def:rankFunc}).
 
 \item\label{item:prop:convRankFuncCoalg3Example}
 $b:X\to \Rtpgz$ is a ranking function for a game structure $G$ and a set $\Acc$ iff
 $b$ is a ranking arrow for $c^{G,\Acc}$ wrt.\ $(\rtpgz,\qtpgz,\leqRtpgz)$.
 
\item\label{item:prop:convRankFuncCoalg2Example}
 $b(x)<\mathfrak{z}$ iff $\qtpgz\circ b(x)=1$.
\end{enumerate}
Here recall the correspondence in Def.~\ref{def:TPGcoalg}.
A formal statement and its proof are found in Prop.~\ref{prop:convRankFuncCoalg}.

Combined with the characterization in Example~\ref{example:concreteDefTwoPlayerSem}, we conclude that the conventional soundness result (Thm.~\ref{thm:soundnessRankConv})  is an instance of our categorical soundness (Thm.~\ref{thm:soundnessranking}). 
\end{myexample}



\begin{myremark}\label{rem:optimalarrow}
Assume the conditions in Thm.~\ref{thm:soundnessranking}.
As $r:FR\to R$ is a corecursive algebra,
there exists a unique arrow $\uniquefp{c}_r:X\to R$ such that $\uniquefp{c}_r=r\circ F\uniquefp{c}_r\circ c$.
Since  $\uniquefp{c}_r$ is obviously the greatest fixed point of $\Phi_{c,\sigma}$, 
by the Knaster--Tarski theorem (cf.\ Thm.~\ref{thm:KTCC}),
for each (categorical) ranking arrow $b:X\to R$ 
we have $b\sqsubseteq_{R}\uniquefp{c}_r$.
 This means that if $h\sqsubseteq_\Omega q\circ b$ then $h\sqsubseteq_\Omega q\circ \uniquefp{c}_r$.
Therefore we can say that the unique arrow $\uniquefp{c}_r:X\to R$ is the ``optimal'' ranking arrow in the sense that
if $\uniquefp{c}_r$ cannot prove liveness then no (categorical) ranking arrow can prove liveness using $R$ and $q$, either.
For two-player games, the optimal ranking arrow is given by the function assigning a state $x$
 the minimum number of steps from $x$ to $\Acc$.
\end{myremark}

We note
that the converse of Thm.~\ref{thm:soundnessranking} (i.e.\ \emph{completeness}) does not necessarily hold.
In other words, it is possible that there exists no ranking arrow $b:X\to R$ such that 
$q\circ b=\llbracket\mu\sigma\rrbracket_c$. 
Here is a counterexample. 


\begin{wrapfigure}[4]{r}{2.2cm}
\vspace{-.4cm}
\small
\begin{xy}
(0,0)*=<1.8mm>[]++!D{x_0}*{\middlebullet} = "x0",
(5,0)*=<1.5mm>[F]++!D{} = "y1",
(10,0)*=<1.5mm>[Fo]++!D{x_1} = "x1",
(15,0)*=<1.5mm>[F]++!D{} = "y2",
(20,0)*=<1.5mm>[Fo]++!D{x_2} = "x2",
(25,0)*=<1.5mm>[F]++!D{} = "y3",
(29,0)*=<1.5mm>[Fo]++!D{x_3} = "x3",
(20,-8)*=<1.5mm>[Fo]++!L{x_\omega} = "yo",
(12,-8)*=<1.5mm>[F]++!L{} = "xo",
%
%
(35,0)*{\normalsize\dots},
%
%
\ar @{->} ^{} "y1"*+{};"x0"*=<3mm>{}
\ar @{->} ^{} "x1"*+{};"y1"*=<3mm>{}
\ar @{->} ^{} "y2"*+{};"x1"*=<3mm>{}
\ar @{->} ^{} "x2"*+{};"y2"*=<3mm>{}
\ar @{->} ^{} "y3"*+{};"x2"*=<3mm>{}
\ar @{->} ^{} "x3"*+{};"y3"*=<3mm>{}
\ar @{->} ^{} (34,0)*+{};"x3"*=<3mm>{}
\ar @{->} ^{} "xo"*+{};"x0"*=<3mm>{}
\ar @{->} ^{} "xo"*+{};"x1"*=<3mm>{}
\ar @{->} ^{} "xo"*+{};"x2"*=<3mm>{}
\ar @{->} ^{} "xo"*+{};"x3"*=<3mm>{}
\ar @{->} ^{} "yo"*+{};"xo"*=<3mm>{}
\end{xy}
\end{wrapfigure}
\vspace{2mm}
\noindent
\begin{minipage}{.3\hsize}
\begin{myexample}\label{example:incompletetpg}
\end{myexample}
\end{minipage}\!\!
We define an $\Ftpg$-coalgebra $c:X\to \Ftpg X$ by
$X=\{x_{\mathfrak{a}}\mid \mathfrak{a}\leq\omega\}$,
$c(x_0)=(\emptyset,1)$ and $c(x_{\mathfrak{a}})=(\{\{x_\mathfrak{b}\mid \mathfrak{b}<\mathfrak{a}\}\},0)$ 
for each $\mathfrak{a}>0$.
Note that in the corresponding game structure $G_c$, all the choices are made by the player $\min$. 
Then we have $\sem{\mu\sigmatpg}_c(x_\omega)=1$ because of well-foundedness of $\omega$.
However, the unique arrow $\uniquefp{c}_{r_{\tpg,\omega}}:X\to
\Ord_{\leq\omega}$ such that 
$\uniquefp{c}_{r_{\tpg,\omega}}=\Phi_{c,\sigmatpg}(\uniquefp{c}_{r_{\tpg,\omega}})$  assigns, to each state $x_{\mathfrak{a}}$, the ordinal $\mathfrak{a}$.
This means that $q_{\tpg,\omega}\circ \uniquefp{c}_{r_{\tpg,\omega}}(x_\mathfrak{\omega})=q_{\tpg,\omega}(\omega)=0$. 
Thus $q_{\tpg,\omega}\circ \uniquefp{c}_{r_{\tpg,\omega}}<\sem{\mu\sigmatpg}_{c}$. 

Similarly, for every ordinal number $\mathfrak{z}$,
we can construct an $\Ftpg$-coalgebra whose reachability cannot be proved by the ranking domain $\rtpgz:\Ftpg\Rtpgz\to\Rtpgz$.
%
\vspace{2mm}

By cardinality arguments we can show that
 sort of ``completeness'' holds in the example above, in the following sense: for every $\Ftpg$-coalgebra $c$
there exists an ordinal $\mathfrak{z}$ such that the reachability of $c$ is provable by the ranking domain $\rtpgz$.
However, in this paper we use the term ``completeness'' in a different sense in which we fix the domain $R$ of ranking functions
in advance.

Here is a categorical sufficient 
 condition for completeness.

\vspace{2mm}
\noindent
\begin{minipage}{\hsize}
\begin{myproposition}[a sufficient condition for completeness]\label{prop:completeness}
\end{myproposition}
\end{minipage}
%
Let $(r,q,\sqsubseteq_{R})$ be a ranking domain,
$c:X\to F X$ be an
$F$-coalgebra, and
$\uniquefp{c}_r:X\to R$ be the unique arrow 
such 
\begin{wrapfigure}[6]{r}{3.8cm}
\small
\vspace{-.3cm}
 \begin{xy}
 \xymatrix@R=1.8em@C=2.4em{
 {F X} \ar@{}[dr]|{=} \ar[r]_{F \uniquefp{c}_r} 
 \ar@/^.6em/[rr]^{F\,\llbracket \mu\sigma\rrbracket_c}
 & {F R} \ar[d]_{r} \ar@{}[dr]|{=} 
 \ar[r]_{F q}  & {F \Omega} \ar[d]_{\sigma} \\
 {X} \ar[u]_{c}  \ar[r]^{\uniquefp{c}_r} \ar@/_.6em/[rr]_{\llbracket \mu\sigma\rrbracket_c} & {R} \ar[r]^{q} & {\Omega} 
 }
 \end{xy}
\end{wrapfigure}
that $\uniquefp{c}_r=r\circ F \uniquefp{c}_r\circ c$.
Assume that we have the equality
\begin{equation} \label{eq:completenessCond}
q\circ r=\sigma\circ Fq\enspace,
\end{equation}
instead of an inequality, in the square on the right.
Then we have $q\circ \uniquefp{c}_r=\sem{\mu\sigma}_c$. 
\qed
\vspace{2mm}
%

\noindent
Intuitively, the equality~(\ref{eq:completenessCond}) means that $r$ approximates the modality $\sigma$ in an adequate way.
The result implies that, in case $h\colon X\to \Omega$ satisfies $h\sqsubseteq \sem{\mu\sigma}_{c}$, the latter inequality can always be witnessed by some ranking arrow (namely $\uniquefp{c}_r$). This is completeness of the proof method of categorical ranking arrows.
An example of a complete ranking domain will be given in \S{}\ref{subsec:distributionValuedSuperMartingale}.



\auxproof{
\section{Categorical Ranking Arrows for Two-player Games}\label{subsec:exampleTPG}
For two-player games, we have already given definitions of a functor, a truth-value domain,
a modality and a ranking domain along the categorical development in the previous section.
We have also shown that the reaching set of a two-player game is captured categorically.
In this section, we 
describe
how the soundness of the conventional notion of ranking function is proved in our framework.
We will also give an example which shows that the ranking domain is not complete.

We first show that the triple in Def.~\ref{def:rankDomTwoPlayerGame} is indeed a ranking domain.

\begin{myexample}\label{example:FsigmaSatisfyAsmTPG}
The truth-value domain $(\Omegatpg,\leqOmegatpg)$ and
the modality $\sigmatpg:\Ftpg\Omegatpg\to\Omegatpg$ in Example~\ref{example:concreteDefTwoPlayer}
satisfy Asm.~\ref{asm:behDomainrank}.
\qed
\end{myexample}

\subsection{Ranking Functions, Categorically}
We fix $\Ftpg$ and $\sigmatpg$ as in 
the previous section.
In this section, we 
describe a existing notion of ranking function (Def.~\ref{def:rankFunc}) in our coalgebraic framework.

\begin{myproposition}\label{prop:FsigmaSatisfyAsmTPG}
The truth-value domain $(\Omegatpg,\leqOmegatpg)$ and
the modality $\sigmatpg:\Ftpg\Omegatpg\to\Omegatpg$ in Def.~\ref{def:concreteDefTwoPlayer}
satisfy Asm.~\ref{asm:behDomainrank}.
\qed
\end{myproposition}

\begin{myproposition}\label{prop:rankDomTwoPlayerGame}
We define $\sigmatpg$ as in Def.~\ref{def:concreteDefTwoPlayer}.
We fix an ordinal $\mathfrak{z}$.
We define 
an $\Ftpg$-algebra $\rtpgz:\Ftpg \Rtpgz\to \Rtpgz$, 
a function $\qtpgz:F\Rtpgz\to \Omegatpg$ and 
a partial order $\leqRtpgz$ over $\Rtpgz$ as follows.
%
\begin{enumerate}
\item[1a)] 
$\Rtpgz=\{\mathfrak{a}\mid \mathfrak{a}\,\text{is an ordinal s.t.}\,\mathfrak{a}\leq\mathfrak{z}\}$. 

\item[1b)] 
$\rtpgz$ is defined by
\begin{displaymath}
\rtpgz(\Gamma,t)=\begin{cases}
0 & (t=1) \\
\min_{A\in\Gamma}\sup_{\mathfrak{a}\in A}(\mathfrak{a}\nadd 1) & (\text{otherwise})
\end{cases}
\end{displaymath}
Here $\mathfrak{a}\nadd 1$ denotes $\min\{\mathfrak{a}+1,\mathfrak{z}\}$.

\item[2)] 
$\qtpgz:\Rtpgz\to\Omegatpg$ is defined by 
\begin{displaymath}
q(\mathfrak{a})=\begin{cases}
0 & (\mathfrak{a}=\mathfrak{z}) \\
1 & (\text{otherwise})\,.
\end{cases}
\end{displaymath}

\item[3)] 
For $\mathfrak{a},\mathfrak{b}\in\Rtpgz$,
$\mathfrak{a}\leqRtpgz\mathfrak{b}\,\defarrow\,\mathfrak{a}\geq\mathfrak{b}$ (note the directions of inequalities).

\end{enumerate}

Then 
$(\rtpgz,\qtpgz,\leqRtpgz)$ 
is a ranking domain.
\end{myproposition}



By Prop.~\ref{prop:rankDomTwoPlayerGame} and soundness of a (categorical) ranking arrow (Thm.~\ref{thm:soundnessranking}),
we have the following. 
\begin{multline}\label{eq:eqTPG}
\exists b.\,\text{$b$ is a (categorical) ranking arrow wrt.\ $(\rtpgz,\qtpgz,\leqRtpgz)$}\;\\\text{and}\; \qtpgz\circ b(x)=1
\\\Rightarrow\; 
\sem{\mu\sigmatpg}_{c^{G,\Acc}}(x)=1
\end{multline}
In contrast, 
the soundness theorem of a (conventional) ranking function (Thm.~\ref{thm:soundnessRankConv}) states 
as follows.
\begin{multline}\label{eq:eqTPG2}
\exists b.\,\text{$b$ is a ranking function}\;\text{and}\; b(x)<\mathfrak{z}\\\Rightarrow\; 
x\in\Reach_{G,\Acc}\,.
\end{multline}
By Prop.~\ref{prop:concreteDefTwoPlayer}, the right-hand sides of (\ref{eq:eqTPG}) and (\ref{eq:eqTPG2}) are equivalent.
The proposition below shows that the left-hand sides 
are equivalent, too.

\begin{myproposition}\label{prop:convRankFuncCoalg}
For the ranking domain $(\rtpgz,\qtpgz,\leqRtpgz)$ in
 Prop.~\ref{prop:rankDomTwoPlayerGame},
we have the followings.
\begin{enumerate}
\item 
 $b$ is a ranking function in the sense of Def.~\ref{def:rankFunc} iff 
 $b$ is a (categorical) ranking arrow in the sense of Def.~\ref{def:rankingarrow}.
\item 
 $b(x)<\mathfrak{z}$ iff $\qtpgz\circ b(x)=1$.
\qed
\end{enumerate}
\end{myproposition}

Therefore we can prove soundness of a (conventional) ranking function
(Thm.~\ref{thm:soundnessRankConv}) using our categorical framework as follows.
%
\begin{align*}
&\exists b.\,\text{$b$ is a ranking function}\;\text{and}\; b(x)<\mathfrak{z} \\ 
&\overset{\text{Prop.~\ref{prop:convRankFuncCoalg}}}{\Leftrightarrow} 
\exists b.\,\text{$b$ is a ranking arrow wrt.\ $(\rtpgz,\qtpgz,\leqRtpgz)$}\;\text{and}\; \\
&\qquad\qquad\qquad \qtpgz\circ b(x)=1 \\
& \overset{\text{Thm.~\ref{thm:soundnessranking}}}{\Rightarrow} \sem{\mu\sigmatpg}_{c^{G,\Acc}} (x)=1 \\
&\overset{\text{Prop.~\ref{prop:concreteDefTwoPlayer}}}{\Leftrightarrow}  x\in\Reach_{G,\Acc}\,.
\end{align*}

To conclude this section,
we give a counterexample that shows that the ranking domain 
$(\rtpgz,\qtpgz,\leqRtpgz)$ does not satisfy completeness.
By Prop.~\ref{prop:completeness}, this means that the inequality in
Cond.~\ref{item:def:rankingdom2} of Def.~\ref{def:rankingdom} is satisfied strictly.

\begin{myexample}\label{example:incompletetpg}
Let $X=\{x_{\mathfrak{a}}\mid \mathfrak{a}\leq\mathfrak{z}\}$. 
We define an $\Ftpg$-coalgebra $c:X\to FX$ by 
$c(x_0)=(\emptyset,1)$ and $c(x_{\mathfrak{a}})=(\{\{x_\mathfrak{a'}\mid \mathfrak{a'}<\mathfrak{a}\}\},0)$ 
for each $\mathfrak{a}>0$.
Note that in the corresponding two-player game $G_c$, all the choices are made by the player $\min$. 
Then we have $\sem{\mu\sigmatpg}_c(x_\mathfrak{z})=1$.
However, the unique arrow $\uniquefp{c}_r:X\to \Rtpgz$ such that 
$\uniquefp{c}_r=\Phi_{c,\sigmatpg}(\uniquefp{c}_r)$  assigns $\mathfrak{z}$ to $x_{\mathfrak{z}}$,
and this means that $q\circ \uniquefp{c}_r(x_\mathfrak{z})=0$.
\end{myexample}
}

\section{Categorical Ranking Arrows for Probabilistic Transition Systems}\label{subsec:examplePTS}
We shall now investigate what our categorical framework in~\S{}\ref{sec:catTheoryRF} entails in the probabilistic setting of~\S{}\ref{subsec:PTSRF}. It turns out that the well-known definition of ranking supermartingale (\emph{$\varepsilon$-additive} ones in Def.~\ref{def:multAddSMPTS}) is \emph{not} an instance. Here we study some variations of the definition of ranking supermartingale; two among them (\emph{distribution-valued} and \emph{non-counting} ones, that are new to our knowledge) exhibit the nice categorical properties in~\S{}\ref{sec:catTheoryRF}. We also discuss some relationships between those variations, showing that the soundness of $\varepsilon$-additive ranking supermartingales (Def.~\ref{def:multAddSMPTS}) can  nevertheless be proved via the categorical arguments in~\S{}\ref{sec:catTheoryRF}.

\subsection{Probabilistic Transition Systems as Coalgebras}\label{subsubsec:PTSasCoalg}
To represent a PTS as a coalgebra,
we use the functor $\Fpts:\Sets\to\Sets$ ($\pts$ stands for ``probability'') defined as follows. 
%
%
\begin{mydefinition}[$\Fpts$]\label{def:PTScoalg}
We let $\Fpts=\dist(\place)\times\{0,1\}$, where $\dist$ is the (discrete) distribution functor in Def.~\ref{def:concreteFunctors}. 
%

For an $\Fpts$-coalgebra $c:X\to \Fpts X$,
we define a PTS $M^c=(X^c,\tau^c)$ and a set $\Acc^c\subseteq X^c$ of accepting states 
by:
$X^c=X$,
$\tau^c(x)=c_1(x)$ for each $x\in X^c$,
 and $\Acc^c=\{x\in X\mid c_2(x)=1\}$.
Here we write $c(x)=(c_1(x),c_2(x))\in\dist X\times\{0,1\}$ for every $x$.

Conversely, for a PTS $M=(X,\tau)$ and a set $\Acc\subseteq X$,
we define an $\Fpts$-coalgebra $c^{M,\Acc}:X\to \Fpts X$ by 
$c^{M,\Acc}(x)=(\tau(x), t)$ where $t$ is $1$ if $x\in \Acc$ and $0$ otherwise.
\end{mydefinition}
%
\noindent
 These correspondences are indeed bijective.
Analogously to Example~\ref{example:concreteDefTwoPlayerSem}, we characterize 
 reachability probabilities of a PTS as a coalgebraic least fixed-point property.

\begin{myproposition}\label{prop:constBehSituationPTS}
We 
define an $\Fpts$-modality $\sigmapts:\Fpts\Omegapts\to\Omegapts$ 
over the truth-value domain $(\Omegapts,\leq)$ 
as follows.
 Here $[0,1]$ is the unit interval and $\leq$ is the usual order on it.
\begin{displaymath}
\sigmapts(\delta,t)\;=\;
\begin{cases}
1 & (t=1) \\
\sum_{a\in\supp(\delta)} a\cdot\delta(a)  & (\text{otherwise})\,.
\end{cases}
\end{displaymath}
Note here that $\delta\in \dist [0,1]$ and $t\in \{0,1\}$.

Then $\sigmapts$ 
satisfies Asm.~\ref{asm:behDomainrank} and thus
has least fixed points (in the sense of Def.~\ref{def:lfpsem}).
 The lfp property $\llbracket\mu\sigmapts\rrbracket_c$ 
coincides with  the 
reachability probability function $\Reach_{M^c,\Acc^c}:X\to \Omegapts$ of the corresponding PTS (Def.~\ref{def:winnerMDPNew}). 
\qed
\end{myproposition}

\subsection{Known Variations: $\varepsilon$-Additive and $\alpha$-Multiplicative Ranking Supermartingales}\label{subsec:additiveAndMultiplicativeSupermartingales}

The definition of ranking supermartingale that we have reviewed 
(\emph{$\varepsilon$-additive} ones in Def.~\ref{def:multAddSMPTS})
is not an instance of our categorical notion (Def.~\ref{def:rankingdom}). 
Specifically, its value domain (the interval $[0,\infty]$ with a suitable $\Fpts$-algebraic structure) fails to be corecursive.
As a result,  soundness of additive ranking supermartingale (Thm.~\ref{thm:soundnessRSMConvPTS}) 
cannot be directly proved using our categorical soundness theorem (Thm.~\ref{thm:soundnessranking}).

The following is an attempt to define a ranking domain for 
$\varepsilon$-additive supermartingales. 
Let us fix a real number $\varepsilon>0$ and
define an $\Fpts$-algebra $\rptspe:\Fpts\Rptsp\to \Rptsp$, an arrow $\qptsp:\Rptsp\to\Omegapts$ and
a partial order $\leqRptsp$ over $\Rptsp$
as follows.
\begin{enumerate}

\item[1)] 
For each $(\psi,t)\in \Fpts\Rptsp=\dist\Rptsp\times \{0,1\}$,
\begin{displaymath}
r'_{\pts,\varepsilon}(\psi,t)
=\begin{cases}
0  & (t=1) \\
\bigl(\,\sum_{a\in\supp(\psi)}a\cdot \psi(a)\, \bigr)+\varepsilon & (\text{otherwise})\,.
\end{cases}
\end{displaymath}

\item[2)] 
$\qptsp(\infty)=0$ and $\qptsp(a)=1$ if $a <\infty$. 

\item[3)] 
$a\leqRptsp b\;\defarrow\; a\geq b$ 
(note the direction). 
\end{enumerate}

\begin{myproposition}\label{prop:rankDomPTSTradD}
In this setting, for each $c:X\to \Fpts X$ and $b':X\to \Rptsp$, we have the following.
\begin{enumerate}
\renewcommand{\labelenumi}{\alph{enumi})}
\item\label{item:prop:convRankFuncCoalg1MPTS}
 $b'$ is an $\varepsilon$-additive ranking supermartingale (Def.~\ref{def:multAddSMPTS}) iff 
 $b'$  
 satisfies
 $b'\sqsubseteq_{\Rptsp}\rptspe\circ \Fpts b'\circ c$.
\item\label{item:prop:convRankFuncCoalg2MPTS}
 $b'(x)\neq\infty$ iff $q'_{\pts}\circ b'(x)=1$.
\qed
\end{enumerate}
\end{myproposition}
Therefore the triple $(\rptspe,\qptsp,\leqRptsp)$ is suited for
accommodating $\varepsilon$-additive supermartingales in our categorical framework. Unfortunately it is not a ranking domain (Def.~\ref{def:rankingdom}).

\begin{myproposition}\label{prop:pseudoRD}
The triple $(\rptspe,\qptsp,\leqRptsp)$ 
introduced above
satisfies the conditions of a ranking domain (Def.~\ref{def:rankingdom}),
except for Cond.~\ref{item:def:rankingdom7}.
\qed
\end{myproposition}

\begin{wrapfigure}[5]{r}{2.3cm}
\vspace{-.2cm}
\small
\begin{xy}
(0,0)*=<2mm>[Fo]++!L{x_0} = "x0",
(-6,6)*=<2mm>[Fo]++!L{x_1} = "x1",
(6,6)*=<2mm>[Fo]++!L{x_2} = "x2",
(0,12)*=<3mm>[]++!L{x_3}*{\bigbullet} = "x3",
\ar @{->} ^{\frac{1}{2}} "x0"*+{};"x1"*=<3mm>{}
\ar @{->} _(.7){\frac{1}{2}} "x0"*+{};"x2"*=<3mm>{}
\ar @{->} @(u,l)_(.4){\frac{1}{2}} "x1"*+{};"x1"*=<3mm>{}
\ar @{->} @(ur,ul)_(.2){1} "x3"*+{};"x3"*=<3mm>{}
\ar @{->} ^{\frac{1}{2}} "x1"*+{};"x3"*=<3mm>{}
\ar @{->} _(.3){1} "x2"*+{};"x3"*=<3mm>{}
\end{xy}
\end{wrapfigure}
\vspace{2mm}
\noindent
\begin{minipage}{0.28\hsize}
\begin{myexample}\label{example:rankSupMartNotCorecNew}
\end{myexample}
\end{minipage}\!\!
We define a coalgebra 
$c:X\to \Fpts X$ by:
$X=\{x_0,x_1,x_2,x_3\}$,
$c(x_0)=([x_1\mapsto \frac{1}{2},x_2\mapsto \frac{1}{2}],0)$,
$c(x_1)=([x_1\mapsto \frac{1}{2},x_3\mapsto \frac{1}{2}],0)$,
$c(x_2)=([x_3\mapsto 1],0)$
and $c(x_3)=([x_3\mapsto 1],1)$.
The corresponding PTS is depicted on the right.

Let $\varepsilon>0$.
For this coalgebra, we define arrows $b_1,b_2:X\to \Rptsp$ by:
$b_1(x_0)=\frac{5}{2}\varepsilon$, $b_1(x_1)=2\varepsilon$, $b_1(x_2)=\varepsilon$ and $b_1(x_3)=0$; and
$b_2(x_0)=\infty$, $b_2(x_1)=\infty$, $b_2(x_2)=\varepsilon$ and $b_2(x_3)=0$.
Both of these qualify as coalgebra-algebra homomorphisms from $c$ to $\rptspe$. Therefore $\rptspe$ is not corecursive. 
%
%
\vspace{1mm}

It is well-known that $\varepsilon$-additive supermartingales (Def.~\ref{def:multAddSMPTS}) witness
\emph{positive almost-sure reachability}~\cite{bournezG05positivealmostsure},
that is, 
the expected number of
steps to accepting states is finite. This is a property strictly stronger than almost-sure reachability (see Example~\ref{example:rankDomExample2}). 
It follows that  $\varepsilon$-additive supermartingales are not complete against
almost-sure reachability.

\begin{wrapfigure}[8]{r}{2.2cm}
\vspace{-.6cm}
\small
\begin{xy}
(0,0)*=<1.5mm>[Fo]++!R{x} = "x0",
(-3,6)*=<1.5mm>[Fo]++!L{} = "x11",
(-3,12)*=<1.8mm>[]++!L{}*{\middlebullet} = "x12",
(3,6)*=<1.5mm>[Fo]++!L{} = "x21",
(3,12)*=<1.5mm>[Fo]++!L{} = "x22",
(3,18)*=<1.5mm>[Fo]++!L{} = "x23",
(3,24)*=<1.8mm>[]++!L{}*{\middlebullet} = "x24",
(15,6)*=<1.5mm>[Fo]++!L{} = "xi1",
(15,12)*=<1.5mm>[Fo]++!L{} = "xi2",
(15,24)*=<1.5mm>[Fo]++!L{} = "xi3",
(15,30)*=<1.8mm>[]++!L{}*{\middlebullet} = "xi4",
(9,13)*{\normalsize\dots},
(22,13)*{\normalsize\dots},
(15,19)*{\vdots},
(18,18)*{\rotatebox{270}{$\overbrace{\hspace{2.2cm}}$}},
(23,18)*{\text{\scriptsize$2^i-1$}},
\ar @{->} ^(.7){\frac{1}{2}} "x0"*+{};"x11"*=<3mm>{}
\ar @{->} _{\frac{1}{4}} "x0"*+{};"x21"*=<3mm>{}
\ar @{->} @/_3mm/_(.8){\frac{1}{2^i}} "x0"*+{};"xi1"*=<3mm>{}
\ar @{->} ^{1} "x11"*+{};"x12"*=<3mm>{}
\ar @{->} ^{1} "x21"*+{};"x22"*=<3mm>{}
\ar @{->} ^{1} "x22"*+{};"x23"*=<3mm>{}
\ar @{->} ^{1} "x23"*+{};"x24"*=<3mm>{}
\ar @{->} ^{1} "xi1"*+{};"xi2"*=<3mm>{}
\ar @{->} ^{} "xi2"*+{};(15,17)*=<3mm>{}
\ar @{->} ^{} (15,19)*+{};"xi3"*=<3mm>{}
\ar @{->} _{1} "xi3"*+{};"xi4"*=<3mm>{}
\ar @{->} @(ur,ul)_(.6){1}  "x12"*+{};"x12"*=<3mm>{}
\ar @{->} @(ul,dl)_(.5){1}  "x24"*+{};"x24"*=<3mm>{}
\ar @{->} @(l,d)_(.5){1}  "xi4"*+{};"xi4"*=<3mm>{}
\end{xy}
\end{wrapfigure}
\vspace{2mm}
\noindent
\begin{minipage}{0.28\hsize}
\begin{myexample}\label{example:rankDomExample2}
\end{myexample}
\end{minipage}\!\!
We define 
$c:X\to \Fpts X$ by 
$X=\{x\}\cup \{x_{i,j}\mid i,j\in\mathbb{N}, 1\leq i, 1\leq j\leq 2^i\}$; and
$c(x)=(\delta,0)$, where $\delta(x_{i,j})$ is $\frac{1}{2^i}$ if $j=0$ and $0$ otherwise, 
$c(x_{i,j})=([x_{i,j+1}\mapsto 1],0)$ for each $i$ and $j<2^i$, and
$c(x_{i,2^i})=([x_{i,2^i}\mapsto 1],1)$.
The corresponding PTS $M^c$ and the set $\Acc^c$ of accepting states 
are shown
on the right. 
This system when run from $x$ is clearly almost-sure terminating; however the expected number of steps to accepting states is infinite.  
\vspace{2mm}

Another known variation of ranking supermartingales is given by
\emph{multiplicative ranking supermartingales}~\cite{chakarov2016deductive}.
%
Let $\alpha\in (0,1)$. 
A function $b:X\to\nonnegrealsinf$ is an \emph{$\alpha$-multiplicative ranking supermartingale} if
we have 
\begin{displaymath}
\forall x\in X\setminus \Acc.\;{\textstyle\sum_{x'\in \supp(\tau(x))} \tau(x)(x')\cdot b(x')
\;\leq\; \alpha\cdot b(x)}\,,
\end{displaymath} 
and moreover there exists $\delta>0$ such that $b(x)\geq\delta$ for each $x\in X\setminus \Acc$. 
For this multiplicative variation, results analogous to Prop.~\ref{prop:rankDomPTSTradD} and 
Prop.~\ref{prop:pseudoRD} hold (see \S{}\ref{subsec:MRSCat} in the appendix for the proofs).
%

\subsection{Distribution-Valued Ranking Functions}\label{subsec:distributionValuedSuperMartingale}
Let us turn to possible instantiations of our categorical framework in the current probabilistic setting. The first uses $\dist\Ninf$, the set of distributions over extended natural numbers, as a ranking domain (instead of $[0,\infty]$). 
In what follows, given a probability distribution $\gamma$ over $\Ninf$ and
$a,b\in\mathbb{N}$, $\gamma([a,b])$ denotes $\sum_{i=a}^{b}\gamma(i)$
and $\gamma([a,\infty))$ denotes $\sum_{i=a}^{\infty}\varphi(i)$.

\begin{myproposition}\label{prop:rankDomPTS}
Recall that $\Ninf=\mathbb{N}\cup\{\infty\}$ and $\dist\Ninf$ collects 
all the distributions over $\Ninf$.
We 
define an $\Fpts$-algebra $\rpts:\Fpts \Rpts\to \Rpts$, a function $\qpts:\Rpts\to\Omegapts$ and
a partial order $\leqRpts$ over $\Rpts$
as follows.
\begin{enumerate}

\item[1)] 
For each $(\Gamma,t)\in \Fpts\Rpts=\dist^2\Ninf\times\{0,1\}$,
\begin{displaymath}
r_{\pts}(\Gamma,t)(a)
=
{\footnotesize
\begin{cases}
1 & \hspace{-6em}(t=1,a=0) \\
0 & \hspace{-6em}(t=1,a> 0\;\text{or}\;t=0,a=0) \\
\sum_{\gamma\in\supp(\Gamma)}\Gamma(\gamma)\cdot \gamma(a-1)  & (t=0, a>0)\,.
\end{cases}}
\end{displaymath}

\item[2)] 
$\qpts(\varphi)=\varphi\bigl([0,\infty)\bigr)$.

\item[3)]\label{item:prop:rankDomMDPNew3}
$\varphi\leqRpts\varphi'\;\defarrow\;\forall a\in\mathbb{N}.\, \varphi([0,a])\leq\varphi'([0,a])$\,.
\end{enumerate}
Then 
the triple $(\rpts,\qpts,\leqRpts)$ is a ranking domain.
Moreover we have $\qpts\circ \rpts=\sigmapts\circ \Fpts\qpts$ (cf.\ (\ref{eq:completenessCond}) in Prop.~\ref{prop:completeness}). 
\qed
\end{myproposition}
\noindent
Intuitively, 
the value $b(x)([0,a])\in[0,1]$ under-approximates the probability 
with which
an accepting state is reached from $x$ within $a$ steps.
%
%
The definition of $\leqRpts$, which is much like in probabilistic powerdomains (see e.g.~\cite{sDjahromi80CPOs}),
also reflects this intuition.
Here
the Dirac distribution
$\delta_0$ (resp.\ $\delta_\infty$) is
 the greatest (resp.\ least) element. 

The definition of ranking arrow (Def.~\ref{def:rankingarrow}) instantiates to the following---as one sees by straightforward calculation---when we fix a ranking domain to be the one in Prop.~\ref{prop:rankDomPTS}.
\begin{mydefinition}[distribution-valued ranking function]\label{def:distRSM}
Let $M=(X,\tau)$ be a PTS and $\Acc\subseteq X$.
A \emph{distribution-valued ranking function} 
is  $b:X\to \Rpts$ such that: 
\begin{displaymath}
\textstyle
\sum_{x'\in \supp(\tau)}\tau(x)(x')\cdot b(x')\bigl([0,a-1]\bigr)  \;\geq\; b(x)\bigl([0,a]\bigr)
\end{displaymath}
for each $x\in X\setminus \Acc$ and $a\in\Ninf$.
Here we let $b(x')([0,-1])=0$.
\end{mydefinition}
\noindent
By Thm.~\ref{thm:soundnessranking} (soundness) we have the following:
given a PTS $c\colon X\to \Fpts X$ and an ``assertion'' $h\colon X\to [0,1]$, if there exists a distribution-valued ranking function $b\colon X\to \dist \Ninf$ such that $h\leq \qpts\circ b$, then we can conclude that $h\le \Reach_{M^c,\Acc^c}$. Here  $\Reach_{M^c,\Acc^c}$ is given by reachability probabilities and coincides with 
 $\sem{\mu\sigmapts}_{c}$  (Prop.~\ref{prop:constBehSituationPTS}). 

We note that \emph{quantitative verification}  is possible using distribution-valued ranking functions. For example an assertion $h\colon X\to [0,1]$ can be such that $h(x)=1/2$; by finding a suitable arrow $b$ we conclude $h\le \Reach_{M^c,\Acc^c}$, that is, that the reachability probability from $x$  is at least $1/2$.\footnote{The problem solved here is more precisely the \emph{threshold reachability checking} problem; and ranking function-based proof methods for the problem would be viable options especially when the state space $X$ is infinite.} Such a quantitative assertion cannot be verified using $\varepsilon$-additive supermartingales: as in Thm.~\ref{thm:soundnessRSMConvPTS} they only witness the \emph{qualitative} property of (positive) almost-sure reachability.

\begin{wrapfigure}[4]{r}{2cm}
\vspace{-.5cm}
\small
\begin{xy}
(0,0)*=<2mm>[Fo]++!LU{x_0} = "x0",
(0,8)*=<2.5mm>++!L{x_1}*{\bigbullet} = "x1",
(8,6)*=<2mm>[Fo]++!D{x_2} = "x2",
\ar @{->} @(ul,dl)_(.5){\frac{1}{3}} "x0"*+{};"x0"*=<3mm>{}
\ar @{->} _{\frac{1}{3}} "x0"*+{};"x1"*=<3mm>{}
\ar @{->} _{\frac{1}{3}} "x0"*+{};"x2"*=<3mm>{}
\ar @{->} @(ul,dl)_(.5){1} "x1"*+{};"x1"*=<3mm>{}
\ar @{->} @(ur,dr)^(.8){1} "x2"*+{};"x2"*=<3mm>{}
\end{xy}
\end{wrapfigure}
\vspace{2mm}
\noindent
\begin{minipage}{0.28\hsize}
\begin{myexample}\label{example:RptsNonAS}
\end{myexample}
\end{minipage}\!\!
We define a PTS $M=(X,\tau)$ by $X=\{x_0,x_1,x_2\}$, 
$\tau(x_0)=[x_0\mapsto \frac{1}{3},x_1\mapsto \frac{1}{3},x_2\mapsto \frac{1}{3}]$,
 $\tau(x_1)=[x_1\mapsto 1]$ and $\tau(x_2)=[x_2\mapsto 1]$.
Let $\Acc=\{x_1\}$.
%
The function $\bpts:X\to \Rpts$ 
defined by 
$\bptsw(x_0)=
[i\mapsto 1/{3^{i+1}},\infty\mapsto 1/2]$,
$\bptsw(x_1)=[0\mapsto 1]$ and $\bptsw(x_2)=[\infty\mapsto 1]$
is a distribution-valued ranking function.
This allows us to conclude
$\Reach_{M,\Acc}(x_0)\geq (1-\frac{1}{2})=\frac{1}{2}$.
\vspace{2mm}

Finally we exhibit \emph{completeness} of distribution-valued ranking functions. This is an immediate consequence of the categorical result (Prop.~\ref{prop:completeness}) and Prop.~\ref{prop:rankDomPTS}.
\begin{myproposition}\label{prop:ptsCompleteness}
For each $\Fpts$-coalgebra $c:X\to \Fpts X$, there exists a (categorical) ranking arrow $b:X\to\Rpts$ such that
$\qpts\circ b=\sem{\mu\sigmapts}_c$.
\qed
\end{myproposition}

\begin{myexample}
In Example~\ref{example:rankDomExample2}, 
the function $b:X\to \Rpts$ defined by 
 \begin{displaymath}
 b(x)(n)=
 {\small\begin{cases}
 1/2^{i} & (n=2^i,i>0) \\
 0 & (\text{otherwise})
 \end{cases}}
 \;\;\text{and}\;\;
 b(x_{i,j})=\delta_{2^i-j}\,
 \end{displaymath}
 is a distribution-valued ranking function 
 (here $\delta_{2^i-j}$ denotes a Dirac distribution).
 We have $b(x)\bigl([0,\infty)\bigr)=1$ and thus successfully verify almost-sure reachability from $x$. 
\end{myexample}

In~\S{}\ref{subsec:additiveAndMultiplicativeSupermartingales} we have argued that $\varepsilon$-additive ranking supermartingales (Def.~\ref{def:multAddSMPTS}) is not an instance of our categorical ranking arrow. It is nevertheless possible to prove its soundness (Thm.~\ref{thm:soundnessRankConv}) using the categorical framework---specifically 
by showing that an $\varepsilon$-additive ranking supermartingale gives rise to a 
distribution-valued ranking function (Def.~\ref{def:distRSM}).
Details are found in~\S{}\ref{subsec:soundARSapp} in the appendix.

Similarly our 
framework can prove soundness of $\alpha$-\linebreak
multiplicative ranking supermartingales 
(see \S{}\ref{subsec:soundMRSapp} for the details).

\subsection{$\gamma$-Scaled Non-Counting Ranking Supermartingales}
 The notion of distribution-valued ranking function exhibits pleasant properties like completeness and quantitative assertion checking. A major drawback, however, is the complexity of its value domain $\dist\Ninf$. 

In many realistic verification scenarios a ranking function/supermartingale $b$ would be \emph{synthesized} as follows: the function $b$ is expressed in a predetermined template $b_{\vec{p}}$  (such as polynomials up-to a certain degree)
in which some parameters $\vec{p}$ occur; the requirements on $b_{\vec{p}}$ translate to constraints on $\vec{p}$; and one relies on some optimization solver (for SAT, LP, SDP, etc.) to solve the constraints. It significantly increases the complexity of the workflow
if the value $b(x)$ is a distribution in $\dist\Ninf$ instead of an (extended) real number in $[0,\infty]$.

Here we present another probabilistic instantiation of the categorical framework. It takes values in the unit interval $\Rptsw$.

\begin{myproposition}\label{prop:weakrankDomPTS}
We fix a real number $\gamma\in [0,1)$.
We define an algebra $\rptswg:\Fpts \Rptsw\to \Rptsw$ (here $\mathsf{nc}$ stands for ``non-counting'') as follows.
\begin{displaymath}\small
\rptswg(\varphi,t)=\begin{cases}
1 & (t=1) \\
\gamma\cdot \sum_{a\in \supp(\varphi)} a\cdot\varphi(a) & (\text{otherwise})\,.
\end{cases}
\end{displaymath}
We further define $\qptsw:\Rptsw\to \Omegapts$ by $\qptsw(a)=a$.
Then $(\rptswg,\qptsw,\leqptsw)$, where $\leqptsw$ is the usual on $[0,1]$, is a ranking domain
with respect to the modality $\sigmapts$ (Prop.~\ref{prop:constBehSituationPTS}). 
\qed
\end{myproposition}

Thus it makes sense to consider 
ranking arrows (Def.~\ref{def:rankingarrow}) 
with respect to the ranking domain  $(\rptswg,\qptsw,\leqptsw)$. 
By straightforward calculation, their definition unravels as follows.

\begin{mydefinition}[$\gamma$-scaled non-counting ranking supermartingale]\label{def:weakRSM}
Let $M=(X,\tau)$ be a PTS and $\Acc\subseteq X$.
We fix a real number $\gamma$ such that $0\leq\gamma<1$.
A \emph{$\gamma$-scaled non-counting ranking supermartingale} 
is a function $\bptsw:X\to \Rptsw$ such that 
for each $x\in X\setminus \Acc$ we have:
\begin{displaymath}
\textstyle\gamma\cdot\sum_{x'\in X} \tau(x)(x')\cdot \bptsw(x')\;\geq\;  \bptsw(x)\enspace.
\end{displaymath}
\end{mydefinition}
\noindent
This notion of supermartingale seems new. The intuition is that $b(x)$ is a lower bound for the reachability probability. Note that, unlike for the other variations of supermartingales in this section (where we \emph{over}-approximate the number of steps),  reachability probabilities should be \emph{under}-approximated. 

We obtain the following soundness result as an instance of Thm.~\ref{thm:soundnessranking}. Here a non-counting ranking supermartingale  $b$  itself gives lower bounds for reachability probabilities since $\qptsw\colon \Rptsw\to\Omegapts$ is the identity map.




\begin{mycorollary}\label{cor:rankDomPTSconcrete}
Let $M=(X,\tau)$ be a PTS and 
$\bptsw:X\to \Rptsw$ be a $\gamma$-scaled non-counting ranking supermartingale. 
Then we have $\Reach_\Acc(x)\geq  \bptsw(x)$. 
\qed
\end{mycorollary}

\begin{myexample}\label{example:MDPhalfProb}
Consider the PTS $M$ and $\Acc$ in Example~\ref{example:RptsNonAS}. 
%
 For each $\gamma\in [0,1)$, 
we define $\bptsw_{\gamma}:X\to \Rptsw$ 
 by 
$\bptsw_{\gamma}(x_0)=\frac{\gamma}{3-\gamma}$, $\bptsw_{\gamma}(x_1)=1$ and $\bptsw_{\gamma}(x_2)=0$. Then 
$\bptsw_{\gamma}$ is seen to be a $\gamma$-scaled non-counting ranking supermartingale.
 By Cor.~\ref{cor:rankDomPTSconcrete} we have 
$\Reach_\Acc(x_0)\geq \frac{\gamma}{3-\gamma}$; 
as this holds for any  $\gamma\in [0,1)$,  by letting $\gamma\to 1$,
we conclude 
$\Reach_\Acc(x_0)\geq \frac{1}{2}$.
\end{myexample}
\noindent
We note that in general a scaling factor $\gamma\in[0,1)$ results in suboptimality of under-approximation of reachability properties: in the last example $\bptsw_{\gamma}(x_{0})=\frac{\gamma}{3-\gamma}$ is smaller than the reachability probability $1/2$. 
Such suboptimality is an issue especially when we aim at \emph{qualitative} verification of almost-sure reachability.
In the last example we exercised an \emph{asymptotic} argument in which we think of $\gamma$  as a free variable and take the limit under $\gamma\to 1$. This strategy can be employed for almost-sure reachability checking.

Finally, the following example demonstrates that application of non-counting supermartingales is not limited to \emph{positive} almost-sure reachability. 
 \begin{myexample}\label{example:rankDomExample3}
We define a PTS $M=(X,\tau)$ and $\Acc\subseteq X$  
as in Example~\ref{example:rankDomExample2}.
For each $\gamma$, if we define $\bptsw_{\gamma}:X\to \Rptsw$ by 
$b_{\gamma}(x)=\sum_{i=1}^{\infty}\frac{\gamma^{2^i}}{2^i}$ and 
$b_{\gamma}(x_{i,j})=\gamma^{2^i-j}$ for each $i$ and $j$, then it is a $\gamma$-scaled non-counting supermartingale.
Hence we have $\Reach_{M,\Acc}(x)\geq\lim_{\gamma\to 1}b_{\gamma}(x)=1$.
 \end{myexample}




Let us summarize the section. We presented four variations of ranking supermartingales: \emph{$\varepsilon$-additive} ones, \emph{$\alpha$-multiplicative} ones, \emph{distribution-valued} ranking functions and $\gamma$-scaled \emph{non-counting} ranking supermartingales. The former two are known in the literature while the latter two seem to be new. The known notions are not instances of our generic definition, but 
their soundness
can be derived via our generic theory (see the end of \S{}\ref{subsec:distributionValuedSuperMartingale}).  Among the (seemingly) new notions, distribution-valued ranking functions enjoy nice properties like completeness (Prop.~\ref{prop:ptsCompleteness}) and support of quantitative reasoning (see Example~\ref{example:RptsNonAS})---at the cost of their complexity 
(they take as values distributions in $\dist\Ninf$). 
Non-counting ranking supermartingales (whose values are simply real numbers) are advantageous in quantitative reasoning
(see Example~\ref{example:MDPhalfProb}) and non-positive almost-sure termination, but the scaling factor $\gamma$ in it leads to suboptimal approximation. Typically one needs to rely on asymptotic arguments to obtain sharp bounds.

\section{Conclusions and Future Work}\label{sec:conclusions}
We have given a categorical account for liveness checking: we identify the essence of ranking function-based proof methods as the combination of corecursive algebras (as value domains) and lax homomorphisms; and for our notion of ranking arrow a soundness theorem has been presented. Our leading examples have been two-player games and probabilistic transition systems; in the course of studying them we were led to (seemingly) new variations of ranking martingales. 

Besides the concrete examples of ``ranking functions'' in this paper, we wish to derive yet other concrete examples from our categorical modeling, so that they provide novel proof methods for various liveness properties. 
A possible direction towards this goal is discussed in \S{}\ref{sec:nonTotalRD} in the appendix,
motivated by categorical closure properties of corecursive algebras.
%
Some abstract categorical questions remain open, too,
 such as characterization of the \emph{non-well-founded part}---that represents failure of liveness properties---of a corecursive algebra. 
We are also interested in 
the relationship between these (expected) results/examples and \emph{productivity} 
for coinductive datatypes 
in functional programming~\cite{AbelP16}. 
Intuitively, the latter is a property that any finite prefix of a coinductively defined data is obtained in finite time.



We have used two-player games (systems with angelic and demonic transitions) and 
PTSs (systems with probabilistic transitions) as leading examples.
A natural direction of future work is to consider stochastic games,
which involve angelic, demonic and probabilistic transitions~\cite{Shapley01101953}. 

In this paper we have focused on least fixed-point properties. Extension to  nested fixed-point specifications---\emph{persistence}, \emph{recurrence}, and general fixed-point formulas---is important future work. There we will need to categorically axiomatize \emph{progress measures} for  parity games (\cite{Jurdzinski00}; see also~\cite{HasuoSC16}).  Possibly relevant to this direction is our recent coalgebraic modeling of B\"uchi and parity acceptance conditions~\cite{urabeSH16coalgebraictrace}.

%



Practical implications of the proposed framework (and concrete ``ranking functions'' derived thereby) shall be investigated, too. We are especially interested in cyber-physical applications in which state spaces are inherently infinite but often allow succinct symbolic presentations (e.g.\ by polynomials). The work closely related to this direction is~\cite{chakarov2016deductive}. 

A categorical account on martingales is also found in recent~\cite{kozen16kolmogorovExtension}, where  a connection between two classic results---Kolmogorov's extension theorem and 
Doob's martingale convergence theorem---is established in categorical terms. The relationship between this work and ours shall be pursued, possibly centered around the notion of final sequence. 

\section*{Acknowledgment}
Thanks are due to Eugenia Sironi and anonymous referees for their useful comments.
The authors are supported by ERATO HASUO Metamathematics for Systems Design Project (No.\
JPMJER1603), JST, and Grants-in-Aid No.\ 15KT0012 \& 15K11984, JSPS. Natsuki Urabe is supported by Grant-in-Aid for JSPS Fellows (No.\ 16J08157).






\clearpage
\appendix


\subsection{\textbf{Formal Discussions for Two-Player Games}}\label{subsec:appendixTPG}
\begin{mydefinition}[$\sigmatpg$]\label{def:concreteDefTwoPlayer}
Let $(\Omegatpg,\leqOmegatpg)$ be a truth-value domain where 
$\leqOmegatpg$ denotes the usual order.
We define an $\Ftpg$-modality $\sigmatpg:\Ftpg \Omegatpg\to\Omegatpg$ by
\begin{displaymath}
\sigmatpg(\Gamma,t)=
\begin{cases}
1 & (t=1) \\
\max_{A\in\Gamma}\min_{a\in A}a & (\text{otherwise})\,.
\end{cases}
\end{displaymath}
\end{mydefinition}

\begin{myproposition}\label{prop:concreteDefTwoPlayer}
We define 
an $\Ftpg$-modality $\sigmatpg:\Ftpg\Omegatpg\to\Omegatpg$ as in Def.~\ref{def:concreteDefTwoPlayer}.
\begin{enumerate}
\item\label{item:prop:concreteDefTwoPlayer1}
The modality $\sigmatpg$ 
has least fixed points,
and for each coalgebra $c:X\to \Ftpg X$, the corresponding least fixed-point property 
$\sem{\mu\sigmatpg}_{c}:X\to\Omegatpg$
is given as follows:
%
\begin{displaymath}
\sem{\mu\sigmatpg}_{c}(x)=
\begin{cases}
1 & (\text{$x\in\Reach_{G^c,\Acc^c}$}) \\
0 & (\text{otherwise})\,.
\end{cases}
\end{displaymath}

\item\label{item:prop:concreteDefTwoPlayer2}
For a game structure $G=(X_{\max},X_{\min},\tau)$ and a set $\Acc\subseteq X_{\max}$,
we have:
\begin{displaymath}
\sem{\mu\sigmatpg}_{c^{G,\Acc}}(x)=
\begin{cases}
1 & (\text{$x\in\Reach_{G,\Acc}$}) \\
0 & (\text{otherwise})\,.
\end{cases}
\end{displaymath}
\end{enumerate}
\end{myproposition}

\begin{proof}\mbox{}
We first prove (\ref{item:prop:concreteDefTwoPlayer1}).
We define  $f:X\to\Omegatpg$ by 
\begin{displaymath}
f(x)=
\begin{cases}
1 & (\text{$x\in\Reach_{G^c,\Acc^c}$}) \\
0 & (\text{otherwise})\,.
\end{cases}
\end{displaymath}
By definition of the least fixed-point property,
it suffices to show that $f$ is the least fixed point of 
the function
$\Phi_{c,\sigmatpg}:\Setsto{X}{\Omegatpg}\to\Setsto{X}{\Omegatpg}$ 
in Def.~\ref{def:Phicsigma}.

We first show that $f$ is a fixed point of $\Phi_{c,\sigmatpg}$.
For each $x\in X$, we have:
\allowdisplaybreaks[4]
\begin{align*}
&\Phi_{c,\sigmatpg}(f)(x)=1 \\
&\Leftrightarrow (\sigmatpg\circ \Ftpg f\circ c)(x)=1 & (\text{by def.\ of $\Phi_{c,\sigmatpg}$}) \\
&\Leftrightarrow c_2(x)=1,\;\text{or}\;\exists A\in c_1(x).\, \forall x'\in A.\, f(x')=1 \hspace{-4cm}& \\
& & (\text{by def.\ of $\sigmatpg$ and $\Ftpg$}) \\
&\Leftrightarrow c_2(x)=1,\;\text{or}\; \\
& \qquad\quad\exists A\in c_1(x).\; \forall x'\in A.\; \\ 
&\qquad\qquad \exists \alpha':\text{a strategy of $\max$}.\; \forall \beta':\text{a strategy of $\min$}.\, \hspace{-4cm}&\\
&\qquad\qquad\quad \text{$\rho^{\alpha',\beta',x'}$ is winning for $\max$} \hspace{-2cm}& (\text{by def.\ of $f$}) \\
&\Leftrightarrow c_2(x)=1,\;\text{or}\; \\
&\qquad \quad\exists A\in c_1(x).\;\\
& \qquad\quad\quad\exists (\alpha'_{x'})_{x'\in A}:\text{a family of strategies of $\max$}.\;  \hspace{-4.5cm} \\ 
&\qquad\qquad\quad \forall x'\in A.\; \;\forall \beta':\text{a strategy of $\min$}.\; \hspace{-4.5cm}&\\
&\qquad\qquad\qquad \text{$\rho^{\alpha'_{x'},\beta',x'}$ is winning for $\max$} \hspace{-2cm}&  \\
&\Leftrightarrow x\in \Acc^c,\;\text{or}\; \\
& \quad\exists y\in X^c_{\min} \;\text{s.t.}\;(x,y)\in \tau^c.\;\\
& \quad\;\exists (\alpha'_{x'})_{x'\in X^c_{\max}\;\text{s.t.}\;(y,x')\in\tau^c}:\text{a family of strategies of $\max$}.\;  \hspace{-4.5cm} \\ 
&\qquad\quad \forall x'\in X^c_{\max}\;\text{s.t.}\;(y,x')\in\tau^c.\; \;\forall \beta':\text{a strategy of $\min$}.\; \hspace{-4.5cm}&\\
&\qquad\qquad \text{$\rho^{\alpha'_{x'},\beta',x'}$ is winning for $\max$} \hspace{-2cm}&  \\
& & (\text{by def.\ of $G^c$ and $\Acc^c$}) \\
&\Leftrightarrow\exists \alpha:\text{a strategy of $\max$}.\; \forall \beta:\text{a strategy of $\min$}.\, \hspace{-4.5cm}&\\
&\qquad\quad \text{$\rho^{\alpha,\beta,x}$ is winning for $\max$}\hspace{-2cm} &  (\text{by def.\ of $\rho^{\alpha,\beta,x}$}) \\
&\Leftrightarrow f(x)=1 \hspace{-2cm}& (\text{by def.\ of $f$}) \,.
\end{align*}
%
Hence $f$ is a fixed point of $\Phi_{c,\sigmatpg}$.

It remains to show that $f:X\to\Omegatpg$ is the least fixed point with respect to the 
pointwise extension of $\leq$.
Let $f':X\to \Omegatpg$ be a fixed point of $\Phi_{c,\sigmatpg}$.
To prove $f\leq f'$, it suffices to prove that
$f'(x)=0$ implies $f(x)=0$ for each $x\in X$.

For each $x'\in X$, we have:
\begin{align*}
&f'(x')=0 \\
&\Leftrightarrow \Phi_{c,\sigmatpg}(f')(x')=0 \hspace{-1cm}& (\text{$f'$ is a fixed point of $\Phi_{c,\sigmatpg}$})\\
&\Leftrightarrow (\sigmatpg\circ \Ftpg f'\circ c)(x')=0 & (\text{by def.\ of $\Phi_{c,\sigmatpg}$}) \\
&\Leftrightarrow c_2(x')=0,\;\;\text{and}\;\;\forall A\in c_1(x').\; \exists x''\in A.\; f'(x'')=0  \hspace{-4.0cm}& \\
& & (\text{by def.\ of $\sigmatpg$ and $\Ftpg$}) \\
&\Leftrightarrow x'\notin \Acc^c ,\;\;\text{and}\\
&\qquad\forall y\in X^c_{\min}\;\text{s.t.}\; (x',y)\in\tau^c.\; \hspace{-3.0cm} \\
&\qquad\quad\exists x''\in X^c_{\max}\;\text{s.t.}\; (y,x'')\in\tau^c.\; f'(x'')=0  \hspace{-9.0cm}& \\
& & (\text{by def.\ of $G^c$ and $\Acc^c$})\,.
\end{align*}
This means that if $f'(x')=0$, then
%
$c_2(x')=0$ and 
for each $A\in c_1(x')$
there exists $x''\in A$ such that $f'(x'')=0$.
%
Hence for each $x\in X$ such that $f'(x)=0$, 
we can define a strategy $\beta_x$ of the player $\min$ so that 
for each strategy $\alpha$ of the player $\max$,
the resulting run $\rho^{\alpha,\beta_x,x}$ from $x$ is not winning for $\max$.

Therefore by the definition of $f$, we have $f(x)=0$.
This concludes the proof. 

The item (\ref{item:prop:concreteDefTwoPlayer2}) is proved in a similar way.
\end{proof}

\subsubsection{Details for Example~\ref{example:FsigmaSatisfyAsmTPG}}\label{subsubsec:details:example:FsigmaSatisfyAsmTPG}

\begin{myproposition}\label{prop:FsigmaSatisfyAsmTPG}
The modality $\sigmatpg:\Ftpg\Omegatpg\to\Omegatpg$ in Def.~\ref{def:concreteDefTwoPlayer}
satisfies Asm.~\ref{asm:behDomainrank}.
\end{myproposition}

\begin{proof}
%
It is easy to see that Cond.~\ref{asm:behDomain4and3} is satisfied.
By monotonicity of the functions $\max$ and $\min$ that are used in $\sigmatpg$, Cond.~\ref{asm:behDomain1and2} is satisfied.
\end{proof}

\subsubsection{Details of Example~\ref{example:rankDomTwoPlayerGame}}\label{subsubsec:details:example:rankDomTwoPlayerGame}

\begin{myproposition}\label{prop:rankDomTwoPlayerGame}
We define an $\Ftpg$-modality $\sigmatpg$ as in Def.~\ref{def:concreteDefTwoPlayer},
and fix an ordinal $\mathfrak{z}$.
We define an $\Ftpg$-algebra $\rtpgz:\Ftpg \Rtpgz\to \Rtpgz$, 
a function $\qtpgz:\Ftpg\Rtpgz\to \Omegatpg$ and 
a partial order $\leqRtpgz$ over $\Rtpgz$ as follows.
%
\begin{enumerate}
\item[1)] 
$\Rtpgz=\{\mathfrak{a}\mid \mathfrak{a}\;\text{is an ordinal s.t.}\;\mathfrak{a}\leq\mathfrak{z}\}$,
and 
%
$\rtpgz:\Ftpg\Rtpgz\to\Rtpgz$ is defined by
\begin{displaymath}
\rtpgz(\Gamma,t)=\begin{cases}
0 & (t=1) \\
\min_{A\in\Gamma}\sup_{\mathfrak{a}\in A}(\mathfrak{a}\nadd 1) & (\text{otherwise})\,.
\end{cases}
\end{displaymath}
Here $\mathfrak{a}\nadd 1$ denotes $\min\{\mathfrak{a}+1,\mathfrak{z}\}$.

\item[2)] 
$\qtpgz:\Rtpgz\to\Omegatpg$ is defined by 
\begin{displaymath}
\qtpgz(\mathfrak{a})=\begin{cases}
0 & (\mathfrak{a}=\mathfrak{z}) \\
1 & (\text{otherwise})\,.
\end{cases}
\end{displaymath}

\item[3)] 
For $\mathfrak{a},\mathfrak{b}\in\Rtpgz$,
$\mathfrak{a}\leqRtpgz\mathfrak{b}\,\defarrow\,\mathfrak{a}\geq\mathfrak{b}$ (note the directions).

\end{enumerate}
Then the triple
$(\rtpgz,\qtpgz,\leqRtpgz)$ 
is a ranking domain.
\end{myproposition}



The most difficult part of the proof is to prove that $r$ is a corecursive algebra
(Cond.~\ref{item:def:rankingdom7} in Def.~\ref{def:rankingdom}).
We prove it separately.
To this end, we first prove 
the following sublemma.

\begin{mylemma}\label{lem:tfSeqWellDefTPG}
The algebra $\rtpgz:\Ftpg \Rtpgz\to \Rtpgz$ 
satisfies 
Cond.~\ref{item:def:rankingdom4} in Def.~\ref{def:rankingdom}.
\end{mylemma}

\begin{proof}
It is easy to see that 
Cond.~\ref{item:def:rankingdom40and3} is satisfied.
%

We prove that Cond.~\ref{item:def:rankingdom41and2} is satisfied.
Let $b_1,b_2:X\to \Rtpgz$ and assume that $b_1\leqRtpgz b_2$.
Let $x\in X$.
As $0$ is the greatest element in $(\Rtpgz,\leqRtpgz)$, 
if $\Phi_{c,\rtpgz}(b_2)(x)=0$ then we have $\Phi_{c,\rtpgz}(b_1)(x)\leqRtpgz\Phi_{c,\rtpgz}(b_2)(x)$.

Assume that $\Phi_{c,\rtpgz}(b_2)(x)=n\in \Rtpgz\setminus\{0\}$.
Let $c(x)=(\Gamma,t)\in \Ftpg X=\pow^2X\times\{0,1\}$.
Then by definition of $\rtpgz$ and $\Ftpg$, $t=0$ and moreover
for all $A\in\Gamma$ there exists $x\in A$ such that $b_2(x)\geq n$.
As $b_1\leqRtpgz b_2$, 
$b_2(x)\geq n$ implies $b_1(x)\geq n$.
Hence we have $\rtpgz\circ \Ftpg b_1(\Gamma,t)\geq n$, and this implies 
$\Phi_{c,\rtpgz}(b_1)(x)\leqRtpgz\Phi_{c,\rtpgz}(b_2)(x)$.
Therefore 
Cond.~\ref{item:def:rankingdom41and2} is satisfied.
\end{proof}

\begin{mylemma}\label{lem:corecAlgNodet}
The algebra $\rtpgz:\Ftpg\Rtpgz\to \Rtpgz$ 
is corecursive.
\end{mylemma}

\begin{proof}
Let $c:X\to \Ftpg X$ be an $\Ftpg$-coalgebra.
It suffices to show that $\Phi_{c,\rtpgz}$ (Def.~\ref{def:Phicsigma}) has a unique fixed point.

By Lem.~\ref{lem:tfSeqWellDefTPG},
the poset $(\Setsto{X}{\Rtpgz},\leqRtpgz)$ (here $\leqRtpgz$ denotes the pointwise extension of $\leqRtpgz$ over $\Rtpgz$) 
is a complete lattice.
Therefore we can construct a transfinite sequence 
\begin{multline*}
\botRtpgz\leqRtpgz \Phi_{c,\rtpgz}(\botRtpgz) \leqRtpgz \cdots \\
\leqRtpgz \Phi_{c,\rtpgz}^{\mathfrak{a}}(\botRtpgz)\leqRtpgz \cdots
\end{multline*}
as in 
Thm.~\ref{thm:KTCC}.2.
By Thm.~\ref{thm:KTCC}.2,
there exists an ordinal $\mathfrak{n}$ such that
%
$\Phi_{c,\rtpgz}^\mathfrak{n}(\botRtpgz)$  is the least fixed point of $\Phi_{c,\rtpgz}$. 

It remains to show that this is the unique fixed point.
Let $f_1,f_2:X\to \Rtpgz$ be fixed points of $\Phi_{c,\rtpgz}$.
We prove $f_1(x)=\mathfrak{a}\Leftrightarrow f_2(x)=\mathfrak{a}$ for each $x\in X$ and $\mathfrak{a}<\mathfrak{z}$ 
by transfinite induction on $\mathfrak{a}$.

For 
the base case,
we have:
\begin{align*}
&f_1(x)=0 \\
&\Leftrightarrow  \rtpgz\circ \Ftpg f_1\circ c(x)=0 & (\text{$f_1$ is a fixed point of $\Phi_{c,\rtpgz}$}) \\
&\Leftrightarrow c_2(x)=1 & (\text{by def.\ of $\rtpgz$ and $\Ftpg$}) \\
&\Leftrightarrow \rtpgz\circ \Ftpg f_2\circ c(x)=0 & (\text{by def.\ of $\Ftpg$ and $\rtpgz$}) \\
&\Leftrightarrow f_2(x)=0 & (\text{$f_2$ is a fixed point of $\Phi_{c,\rtpgz}$})\,.
\end{align*}

Let $0<\mathfrak{a}<\mathfrak{z}$ and 
assume $f(x)=\mathfrak{a}'$ iff $g(x)=\mathfrak{a}'$ for each  $x\in X$ and $\mathfrak{a}'<\mathfrak{a}$.
Note here that by $\mathfrak{a}<\mathfrak{z}$, we have $\mathfrak{a}\nadd 1=\mathfrak{a}+1$.
Hence we have:
\begin{align*}
&f_1(x)=\mathfrak{a} \\
&\Leftrightarrow 
\rtpgz\circ \Ftpg f_1\circ c(x)=\mathfrak{a} \hspace{-3.5cm}
& (\text{$f_1$ is a fixed point of $\Phi_{c,\rtpgz}$}) \\
&\Leftrightarrow
\text{$c_2(x)=0$ and }
\min_{A\in c_1(x)} 
\sup_{x'\in A} 
f_1(x')=\mathfrak{a}-1 \hspace{-3.0cm} &\\
& & \text{(by def.\ of $\rtpgz$ and $\Ftpg$)} \\
&\Leftrightarrow
\text{$c_2(x)=0$ and }\min_{A\in c_1(x)} 
\sup_{x'\in A} 
f_2(x')=\mathfrak{a}-1 \hspace{-3.0cm}& \text{(by IH)} \\
&\Leftrightarrow \Phi_{c,\rtpgz}(f_2)(x)=\mathfrak{a} & (\text{by def.\ of $\Phi_{c,\rtpgz}$}) \\
&\Leftrightarrow f_2(x)=\mathfrak{a} & (\text{$f_2$ is a fixed point of $\Phi_{c,\rtpgz}$})\,.
\end{align*}

Hence we have $f_1(x)=\mathfrak{a}$ iff $f_2(x)=\mathfrak{a}$ for each $x\in X$ and $\mathfrak{a}<\mathfrak{z}$.
This immediately implies that $f_1(x)=\mathfrak{z}$ iff $f_2(x)=\mathfrak{z}$ for each $x\in X$.
\end{proof}

\begin{proof}[Proof (Prop.~\ref{prop:rankDomTwoPlayerGame})]
%
We prove that 
the conditions in Def.~\ref{def:rankingdom} are satisfied.

\subsubsection*{Cond.~\ref{item:def:rankingdom2}}
Let $(\Gamma,t)\in \Ftpg\Rtpgz=\pow^2\rtpgz\times\{0,1\}$ and assume that 
$\sigmatpg\circ \Ftpg \qtpgz(\Gamma,t)=0$.
Then by the definitions of $\sigmatpg$, $\qtpgz$ and $\Ftpg$, 
we have $t=0$ and $\forall A\in\Gamma.\;\mathfrak{z}\in A$.
Hence by the definition of $\rtpgz$, we have $\rtpgz(\Gamma,t)=\mathfrak{z}$. 
Therefore
by the definition of $\qtpgz$, we have $\qtpgz\circ \rtpgz(\Gamma,t)=0$.
Hence 
Cond.~\ref{item:def:rankingdom2} is satisfied.
%

 \subsubsection*{Cond.~\ref{item:def:rankingdom4}}
Already proved in Lem.~\ref{lem:tfSeqWellDefTPG}
 
 \subsubsection*{Cond.~\ref{item:def:rankingdom3}}
By the definition of $\qtpgz:R\to\Omegatpg$, for $a_1,a_2\in \Rtpgz$, 
$a_1\leq a_2$ implies $\qtpgz(a_1)\geq \qtpgz(a_2)$ (here each $\leq$ denotes the usual order).
Therefore $\qtpgz$ is monotone.
%

By its definition, $\qtpgz(\mathfrak{z})=0$ and hence $\qtpgz$ is strict.

For a subset $K\subseteq \Rtpgz$,
it is easy to see that
$\bigsqcup_{a\in K}\bigl(\qtpgz(a)\bigr)
\leqOmegatpg
\qtpgz\bigl(\bigsqcup_{a\in K}a\bigr)$ is satisfied.
We prove the opposite direction.
%
Assume that 
$\qtpgz\bigl(\bigsqcup_{a\in K}a\bigr)=1$.
Then by the definition of $\qtpgz$, 
$\bigsqcup_{a\in K}a=\mathfrak{n}$ 
for some $\mathfrak{n}<\mathfrak{z}$.
By the definition of $\leqRtpgz$, 
this implies $a=\mathfrak{n}$ for some $a\in K$.
Therefore we have $\qtpgz(a)=1$, and this implies
$\bigsqcup_{a\in K}\bigl(\qtpgz(a)\bigr)=1$.
Hence
$\qtpgz\bigl(\bigsqcup_{a\in K}a\bigr)
\leqOmegatpg
\bigsqcup_{a\in K}\bigl(\qtpgz(a)\bigr)$
holds, and $\qtpgz$ is 
 continuous.

\subsubsection*{Cond.~\ref{item:def:rankingdom7}}
Already proved in Lem.~\ref{lem:corecAlgNodet}.

Hence $(\rtpgz,\qtpgz,\leqRtpgz)$ is a ranking domain.
\end{proof}

\begin{myproposition}\label{prop:convRankFuncCoalg}
Let $(\rtpgz,\qtpgz,\leqRtpgz)$ be the ranking domain in
 Prop.~\ref{prop:rankDomTwoPlayerGame}.
 For each function $b:X\to\Rtpgz$, 
we have the followings.
\begin{enumerate}
\item\label{item:prop:convRankFuncCoalg1}
Let  $c:X\to\Ftpg X$ be an $\Ftpg$-coalgebra.
A function $b:X\to \Rtpgz$ is a ranking arrow (Def.~\ref{def:rankingarrow}) for $c$ wrt.\ $(\rtpgz,\qtpgz,\leqRtpgz)$ 
if and only if
$b$ is a ranking function (Def.~\ref{def:rankFunc}) for $G^c$ and $\Acc^c$.

 
\item\label{item:prop:convRankFuncCoalg3} 
 Let $G=(X_{\max},x_{\min},\tau)$ be a game structure and  $\Acc\subseteq X_{\max}$ be a set of accepting states.
 A function $b:X_{\max}\to \Rtpgz$ is a ranking function for $G$ and $\Acc$ if and only if
 $b$ is a ranking arrow  for $c^{G,\Acc}$  wrt.\ $(\rtpgz,\qtpgz,\leqRtpgz)$.

\item\label{item:prop:convRankFuncCoalg2}
 $b(x)<\mathfrak{z}$ iff $\qtpgz\circ b(x)=1$.
\end{enumerate}
\end{myproposition}

\begin{proof}\mbox{}
\subsubsection*{\ref{item:prop:convRankFuncCoalg1}}
For $x\in X$ such that $c(x)=(\Gamma,t)$ we write $c_1(x)$ and $c_2(x)$ for $\Gamma$ and $t$ respectively.
We have:
\begin{align*}
&\text{$b:X\to\Rtpgz$ is a ranking arrow for $c$}\hspace{-2.5cm} \\
&\Leftrightarrow \forall x\in X.\; b(x)\leqRtpgz \rtpgz\circ \Ftpg b\circ c(x) \hspace{-4cm}
 &\text{(by Def.~\ref{def:rankingarrow})}\\
&\Leftrightarrow \forall x\in X.\;  \rtpgz\circ \Ftpg b\circ c(x)\leq b(x) \hspace{-4cm}& \\
& & \text{(by the definition of $\leqRtpgz$)} \\
&\Leftrightarrow 
\forall x\in X.\;\bigl( c_2(x)= 0\;\Rightarrow\;\min_{A\in c_1(x)}\sup_{x'\in A} b(x')\nadd 1\leq b(x)\bigr)\hspace{-8cm} \\
& & \text{(by the definitions of $\rtpgz$ and $\Ftpg$)} \\
&\Leftrightarrow 
\forall x\in X^c_{\max}.\;\\
& \qquad\bigl( x\notin \Acc^c\;\Rightarrow\;\min_{y:(x,y)\in\tau^c}\sup_{x':(y,x')\in\tau^c} b(x')\nadd 1\leq b(x)\bigr)\hspace{-8cm} \\
& & \text{(by the definitions of $G^c$ and $\Acc^c$)} \\
&\Leftrightarrow \text{$b:X\to\Rtpgz$ is a ranking function for $G^c$ and $\Acc^c$} \hspace{-8cm} \\
& & (\text{by Def.~\ref{def:rankFunc}})\,.
\end{align*}
%
Hence Cond.~\ref{item:prop:convRankFuncCoalg1} is satisfied.

\subsubsection*{\ref{item:prop:convRankFuncCoalg3}}
This is proved in a similar manner to Cond.~\ref{item:prop:convRankFuncCoalg1}.

\subsubsection*{\ref{item:prop:convRankFuncCoalg2}}
Immediate from the definition of $\qtpgz$.
\end{proof}

\vspace{0.5\baselineskip}
\subsection{\textbf{Proof of Lem.\,\ref{lem:cousotCousotExt}}}
The basic idea of the proof is similar to the proof of Thm.~\ref{thm:KTCC}.2 in \cite{cousotC79}.

\begin{proof}
Let $\mathfrak{m}$ be an ordinal such that $|L|<|\mathfrak{m}|$.
As $|\{f^\mathfrak{a}(l)\in L\mid \mathfrak{a}\leq\mathfrak{m}\}|\leq|L|$, 
there exist ordinals $\mathfrak{a}, \mathfrak{a}'$  
such that $f^\mathfrak{a}(l)=f^{\mathfrak{a}'}(l)$.
Without loss of generality, we assume $\mathfrak{a}<\mathfrak{a}'$\,. 

By monotonicity of $f$ and that $l\sqsubseteq f(l)$, we have 
\[
f^\mathfrak{a}(l)=f^{\mathfrak{a}+1}(l)=\cdots=f^{\mathfrak{a}'}(l)\,.
\]
This concludes the proof.
\end{proof}

\vspace{0.5\baselineskip}
\subsection{\textbf{Proof of Prop.\,\ref{prop:completeness}}}
\begin{wrapfigure}[7]{r}{3.8cm}
\vspace{-.4cm}
 \begin{xy}
 \xymatrix@R=2.0em@C=2.6em{
 {F X} \ar@{}[dr]|{=} \ar[r]_{F \uniquefp{c}_r} 
 \ar@/^.8em/[rr]^{F\,\llbracket \mu\sigma\rrbracket_c}
 & {F R} \ar[d]_{r} \ar@{}[dr]|{=} 
 \ar[r]_{F q}  & {F \Omega} \ar[d]_{\sigma} \\
 {X} \ar[u]_{c}  \ar[r]^{\uniquefp{c}_r} \ar@/_.8em/[rr]_{\llbracket \mu\sigma\rrbracket_c} & {R} \ar[r]^{q} & {\Omega} 
 }
 \end{xy}
\end{wrapfigure}
\vspace{1mm}
\proofmark
By $\uniquefp{c}_r=r\circ F \uniquefp{c}_r\circ c$ and 
 $q\circ r=\sigma\circ Fq$, we have:
 \begin{displaymath}
q\circ \uniquefp{c}_r= \Phi_{c,\sigma}(q\circ \uniquefp{c}_r)\,.
\end{displaymath}
This means that $q\circ \uniquefp{c}_r$ is a fixed point of $\Phi_{c,\sigma}$.
As $\sem{\mu\sigma}_c$ is the least fixed point of $\Phi_{c,\sigma}$ (Def.~\ref{def:lfpsem}),
we have 
\begin{displaymath}
q\circ \uniquefp{c}_r\sqsupseteq_{\Omega}\sem{\mu\sigma}_c\,.
\end{displaymath}
Together with Thm.~\ref{thm:soundnessranking},
we have $q\circ \uniquefp{c}_r=\llbracket\mu\sigma\rrbracket_c$.
\qed
\vspace{1mm}

\vspace{0.5\baselineskip}
\subsection{\textbf{Proof of Prop.~\ref{prop:constBehSituationPTS}}}\label{subsec:proof:prop:constBehSituationPTS}
\begin{proof}
\allowdisplaybreaks[4]
Recall that $\Reach_{\Acc^c}$ is defined as the function $f:X\to [0,1]$ in Def.~\ref{def:winnerMDPNew}.
We first show that this $f$ is a fixed point of $\Phi_{c,\sigmapts}$.
For $x\in X$, we have: 
\begin{align*}
&\Phi_{c,\sigmapts}(f)(x) \\
&= (\sigmapts\circ \Fpts f\circ c)(x) & (\text{by def.\ of $\Phi_{c,\sigmapts}$}) \\
&= \begin{cases}
1 & (c_2(x)=1) \\
{\displaystyle\sum_{x'\in \supp(c_1(x))}c_1(x)(x')\cdot f(x')} & (c_2(x)=0)
\end{cases}\hspace{-3cm}& \\
& &(\text{by def.\ of $\sigmapts$ and $\Fpts$}) \\
&= \begin{cases}
1 & (c_2(x)=1) \\
{\displaystyle\sum_{x'\in \supp(c_1(x))}c_1(x)(x')\cdot 
\lim_{n\to\infty}
f_n(x') 
} 
& (c_2(x)=0)
\end{cases} \hspace{-5cm}& \\
& & (\text{by def.\ of $f$}) \\
&= \begin{cases}
1 & (c_2(x)=1) \\
{\displaystyle
\lim_{n\to\infty}
\sum_{x'\in \supp(c_1(x))}c_1(x)(x')\cdot 
f_n(x') 
} 
& (c_2(x)=0)
\end{cases} \hspace{-5cm}& \\
&= \begin{cases}
1 & (c_2(x)=1) \\
{\displaystyle
\lim_{n\to\infty}
f_{n+1}(x)} & (c_2(x)=0)
\end{cases} \hspace{-2cm}& (\text{by def.\ of $f_n$}) \\
%
&=
\lim_{n\to\infty}
f_n(x) & (\text{by def.\ of $f_n$}) \\
&= f(x) & (\text{by def.\ of $f$})\,.
\end{align*}
Hence $f$ is a fixed point of $\Phi_{c,\sigmapts}$.

It remains to show that $f$ is the least fixed point.
Let $f':X\to\Omegapts$ be a fixed point of $\Phi_{c,\sigmapts}$.

We prove $f(x)\leq f'(x)$ for each $x\in X$.
To this end, 
by the definition of $f$,
it suffices to prove 
$f_n(x) \;\leq\; f'(x)$
 for each 
 $x\in X$ and $n\in\mathbb{N}$.
We prove this by induction on $n$.

For $n=0$, 
it is immediate from that $f_0(x)=0$. 

For $n>0$, we have:
\begin{align*}
&f_n(x) \\
&= \begin{cases}
1 & (c_2(x)=1) \\
{\displaystyle
\sum_{x'\in \supp(c_1(x))}c_1(x)(x')\cdot f_{n-1} (x')
} & (c_2(x)=0) 
\end{cases} \hspace{-3.5cm} 
\\
& &\text{(by def.\ of $f_n$)} \\
&\leq \begin{cases}
1 & (c_2(x)=1) \\
{\displaystyle 
\sum_{x'\in \supp(c_1(x))}}c_1(x)(x')\cdot f'(x') & (c_2(x)=0) 
\end{cases} \hspace{-2.5cm}
&\text{(by IH)} \\
&= (\sigmapts\circ \Fpts f'\circ c)(x) & (\text{by def.\ of $\Fpts$ and $\sigmapts$})\\
&= f'(x) & (\text{$f'$ is a fixed point})\,.
\end{align*}
%
Hence we have $f_n(x) \;\leq\; f'(x)$
 for each  $x\in X$ and $n\in\mathbb{N}$,
 and this implies $f(x)\leq f'(x)$ for each $x\in X$.
 Therefore $f$ is the least fixed point of $\Phi_{c,\sigmapts}$.
\end{proof}



\vspace{0.5\baselineskip}
\subsection{\textbf{Proof of Prop.~\ref{prop:rankDomPTSTradD}}}\label{subsec:proof:prop:rankDomPTSTradD}
\begin{proof}\mbox{}
\subsubsection*{\ref{item:prop:convRankFuncCoalg1MPTS}}
For each $b':X\to\Rptsp$, we have:
\begin{align*}
&\text{$b'$ is an $\varepsilon$-additive ranking supermartingale} \hspace{-2.5cm}&\\
&\Leftrightarrow \;
\bigl(\forall x\in X.\,\\
&\qquad\quad c_2(x)= 0\, \\
&\qquad\qquad \Rightarrow\sum_{x'\in \supp(c_1(x))} c_1(x)(x')\cdot b'(x')+\varepsilon\leq b'(x)\bigr) \hspace{-2.5cm}& \\
& & (\text{by Def.~\ref{def:multAddSMPTS}}) \\
&\Leftrightarrow\; \forall x\in X.\;  \rptspe\circ \Fpts b'\circ c(x)\leq b'(x)\hspace{-2.5cm}& (\text{by def.\ of $\rptspe$}) \\
&\Leftrightarrow\; 
\text{$b'$ is a (categorical) ranking arrow wrt.\ $(\rptspe,\qptsp,\leqRptsp)$} \hspace{-2.5cm} \\
& & (\text{by Def.~\ref{def:rankingarrow}})\,.
\end{align*}
Hence Cond.~\ref{item:prop:convRankFuncCoalg1MPTS} holds.

\subsubsection*{\ref{item:prop:convRankFuncCoalg2MPTS}}
Immediate from the definition of $\qptsp$.
\end{proof}

\vspace{0.5\baselineskip}
\subsection{\textbf{Proof of Prop.~\ref{prop:pseudoRD}}}\label{subsec:proof:lem:pseudoRD}
\begin{proof}\mbox{}
%
\subsubsection*{Cond.~\ref{item:def:rankingdom2}}
Let $(\psi,t)\in \Fpts\Rptsp=\dist \Rptsp\times\{0,1\}$.
Then we have:
\allowdisplaybreaks[4]
\begin{align*}
&\qptsp\circ \rptspe(\psi,t) \\
&= \begin{cases}
\qptsp(0) & (t=1) \\
\qptsp(\sum_{a\in\supp(\psi)} a\cdot\psi(a)+\varepsilon) & (t=0)
\end{cases} \hspace{-3.3cm} &\\
& & (\text{by def.\ of $\rptspe$}) \\
&= \begin{cases}
1 & (t=1\;\text{or}\;\sum_{a\in\supp(\psi)} a\cdot\psi(a)+\varepsilon<\infty)\\
0 & (\text{otherwise})
\end{cases} \hspace{-3.3cm} & \\
& & (\text{by def.\ of $\qptsp$}) \\
&\leq \begin{cases}
1 & (t=1\;\text{or}\;\psi([0,\infty))=1)\\
0 & (\text{otherwise})
\end{cases}\\
&\leq \begin{cases}
1 & (t=1)\\
\psi([0,\infty)) & (\text{otherwise})
\end{cases}\\
&= \begin{cases}
1 & (t=1)\\
1\cdot \dist\qpts(\psi)(1) & (\text{otherwise})
\end{cases} & (\text{by def.\ of $\qptsp$}) \\
&=\sigmapts\circ \Fpts \qptsp(\psi,t) & (\text{by def.\ of $\sigmapts$})\,.
\end{align*}
Hence we have $\qptsp\circ \rptspe \leqOmegapts \sigmapts\circ \Fpts \qptsp$.

\subsubsection*{Cond.~\ref{item:def:rankingdom4}}
It is easy to see that Cond.~\ref{item:def:rankingdom40and3} is satisfied.

%

We prove that Cond.~\ref{item:def:rankingdom41and2} is satisfied.
Let $b_1,b_2:X\to \Rptsp$ and assume that $b_1\leqRptsp b_2$.
For each $x\in X$ 
such that $c(x)=(\varphi,t)\in \Fpts X=\dist X\times\{0,1\}$,
we have:
\begin{align*}
&\Phi_{c,\rptspe}(b_1)(x) \\
&=\rptspe\circ \Fpts b_1(\varphi,t)  & (\text{by def.\ of $\Phi_{c,\rptspe}$}) \\
&= \begin{cases}
1 & (t=1) \\
\sum_{x\in \supp(\varphi)} \varphi(x)\cdot b_1(x)+\varepsilon 
& (t=0)
\end{cases} \hspace{-2cm}& \\
& & (\text{by def.\ of $\rptspe$}) \\ 
&\leq \begin{cases}
1 & (t=1) \\
\sum_{x\in \supp(\varphi)} \varphi(x)\cdot b_2(x)+\varepsilon
& (t=1)
\end{cases} \hspace{-2cm}& \\ 
&= \rptspe\circ \Fpts b_2(\varphi,t)&  (\text{by def.\ of $\rptspe$}) \\ 
&= \Phi_{c,\rptspe}(b_2)(x)&(\text{by def.\ of $\Phi_{c,\rptspe}$})\,.
\end{align*}
Therefore Cond.~\ref{item:def:rankingdom41and2} is satisfied.
%
%

\subsubsection*{Cond.~\ref{item:def:rankingdom3}}
By the definition of $\qptsp$, for each $a_1,a_2\in \Rptsp$, 
$a_1\geq a_2$ implies $\qptsp(a_1)\leq \qptsp(a_2)$ (here each $\leq$ denotes the standard order).
Hence $\qptsp$ is monotone.

By the definition, we have $\qptsp(\infty)=0$.
Hence $\qptsp$ is strict.


We prove that 
$\qptsp$ is 
 continuous.
Namely, for each 
 subset $K\subseteq \Rptsp$, we prove 
\begin{displaymath}
\bigsqcup_{a\in K}\bigl(\qptsp(a)\bigr)
=\qptsp\bigl(\bigsqcup_{a\in K}a\bigr)\,.
\end{displaymath}
It is easy to prove $(\mathrm{LHS})\leq(\mathrm{RHS})$.
We prove the opposite direction.
To this end, by the definition of $\qptsp$, it suffices to prove 
that if the right-hand side is $1$ then the left-hand side is also $1$.
%
We can prove it as follows (here $\bigsqcup$ denotes the supremum with respect to $\leqRptsp$, and
$\bigwedge$ denotes the infimum with respect to the ordinary order over $\Omegapts$).
\begin{align*}
 \qptsp\bigl(\bigsqcup_{a\in K}a\bigr)=1 
&\Rightarrow\; \bigsqcup_{a\in K}a <\infty 
& (\text{by def.\ of $\qptsp$}) \\
&\Rightarrow\; \bigwedge_{a\in K}a <\infty 
& (\text{by def.\ of $\leqRptsp$}) \\
&\Rightarrow\; \exists a\in K.\; a <\infty \\
&\Rightarrow\; \exists a\in K.\; \qptsp(a)=1 
& (\text{by def.\ of $\qptsp$}) \\
&\Rightarrow\; \bigsqcup_{a\in K}\bigl(\qptsp(a)\bigr)=1
\end{align*}
%
Hence $\qptsp$ is 
 continuous.

This concludes the proof.
\end{proof}

\vspace{0.5\baselineskip}
\subsection{\textbf{Multiplicative Ranking Supermartingale, Categorically}}\label{subsec:MRSCat}
\begin{mydefinition}\label{def:multSM}
Let $M=(X,\tau)$ be a PTS.
Let $\alpha\in(0,1)$.
A function $b:X\to\nonnegrealsinf$ is an \emph{$\alpha$-multiplicative ranking supermartingale} if
we have 
\begin{displaymath}
{\textstyle\sum_{x'\in \supp(\tau(x))} \tau(x)(x')\cdot b(x')
\;\leq\; \alpha\cdot b(x)}
\end{displaymath} 
for each $x \in X\setminus \Acc$,
and moreover there exists $\delta>0$ such that $b(x)\geq\delta$ for each $x\in X\setminus \Acc$.
\end{mydefinition}

The following is an attempt to define a ranking domain for 
$\alpha$-multiplicative supermartingales. 
Let us fix real numbers $\alpha\in(0,1)$ and $\delta>0$, and
define an $\Fpts$-algebra $\rptspa:\Fpts\Rptspm\to \Rptspm$, an arrow $\qptsp:\Rptsp\to\Omegapts$ and
a partial order $\leqRptspm$ over $\Rptspm$
as follows.
\begin{enumerate}

\item[1)] 
For each $(\psi,t)\in \Fpts\Rptspm=\dist\Rptspm\times \{0,1\}$,
\begin{multline*}
\rptspa(\psi,t)= \\
\begin{cases}
\alpha\delta  & (t=1) \\
\frac{1}{\alpha}\cdot\bigl(\,\sum_{a\in\supp(\psi)}a\cdot \psi(a)\, \bigr) & (\text{otherwise})\,.
\end{cases}
\end{multline*}

\item[2)] 
$\qptsp(\infty)=0$ and $\qptsp(a)=1$ if $a <\infty$. 

\item[3)] 
$a\leqRptspm b\;\defarrow\; a\geq b$ 
(note the direction). 
\end{enumerate}

The following proposition is analogous to Prop.~\ref{prop:rankDomPTSTradD}.

\begin{myproposition}\label{prop:rankDomPTSTradDMult}
In this setting, for each $c:X\to \Fpts X$,
we have the following.
\begin{enumerate}
\renewcommand{\labelenumi}{\alph{enumi})}
\renewcommand{\theenumi}{\alph{enumi})}
\item\label{item:prop:convRankFuncCoalg1MPTSMult}
Let $b':X\to \Rptspm$ and assume $b'\leqRptspm \rptspa\circ \Fpts b'\circ c$.
If we define $b:X\to[0,\infty]$ by $b(x)=b'(x)$, then $b$ is an $\alpha$-multiplicative ranking supermartingale.

\item\label{item:prop:convRankFuncCoalg1MPTSMult2}
Let $b:X\to[0,\infty]$ be an $\alpha$-multiplicative ranking supermartingale such that $b(x)\geq \delta$ if $x\in X\setminus\Acc$.
If we define $b':X\to\Rptspm$ by $b'(x)=\max\{\alpha\delta,b(x)\}$, then $b'$ 
satisfies $b'\leqRptspm \rptspa\circ \Fpts b'\circ c$.


\item\label{item:prop:convRankFuncCoalg2MPTSMult}
For each $b':X\to \Rptsp$, we have
 $b'(x)\neq\infty$ iff $q'_{\pts}\circ b'(x)=1$.
\end{enumerate}
\end{myproposition}

\begin{proof}\mbox{}
\subsubsection*{\ref{item:prop:convRankFuncCoalg1MPTSMult}}
Let $x\in X\setminus\Acc$.
Then we have: 
\begin{align*}
&\sum_{x'\in \supp(\tau(x))} \tau(x)(x')\cdot b(x') \\ 
&= \sum_{x'\in \supp(c_1(x))} c_1(x)(x')\cdot b'(x') \\
&= \alpha\cdot (\rptspa\circ \Fpts b'\circ c)(x) \\
&\leq \alpha\cdot b'(x) \\
&=\alpha\cdot b(x)\,.
\end{align*}
By the inequality above, we also have 
\[
b(x)\geq \frac{1}{\alpha}\sum_{x'\in \supp(c_1(x))} c_1(x)(x')\cdot b'(x')\,.
\]
As $b'(x')\geq\alpha\delta$ holds for each $x'\in X$, we have
$b(x)\geq\delta$.
Therefore $b$ is an $\alpha$-multiplicative ranking supermartingale.

\subsubsection*{\ref{item:prop:convRankFuncCoalg1MPTSMult2}}
For each $x\in X$, we have:
\begin{align*}
&(\rptspa\circ \Fpts b'\circ c)(x) \\
&= \begin{cases}
\alpha\delta & (c_2(x)=1) \\
\frac{1}{\alpha}\sum_{x'\in \supp(c_1(x))} c_1(x)(x')\cdot b'(x') & (c_2(x)=0) \\
\end{cases}\\
&= \begin{cases}
\alpha\delta & (x\in\Acc) \\
\frac{1}{\alpha}\sum_{x'\in \supp(\tau(x))} \tau(x)(x')\cdot b(x') & (x\in X\setminus \Acc) \\
\end{cases}\\
&\leq \begin{cases}
\alpha\delta & (x\in\Acc) \\
\frac{1}{\alpha}\cdot\alpha b(x) & (x\in X\setminus \Acc) \\
\end{cases}\\
&\leq 
b'(x)\,.
\end{align*}
Therefore by the definition of $\leqRptspm$, we have
$b'\leqRptspm \rptspa\circ \Fpts b'\circ c$.

\subsubsection*{\ref{item:prop:convRankFuncCoalg2MPTSMult}}
Immediate from the definition of $\qptsp$.
\end{proof}

Therefore the triple $(\rptspa,\qptsp,\leqRptspm)$ is suited for
accommodating $\alpha$-multiplicative supermartingales in our categorical framework. 
However it is not a ranking domain.
The following proposition is analogous to Prop.~\ref{prop:pseudoRD}.

\begin{myproposition}\label{prop:pseudoRDMult}
The triple $(\rptspa,\qptsp,\leqRptsp)$ 
introduced above
satisfies the conditions of a ranking domain (Def.~\ref{def:rankingdom}),
except for Cond.~\ref{item:def:rankingdom7}.
\end{myproposition}

\begin{proof}\mbox{}
%
\subsubsection*{Cond.~\ref{item:def:rankingdom2}}
Let $(\psi,t)\in \Fpts\Rptspm=\dist \Rptspm\times\{0,1\}$.
Then we have:
\allowdisplaybreaks[4]
\begin{align*}
&\qptsp\circ \rptspa(\psi,t) \\
&= \begin{cases}
\qptsp(\alpha\delta) & (t=1) \\
\qptsp(\frac{1}{\alpha}\cdot\sum_{a\in\supp(\psi)} a\cdot\psi(a)) & (t=0)
\end{cases} \hspace{-3.3cm} &\\
& & (\text{by def.\ of $\rptspa$}) \\
&= \begin{cases}
1 & (t=1\;\text{or}\;\frac{1}{\alpha}\cdot\sum_{a\in\supp(\psi)} a\cdot\psi(a)<\infty)\\
0 & (\text{otherwise})
\end{cases} \hspace{-3.3cm} & \\
& & (\text{by def.\ of $\qptsp$}) \\
&\leq \begin{cases}
1 & (t=1\;\text{or}\;\psi([0,\infty))=1)\\
0 & (\text{otherwise})
\end{cases}\\
&\leq \begin{cases}
1 & (t=1)\\
\psi([0,\infty)) & (\text{otherwise})
\end{cases}\\
&= \begin{cases}
1 & (t=1)\\
1\cdot \dist\qpts(\psi)(1) & (\text{otherwise})
\end{cases} & (\text{by def.\ of $\qptsp$}) \\
&=\sigmapts\circ \Fpts \qptsp(\psi,t) & (\text{by def.\ of $\sigmapts$})\,.
\end{align*}
Hence we have $\qptsp\circ \rptspa \leqOmegapts \sigmapts\circ \Fpts \qptsp$.

\subsubsection*{Cond.~\ref{item:def:rankingdom4}}
It is easy to see that Cond.~\ref{item:def:rankingdom40and3} is satisfied.

%

We prove that Cond.~\ref{item:def:rankingdom41and2} is satisfied.
Let $b_1,b_2:X\to \Rptspm$ and assume that $b_1\leqRptspm b_2$.
For each $x\in X$ 
such that $c(x)=(\varphi,t)\in \Fpts X=\dist X\times\{0,1\}$,
we have:
\begin{align*}
&\Phi_{c,\rptspa}(b_1)(x) \\
&=\rptspa\circ \Fpts b_1(\varphi,t)  & (\text{by def.\ of $\Phi_{c,\rptspa}$}) \\
&= \begin{cases}
1 & (t=1) \\
\frac{1}{\alpha}\cdot\sum_{x\in \supp(\varphi)} \varphi(x)\cdot b_1(x)
& (t=0)
\end{cases} \hspace{-2cm}& \\
& & (\text{by def.\ of $\rptspa$}) \\ 
&\leqRptspm \begin{cases}
1 & (t=1) \\
\frac{1}{\alpha}\cdot\sum_{x\in \supp(\varphi)} \varphi(x)\cdot b_2(x)
& (t=0)
\end{cases} \hspace{-2cm}& \\ 
&= \rptspa\circ \Fpts b_2(\varphi,t)&  (\text{by def.\ of $\rptspa$}) \\ 
&= \Phi_{c,\rptspa}(b_2)(x)&(\text{by def.\ of $\Phi_{c,\rptspa}$})\,.
\end{align*}
Therefore Cond.~\ref{item:def:rankingdom41and2} is satisfied.
%
%

\subsubsection*{Cond.~\ref{item:def:rankingdom3}}
It is proved in a similar manner to the proof of Prop.~\ref{prop:pseudoRD}.
%
%
\end{proof}

\vspace{2mm}
\noindent
\begin{minipage}{0.28\hsize}
\begin{myexample}\label{example:rankSupMartNotCorecNewMult}
\end{myexample}
\end{minipage}\!\!
We define a coalgebra 
$c:X\to \Fpts X$ 
as in Example~\ref{example:rankSupMartNotCorecNew}.
We fix $\alpha\in(0,1)$ and $\delta>0$.
We assume that $\alpha>1/2$.
%
For this coalgebra, we define arrows $b_1,b_2:X\to \Rptspm$ by:
$b_1(x_0)=\frac{1}{2}\frac{1}{\alpha}\bigl(\frac{\alpha\delta}{2\alpha-1}+\delta\bigr)$, $b_1(x_1)=\frac{\alpha\delta}{2\alpha-1}$, $b_1(x_2)=\delta$ and $b_1(x_3)=\alpha\delta$; and
$b_2(x_0)=\infty$, $b_2(x_1)=\infty$, $b_2(x_2)=\delta$ and $b_2(x_3)=\alpha\delta$.
Both of these qualify as coalgebra-algebra homomorphisms from $c$ to $\rptspa$. Therefore $\rptspa$ is not corecursive. 
%
\vspace{2mm}


\vspace{0.5\baselineskip}
\subsection{\textbf{Proof of Prop.~\ref{prop:rankDomPTS}}}\label{subsec:proof:prop:rankDomPTS}
We prove Prop.~\ref{prop:rankDomPTS} in a similar manner to the proof of Prop.~\ref{prop:rankDomTwoPlayerGame}:
we first prove that $r_{\pts}$ is a corecursive algebra separately.
To this end, we 
first prove some lemmas.

\begin{mylemma}\label{lem:orderRposetDNew}
The order $\leqRpts$ in Prop.~\ref{prop:rankDomPTS} is a partial order and
$\Rpts$ is a complete lattice with respect to this order.
\end{mylemma}

\begin{mysublemma}\label{sublem:LSmeasDNew}
For every 
nondecreasing function 
$G:\mathbb{N}\to [0,1]$, 
there exists a unique distribution
$\varphi$ over $\Ninf$ such that 
$\varphi([0,a])=G(a)$ for each $a\in \mathbb{N}$.
\end{mysublemma}

\begin{proof}
We define a distribution $\varphi$ over $\Ninf$ by 
\begin{displaymath}
\varphi(a)=\begin{cases}
G(a) & (a=0) \\
G(a)-G(a-1)  & (0<a<\infty)\\
1-\lim_{a'\to\infty}G(a') & (a=\infty) \,. 
\end{cases}
\end{displaymath}
As $G$ is nondecreasing, $\varphi(a)\geq 0$ for each $a$.
By its definition, we have $\sum_{a\in\Ninf}\varphi(a)= 1$.
Hence $\varphi$ is a distribution.

Let $\varphi'$ be a distribution such that $\varphi'([0,a])=G(a)$.
%
Then 
we have $\varphi(0)=G(0)=\varphi'([0,0])=\varphi'(0)$.
Moreover for each $a\in\mathbb{N}\setminus\{0\}$, we have:
\begin{displaymath}
\varphi(a)=G(a)-G(a-1)=\varphi'([0,a])-\varphi'([0,a-1])=\varphi'(a)\,.
\end{displaymath}
Therefore we have $\varphi(a)=\varphi'(a)$ for each $a\in\mathbb{N}$,
and this implies $\varphi(\infty)=\varphi'(\infty)$.
Hence uniqueness is proved, and this concludes the proof.
\end{proof}

\begin{proof}[Proof (Lem.~\ref{lem:orderRposetDNew})]
We first prove that $\leqRpts$ is a partial order.
Reflexivity and transitivity are immediate from those of the standard order $\leq$ over $[0,1]$. 
Assume that $\varphi\leqRpts\varphi'$ and $\varphi'\leqRpts\varphi$.
By the definition of $\leqRpts$, we have $\varphi([0,a])=\varphi'([0,a])$ for each $a\in[0,\infty]$.
Then by Sublem.~\ref{sublem:LSmeasDNew}, we have $\varphi=\varphi'$.
Hence antisymmetry is also satisfied.

We prove that 
each subset $K\subseteq \Rpts$ has the least upper bound. We define $G:\mathbb{N}\to[0,1]$ by 
$G(a)=\sup_{\varphi\in K}\varphi([0,a])$.
Note that for each $\varphi\in K$, $a\leq b$ implies $\varphi([0,a])\leq \varphi([0,b])$.
Hence by the monotonicity of supremums, $G$ is nondecreasing.
Therefore by Sublem.~\ref{sublem:LSmeasDNew},
there  exists a unique distribution $\varphi^K\in\dist\Ninf$ such that 
$\varphi^K([0,a])=G(a)$ for each $a\in[0,\infty]$.

We prove that $\varphi^K$ is the least upper bound of $K$.

Let $\varphi\in K$. For each $a\in\mathbb{N}$, we have:
\begin{displaymath} 
\varphi([0,a])\;\leq\;\sup_{\varphi'\in K}\varphi'([0,a])\;=\;\varphi^K([0,a])\,.
\end{displaymath}
Hence by the definition of $\leqRpts$, $\varphi$ is an upper bound of $K$.

Let $\varphi'$ be an upper bound of $K$.
Then by the definition of $\leqRpts$, we have $\varphi([0,a])\leq \varphi'([0,a])$ for each $\varphi\in K$ and $a\in\mathbb{N}$.
Therefore we have $\varphi^K([0,a])\leq \varphi'([0,a])$, and this means 
$\varphi^K\leqRpts\varphi'$  by the definition of $\leqRpts$.
Hence $\varphi$ is the least upper bound of $K$.

In a similar manner, we can prove that 
each $K\subseteq \Rpts$ has the greatest lower bound.
\end{proof}

\begin{mylemma}\label{lem:tfSeqWellDefPTSD}
%
The algebra $\rpts:\Fpts \Rpts\to \Rpts$ in Prop~\ref{prop:rankDomPTS} 
satisfies 
Cond.~\ref{item:def:rankingdom4} in Def.~\ref{def:rankingdom}.
\end{mylemma}

\begin{proof}
It is already proved in Lem.~\ref{lem:orderRposetDNew} that Cond.~\ref{item:def:rankingdom40and3} is satisfied.

%


We prove that Cond.~\ref{item:def:rankingdom41and2} is satisfied.
Let $f_1,f_2:X\to \Rpts$ and assume $f_1\leqRpts f_2$.
Let $x\in X$ and assume that $c(x)=(\varphi,t)\in \Fpts X$.


If $t=1$ then we have:
\begin{displaymath}
 \rpts\circ\Fpts f_1(f)(\varphi,t)=\rpts\circ\Fpts f_2(f)(\varphi,t)=\delta_0.
\end{displaymath}
Therefore we have 
$\Phi_{c,\sigmapts}(f_1)(x)=\Phi_{c,\sigmapts}(f_2)(x)$.

Let $t=0$ and $a\in\mathbb{N}$.
If $a=0$, by the definition of $\rpts$, we have:
\begin{displaymath}
\rpts\circ\Fpts f_1(f)(\varphi,t)([0,a])=\rpts\circ\Fpts f_2(f)(\varphi,t)([0,a])=0\,.
\end{displaymath}

If $a>0$, we have:
\begin{align*}
&\rpts\circ\Fpts f_1(f)(\varphi,t)([0,a]) \\
&= 
\sum_{x\in\supp(\varphi)}\bigl(\varphi(x)\cdot f_1(x)([0,a-1]) \bigr) 
& (\text{by def.\ of $\rpts$})
\\
&\leq 
\sum_{x\in\supp(\varphi)}\bigl(\varphi(x)\cdot f_2(x)([0,a-1]) \bigr)
& (\text{by $f_1\leqRpts f_2$})
\\
%
&= \rpts\circ \Fpts f_2(\varphi,t)([0,a]) & (\text{by def.\ of $\rpts$})\,.
\end{align*}

Therefore by the definition of $\leqRpts$,
 we have $\rpts\circ \Fpts f_1(\varphi,t)\leqRpts\rpts\circ \Fpts f_2(\varphi,t)$. 
 This means that $\Phi_{c,\rpts}$ is monotone.
This concludes the proof.
\end{proof}

\begin{mylemma}\label{lem:rCorecNewDPTSPTS}
The $\Fpts$-algebra $\rpts:\Fpts\Rpts\to \Rpts$ in Prop.~\ref{prop:rankDomPTS} is corecursive.
\end{mylemma}

\begin{proof}
Let $c:X\to \Fpts X$ be an $\Fpts$-coalgebra.
It suffices to show that 
the function 
$\Phi_{c,\rpts}:\Setsto{X}{(\Rpts)}\to \Setsto{X}{(\Rpts)}$ (Def.~\ref{def:Phicsigma}) has a unique fixed point.

By Lem.~\ref{lem:tfSeqWellDefPTSD},
the poset $(\Setsto{X}{(\Rpts)},\leqRpts)$ (here $\leqRpts$ denotes the pointwise extension) 
is a complete lattice.
Therefore we can construct a transfinite sequence 
\begin{multline*}
\botRpts\leqRpts \Phi_{c,\rpts}(\botRpts) \leqRpts \cdots \\
\leqRpts \Phi_{c,\rpts}^{\mathfrak{a}}(\botRpts)\leqRpts \cdots
\end{multline*}
as in
Thm.~\ref{thm:KTCC}.2.
By Thm.~\ref{thm:KTCC}.2,
there exists an ordinal $\mathfrak{n}$ such that
$\Phi_{c,\rpts}^\mathfrak{n}(\botRpts)$  is the least fixed point of $\Phi_{c,\rpts}$. 

%

It remains to show that this is a unique fixed point.
Let $f_1,f_2:X\to \Rpts$ be fixed points of $\Phi_{c,\rpts}$.
We prove 
\begin{equation}\label{eq:f1f2eqxa}
f_1(x)(a)=f_2(x)(a)
\end{equation}
 for each $x\in X$ and $a\in\mathbb{N}$.
For each $x\in X$ such that $(\varphi,t)=c(x)$,
we write $c_1(x)$ and $c_2(x)$ for $\varphi\in\dist X$ and $t\in\{0,1\}$ respectively.


We first prove (\ref{eq:f1f2eqxa}) for each $x\in X$ such that $c_2(x)=1$.
For such $x$, we have:
\begin{align*}
f_1(x)(a) 
&= \Phi_{c,\rpts}(f_1)(x) &(\text{$f_1$ is a fixed point}) \\
&= \rpts\circ \Fpts f_1\circ c(x) \\
&=
\delta_0 \\
&= \rpts\circ \Fpts f_2\circ c(x) \\
&= \Phi_{c,\rpts}(f_2)(x) &(\text{$f_2$ is a fixed point}) \\
&=g(x)\,.
\end{align*}
Hence we have $f_1(x)(a)=f_2(x)(a)$ for each $a$.

It remains to prove (\ref{eq:f1f2eqxa}) for each $x\in X$ such that $c_2(x)=0$.
We prove it by the induction on $a$.

If $a=0$ then by the definition of $\rpts$, we have
$f_1(x)(a)=f_2(x)(a)=0$.

Let $a>0$ and assume $f_1(x')(a')=f_2(x')(a')$ for each $x'\in X$ and $a'< a$.
Then we have:
\begin{align*}
&f_1(x)(a) \\
&= \rpts\circ \Fpts f_1\circ c(x)(a)& (\text{$f_1$ is a fixed point}) \\
&= \sum_{x'\in\supp(c_1(x))}c_1(x)(x')\cdot f_1(x')(a-1)\hspace{-1cm} & (\text{by def.\ of $\rpts$}) \\
&= \sum_{x'\in\supp(c_1(x))}c_1(x)(x')\cdot f_2(x')(a-1)\hspace{-1cm} & (\text{by IH})\\
&= \rpts\circ \Fpts f\circ c(x)(a) &  (\text{by def.\ of $\rpts$})\\
&=f_2(x)(a) & (\text{$f_2$ is a fixed point})\,.
\end{align*}

Therefore we have $f_1(x)(a)=f_2(x)(a)$ for each $x\in X$ and $a\in\mathbb{N}$.
Note that this implies that $f_x(x)(\infty)=f_2(x)(\infty)$ for each $x\in X$.
This concludes the proof.
\end{proof}

\begin{proof}[Proof (Prop.~\ref{prop:rankDomPTS})]
We prove that $(\rpts,\qpts,\leqRpts)$ satisfies the conditions in Def.~\ref{def:rankingdom}.

\subsubsection*{Cond.~\ref{item:def:rankingdom2}}
Let $(\Gamma,t)\in \Fpts \Rpts=\dist^2\Ninf\times\{0,1\}$. 
If $t=1$ then by the definitions of $\rpts$ and $\sigmapts$, we have:
\begin{displaymath}
\qpts\circ \rpts(\Gamma,t) = \qpts(\delta_0) = 1 = \sigmapts(\dist q_{\pts}(\Gamma),t)=\sigmapts\circ \Fpts \qpts(\Gamma,t)\,.
\end{displaymath}

Assume $t=0$. Then we have:
\allowdisplaybreaks[4]
\begin{align*}
& \qpts\circ \rpts(\Gamma,t) \\
&= \rpts(\Gamma,1)([0,\infty)) & (\text{by def.\ of $\qpts$}) \\
&= \sum_{a=0}^\infty r_{\pts}(\Gamma,1)(a) \\
&=  \sum_{a=1}^\infty\sum_{\gamma\in\supp(\Gamma)}\Gamma(\gamma)\cdot \gamma(a-1) & (\text{by def.\ of $\rpts$})\\
&=  \sum_{a=0}^\infty\sum_{\gamma\in\supp(\Gamma)}\Gamma(\gamma)\cdot \gamma(a) \\
&=  \sum_{\gamma\in\supp(\Gamma)}\Gamma(\gamma)\cdot \bigl(\sum_{a=0}^\infty\gamma(a)\bigr) \\
&= \sum_{\gamma\in\supp(\Gamma)}\Gamma(\gamma)\cdot q_{\pts}(\gamma) & (\text{by def.\ of $\qpts$})\\
&= \sigmapts\circ \Fpts q_{\pts}(\Gamma,t)& (\text{by def.\ of $\sigmapts$})\,.
\end{align*}
Hence Cond.~\ref{item:def:rankingdom2} is satisfied. 
Note also that we have $\qpts\circ \rpts=\sigmapts\circ \Fpts \qpts$.


\subsubsection*{Cond.~\ref{item:def:rankingdom4}}
We have already shown in Lem.~\ref{lem:tfSeqWellDefPTSD}.

\subsubsection*{Cond.~\ref{item:def:rankingdom3}}
We first prove that $\qpts$ is monotone.
Let $f_1,f_2:X\to \Rpts$ and assume that $f\leqRpts g$.
Then we have:
\begin{align*}
\qpts\circ f_1(x)
&= f_1(x)([0,\infty)) & (\text{by def.\ of $\qpts$})\\
&= \lim_{a\to\infty} f_1(x)([0,a]) \\ 
&\leq \lim_{a\to\infty} f_2(x)([0,a]) & (\text{by $f_1\leqRpts f_2$}) \\
&= f_2(x)([0,\infty)) \\
&= \qpts\circ f_2(x) & (\text{by def.\ of $\qpts$})\,.
\end{align*}
Hence 
we have $\qpts\circ f_1\leqOmegapts \qpts\circ f_2$,
and therefore Cond.~\ref{item:def:rankingdom3} was proved.

%
%
%
%
%
%
By the definition of $\qpts$, we have 
$\qpts(\delta_\infty)=0$.
Therefore 
$\qpts$ is strict.

We prove that $\qpts$ is 
 continuous.
By the definition of $\qpts$, it suffices to 
prove $(\bigsqcup_{\gamma\in K}\gamma)([0,\infty)) = \sup_{\gamma\in K}(\gamma([0,\infty)))$.
Note that the supremum on the left-hand side is taken with respect to $\leqRpts$ over $\Rpts$ while the latter is
taken with respect to the ordinary order over $[0,1]$. 
We have:
\begin{align*}
&\Bigl(\bigsqcup_{\gamma\in K}\gamma\Bigr)([0,\infty)) \\
&= \lim_{a\to\infty}(\bigsqcup_{\gamma\in K}\gamma)([0,a])  \\
&= \lim_{a\to\infty}\sup_{\gamma\in K}(\gamma([0,a])) & (\text{see the proof of Lem.~\ref{lem:orderRposetDNew}})\\
&= \sup_{\gamma\in K}\lim_{a\to\infty}(\gamma([0,a])) \\
&= \sup_{\gamma\in K}(\gamma([0,\infty)))\,.
\end{align*}
Therefore $\qpts$ is 
 continuous.

\subsubsection*{Cond.~\ref{item:def:rankingdom7}}
We have already shown in Lem.~\ref{lem:rCorecNewDPTSPTS}. 

Hence the triple $(\rpts,\qpts,\leqRpts)$ is a ranking domain. 
\end{proof}

\vspace{0.5\baselineskip}
\subsection{\textbf{Soundness of Additive Ranking Supermartingale, Categorically}}\label{subsec:soundARSapp}
In \S{}\ref{subsubsec:PTSasCoalg}, 
we have seen that the triple $(\rptspe,\qptsp,\leqRptsp)$, which captures the definition of $\varepsilon$-additive ranking 
supermartingales, is not a ranking domain because $\rptspe$ is not a corecursive algebra.
Hence we cannot prove soundness of $\varepsilon$-additive ranking supermartingales directly using our categorical framework.
In this section,
we show that 
its soundness is proved via that of distribution-valued ranking functions.
More concretely, we have the following proposition.

\begin{myproposition}\label{prop:rankDomPTSConnect}
We define triples $(\rpts,\qpts,\leqRpts)$ and
$(\rptspe,\qptsp,\leqRptsp)$ as in Prop.~\ref{prop:rankDomPTS} and Prop.~\ref{prop:rankDomPTSTradD} respectively.
Let $c:X\to \Fpts X$ be an $\Fpts$-coalgebra. 
Then we have:
\begin{align*}
&\exists b':X\to \Rptsp.\; b'\leqRptsp \Phi_{c,\rptspe}(b') \\
&\quad\Rightarrow\; \exists b:X\to \Rpts .\; \\
&\qquad\qquad\qquad b \leqRpts \Phi_{c,\rpts}(b),\;\;\text{and}\; \\
&\qquad\qquad\qquad \forall x\in X.\; (q'_{\pts}\circ b'(x)=1\,\Rightarrow q_{\pts}\circ b(x)=1)\,.
\end{align*}
\end{myproposition}

\noindent This proposition is an immediate corollary of the following lemma.

\begin{mylemma}\label{lem:prop:rankDomPTSConnectD}
We assume the conditions in Prop.~\ref{prop:rankDomPTSConnect}.
We define a function $\ppts:\Rpts\to \Rptsp$ by 
\begin{displaymath}
\ppts(\gamma)\;=\;
\varepsilon\cdot \sum_{a\in \supp(\gamma)}a\cdot \gamma(a)\,.
\end{displaymath}
Then the following statements hold. 
%
\begin{enumerate}
\item\label{item:lem:prop:rankDomPTSConnect1}
$\qptsp\circ \ppts\leq \qpts$ (here $\leq$ denotes the pointwise extension over $\Setsto{(\Rptsp)}{\Rpts}$). 


\item\label{item:lem:prop:rankDomPTSConnect2} 
Let $b':X\to \Rptsp$ and assume $b'\leqRptsp \Phi_{c,\rptspe}(b')$.
As $\rpts:\Fpts\Rpts\to \Rpts$ is corecursive (Prop.~\ref{prop:rankDomPTS}),
there exists a unique function $\uniquefp{c}_{\rpts}:X\to \Rpts$ such that 
$\uniquefp{c}_{\rpts}=\Phi_{c,\rpts}(\uniquefp{c}_{\rpts})$.
For this function, 
we have
$b'\leqRptsp \ppts\circ \uniquefp{c}_{\rpts}$\,.
\end{enumerate}
\begin{displaymath}
\small
 \begin{xy}
 \xymatrix@R=1.8em@C=2.6em{
 {\Fpts  X} \ar@{}[dr]|{=} 
 \ar[r]_{\Fpts \uniquefp{c}_{\rpts}} 
  \ar@/^.8em/[rr]^(.5){\Fpts b'}
 \ar@/^2.4em/[rrr]^(.85){\Fpts\,\sem{\mu\sigmapts}_c}
  & {\Fpts \Rpts} \ar[d]_{\rpts}   \ar@/^.8em/[rr]^(.5){\Fpts\qpts} 
  \ar[r]_{\Fpts \ppts}
  & {\Fpts \Rptsp} \ar[d]_{\rptspe} \ar@{}[dr]|{\sqsubseteq} 
 \ar[r]_{\Fpts \qptsp}  & {\Fpts \Omegapts} \ar[d]_{\sigmapts} \\
 {X} \ar[u]_{c}  \ar[r]^{\uniquefp{c}_{\rpts}} \ar@/_.8em/[rr]_(.5){b'}
 \ar@/_2.2em/[rrr]_(.85){\sem{\mu\sigmapts}_c} & {\Rpts} \ar[r]^{\ppts} 
 \ar@/_.8em/[rr]_(.5){\qpts}
 & {\Rptsp} \ar[r]^{\qptsp} & {\Omegapts} 
 }
 \end{xy}
\end{displaymath}
\end{mylemma}

\begin{mysublemma}\label{sublem:1611301449}
For the function $\ppts:\Rpts\to \Rptsp$ in Lem.~\ref{lem:prop:rankDomPTSConnectD}, 
we have the following.
\begin{equation}\label{eq:keylem:prop:rankDomPTSConnect}
\ppts\circ \rpts=\rptspe\circ \Fpts \ppts
\end{equation}
\end{mysublemma}

%
%
%
\begin{proof}
Let $(\Gamma,t)\in \Fpts \Rpts=\dist^2\Ninf\times\{0,1\}$. 
If $t=1$, then by the definitions of $\rpts$ and $\rptspe$, we have 
\begin{displaymath}
\ppts\circ \rpts(\Gamma,t)\,=\,\rptspe\circ \Fpts \ppts(\Gamma,t)\,=\,0\,.
\end{displaymath}

Let $t=0$. 
Then we have:
\allowdisplaybreaks[4]
\begin{align*}
&\ppts\circ \rpts(\Gamma,t) \\
&= \varepsilon\cdot \sum_{a\in\Ninf\setminus\{0\}} a\cdot \rpts(\Gamma,t)(a)  \hspace{-2.5cm}& (\text{by def.\ of $\ppts$}) \\
&=\varepsilon\cdot\!\!\sum_{a\in\Ninf\setminus\{0\}}\!\! a\cdot\!\!\sum_{\gamma\in\supp(\Gamma)}\!\!\Gamma(\gamma)\cdot\gamma(a-1) \hspace{-2.5cm} & (\text{by def.\ of $\rpts$}) \\
&=\varepsilon\cdot\!\!\sum_{a\in\Ninf}\!\! (a+1)\cdot\sum_{\gamma\in\supp(\Gamma)}\Gamma(\gamma)\cdot\gamma(a)  \hspace{-2.5cm}& \\
&=\varepsilon\cdot\sum_{a\in\Ninf} a\cdot\sum_{\gamma\in\supp(\Gamma)}\Gamma(\gamma)\cdot\gamma(a) \hspace{-2.5cm}& \\
&\qquad+\varepsilon\cdot\sum_{a\in\Ninf}\sum_{\gamma\in\supp(\Gamma)}\Gamma(\gamma)\cdot\gamma(a) \hspace{-2.5cm}& \\
&=\sum_{a\in\Ninf}\Bigl(  a\cdot\varepsilon\cdot\!\!\!\!\sum_{\gamma\in\supp(\Gamma)}\!\!\!\!\Gamma(\gamma)\cdot\gamma(a)\Bigr)+\varepsilon  \hspace{-.0cm}& \\
&=\!\!\!\!\sum_{\gamma\in\supp(\Gamma)}\!\!\!\!\Bigl(\Gamma(\gamma)\cdot\varepsilon\cdot\sum_{a\in\Ninf}a\cdot\gamma(a)\Bigr)+\varepsilon  \hspace{-.0cm}& \\
&=\!\!\!\!\sum_{\gamma\in\supp(\Gamma)}\!\!\!\!\Bigl(\Gamma(\gamma)\cdot\ppts(\gamma)\Bigr)+\varepsilon  \hspace{-.0cm}
&\text{(by def.\ of $\ppts$)} \\
%
%
&=\!\!\!\!\!\!\sum_{b\in\supp(\dist\ppts(\Gamma))}\!\!\!\!\!\bigl(b\cdot\! \!\!\sum_{\gamma'\in\ppts^{-1}(b)}\!\!\Gamma(\gamma') \bigr)+\varepsilon  \hspace{-2.5cm}& \\
&=\sum_{b\in\supp(\dist\ppts(\Gamma))}\bigl(b\cdot \dist\ppts(\Gamma)(b) \bigr)+\varepsilon  \hspace{-2.5cm} & \\
&=\rptspe\circ \Fpts \ppts(\Gamma,t) & \text{(by def.\ of $\rptspe$)}\,.
\end{align*}
This concludes the proof.
\end{proof}

\begin{mysublemma}\label{sublem:reversedCondSatisfied}
The $\Fpts$-algebra $\rptspe:\Fpts \Rptsp\to \Rptsp$ in Prop.~\ref{prop:rankDomPTSTradD} satisfies 
the following conditions.
\begin{enumerate}
%
\item\label{asm:behDomain4and3Dual}
The poset $(\Rptsp,\leqRptsp)$ 
is a complete lattice.
%
%
%

\item\label{asm:behDomain1and2Dual}
For each $\Fpts$-coalgebra $c:X\to \Fpts X$, the function $\Phi_{c,\rptspe}:\Setsto{X}{\Rptsp}\to\Setsto{X}{\Rptsp}$ in
Def.~\ref{def:lfpsem} is monotone with respect to the pointwise extension of $\leqRptsp$. 
\end{enumerate}
Moreover the triple 
$(\rpts:\Fpts\Rpts\to \Rpts,\ppts:\Rpts\to \Rptsp,\leqRpts)$ (see Prop.~\ref{prop:rankDomPTS} and 
Lem.~\ref{lem:prop:rankDomPTSConnectD}) satisfies 
the following conditions.
\begin{enumerate}
\setcounter{enumi}{2}

\item\label{item:def:rankingdom2Dual}
We have $\ppts\circ \rpts\geqRptsp \rptspe\circ \Fpts \ppts$ between arrows $\Fpts\Rpts\to \Rptsp$.

\item\label{item:def:rankingdom4Dual}
The following conditions are satisfied by
$\rpts$: 
\begin{enumerate}
\item\label{item:def:rankingdom40Dual}
the poset $(\Rpts,\leqRpts)$ has the greatest element $\topRpts$;

\item\label{item:def:rankingdom43Dual}
the poset $(\Rpts,\leqRpts)$ is $\omega^\op$-complete (i.e.\ for each decreasing sequence
$\varphi_0\geqRpts \varphi_1\geqRpts\cdots$ in $\Rpts$, the infimum $\bigsqcap_{i\in\omega}\varphi_i$ exists);

\item\label{item:def:rankingdom41and2NewDual}
for each $c:X\to \Fpts X$, the function $\Phi_{c,\rpts}:\Setsto{X}{(\Rpts)}\to\Setsto{X}{(\Rpts)}$ 
(Def.~\ref{def:Phicsigma}) is monotone and moreover $\omega^\op$-continuous, 
i.e.\ for each decreasing sequence $f_0\geqRpts f_1\geqRpts\cdots$ in $\Setsto{X}{(\Rpts)}$ wrt.\ the pointwise extension of $\leqRpts$, we have
$\Phi_{c,\rpts}\bigl(\bigsqcap_{i\in\omega}f_i\bigr)\;=\;\bigsqcap_{i\in\omega}\Phi_{c,\rpts}(f_i)$\,.



\end{enumerate}

\item\label{item:def:rankingdom3Dual}
The function $\ppts:\Rpts\to \Rptsp$ is monotone (i.e.\ $a\leqRpts b\,\Rightarrow \ppts(a)\leqRptsp \ppts(b)$),
top-preserving (i.e.\ $\ppts(\topRpts)=\topRptsp$) and
$\omega^\op$-continuous (i.e.\ 
each decreasing sequence $\varphi_0\geqRpts \varphi_1\geqRpts\cdots$ in $\Rpts$,
we have $\ppts(\bigsqcap_{i\in\omega}\varphi_i)=\bigsqcap_{i\in\omega}\ppts(\varphi_i)$). 
Note that in the last equality, the infimum in the left-hand side is take wrt.\ $\leqRpts$ while the one in
the right-hand side is taken wrt.\ $\leqRptsp$.

\item\label{item:def:rankingdom7Dual}
The algebra $\rpts:\Fpts\Rpts\to \Rpts$ is corecursive.
\end{enumerate}

\begin{displaymath}
 \begin{xy}
 \xymatrix@R=2.0em@C=2.6em{
 {\Fpts X} 
 \ar[r]_{\Fpts b} 
 \ar@/^.8em/[rr]^{\Fpts b'}
 & {\Fpts \Rpts} \ar[d]_{\rpts} \ar@{}[dr]|{\sqsupseteq} 
 \ar[r]_{\Fpts \ppts}  & {\Fpts \Rptsp} \ar[d]_{\rptspe} \\
 {X} \ar[u]_{c}  \ar[r]^{b} \ar@/_.8em/[rr]_{b'}  & {\Rpts} \ar[r]^{\ppts} & {\Rptsp} 
 }
 \end{xy}
\end{displaymath}
\end{mysublemma}

\begin{proof}\mbox{}\nobreak
%
\subsubsection*{\ref{asm:behDomain4and3Dual}} Easy. 

\subsubsection*{\ref{asm:behDomain1and2Dual}} 
Already proved in Prop.~\ref{prop:pseudoRD}.

\subsubsection*{\ref{item:def:rankingdom2Dual}}
Immediate from Sublem.~\ref{sublem:1611301449}.

\subsubsection*{\ref{item:def:rankingdom4Dual}} 
It is easy to see that the Dirac distribution $\delta_0$ concentrated at $0$ is the greatest element
in $(\Rpts,\leqRpts)$.
Hence Cond.~\ref{item:def:rankingdom40Dual} is satisfied.
Cond.~\ref{item:def:rankingdom43Dual} is immediate from that 
$(\Setsto{X}{(\Rpts)},\leqRpts)$ is a complete lattice (Lem.~\ref{lem:orderRposetDNew}). 

We prove that Cond.~\ref{item:def:rankingdom41and2NewDual} is satisfied.
Monotonicity of $\Phi_{c,\rpts}$ is already proved in Prop.~\ref{prop:rankDomPTS}.
We prove that $\Phi_{c,\rpts}$ is $\omega^\op$-continuous.
To this end, it suffices to prove the following equality for each $x\in X$.
\begin{equation}\label{eq:omegacontPhipts}
\Phi_{c,\rpts}\bigl(\bigsqcap_{i\in\omega}f_i\bigr)(x)\;=\;\bigsqcap_{i\in\omega}\bigl(\Phi_{c,\rpts}(f_i)(x)\bigr)
\end{equation}
Note that $\bigsqcap$ on the right-hand side denotes the infimum with respect to $\leqRpts$ while
that on the left-hand side denotes the infimum with respect to its pointwise extension.
Let 
$(\delta,t)=c(x)$.

If $t=1$, then by the definition of $\rpts$, we have 
$\Phi_{c,\rpts}(f)(x)=\delta_0$ for each $f:X\to \Rpts$.
Hence we have (\ref{eq:omegacontPhipts}).

If $t=0$, 
for each 
$a\in\mathbb{N}$, we have:
\allowdisplaybreaks[4]
\begin{align*}
& \Bigl(\bigsqcap_{i\in\omega}\bigl(\Phi_{c,\rpts}(f_i)(x)\bigr)\Bigr)([0,a])\hspace{-10cm}& \\
&= \inf_{i\in\omega}\bigl(\Phi_{c,\rpts}(f_i)(x)([0,a])\bigr) \hspace{-10cm}& (\text{see the proof of Lem.~\ref{lem:orderRposetDNew}}) \\
&= \inf_{i\in\omega}\sum_{x\in\supp(\delta)}\delta(x)\cdot \bigl(f_i(x)([0,a-1])\bigr)  \hspace{-10cm} &
  (\text{by def.\ of $\rpts$})\\
&= \sum_{x\in\supp(\delta)}\delta(x)\cdot\bigl(\inf_{i\in\omega}f_i(x)([0,a-1])\bigr) \hspace{-3.5cm}&\\
& &(\text{by the dominated convergence theorem})\\
&= \sum_{x\in\supp(\delta)}\delta(x)\cdot\bigl(\bigsqcap_{i\in\omega}f_i\bigr)(x)([0,a-1])\hspace{-10cm} &\\
& & (\text{see the proof of Lem.~\ref{lem:orderRposetDNew}}) 
\\
&= \Phi_{c,\rpts}\bigl(\bigsqcap_{i\in\omega}f_i\bigr)(x)([0,a]) \hspace{-10cm}
& (\text{by def.\ of $\Phi_{c,r_{\pts}}$})\,.
\end{align*}
By Sublem.~\ref{sublem:LSmeasDNew}, this implies 
(\ref{eq:omegacontPhipts}).

\subsubsection*{\ref{item:def:rankingdom3Dual}} 
We first prove that $\ppts:\Rpts\to\Rptsp$ is monotone. Let $\varphi_1,\varphi_2\in\Rpts$ and assume that 
$\varphi_1\leqRpts \varphi_2$.

If $\varphi_2(\infty)=C>0$, then we have $\varphi_2([0,a])\leq 1-C$ for each $a\in\mathbb{N}$.
By the definition of $\leqRpts$, this implies $\varphi_1([0,a])\leq 1-C$ for each $a$, and therefore
we have $\varphi_1(\infty)>0$.
Hence in this case, we have $\ppts(\varphi_1)=\ppts(\varphi_2)=\infty$.

Let $\varphi_2(\infty)=0$.
Then we have:
\begin{align*}
&\ppts(\varphi_2) \\
&= \varepsilon\cdot\sum_{a=1}^{\infty} a\cdot\varphi_2(a) & (\text{by def.\ of $\ppts$}) \\
&= \lim_{n\to\infty} \varepsilon\cdot\sum_{a=1}^{n} \bigl(a\cdot\varphi_2([0,a])-a\cdot\varphi_2([0,a-1])\bigr)\hspace{-10cm} & \\
&= \lim_{n\to\infty} \varepsilon\cdot\Bigl(n\cdot\varphi_2([0,n])-\sum_{a=1}^{n}1\cdot\varphi_2([0,a-1])\Bigr) \hspace{-4cm}& \\
&= \lim_{n\to\infty} \varepsilon\cdot\sum_{a=1}^{n}\varphi_2([a,n]) \hspace{-10cm}& \\
&= \varepsilon\cdot\sum_{a=1}^{\infty}\varphi_2([a,\infty]) \hspace{-10cm}& \\
&= \varepsilon\cdot\sum_{a=1}^{\infty}\bigl(1-\varphi_2([0,a-1])\bigr) \hspace{-10cm}& \\
&\leq \varepsilon\cdot\sum_{a=1}^{\infty}\bigl(1-\varphi_1([0,a-1])\bigr) \hspace{-10cm}& (\text{by $\varphi_1\leqRpts \varphi_2$}) \\
&=\ppts(\varphi_1) & (\text{as the transformation for $\varphi_2$})\,.
\end{align*}
Hence we have $\ppts(\varphi_1)\leqRptsp\ppts(\varphi_2)$ and therefore $\ppts$ is monotone.

By the definition of $\ppts$, $\ppts(\delta_0)=0$. Hence $\ppts$ is top-preserving.

We prove that $\ppts$ is $\omega^\op$-continuous.
Let $\varphi_0\geqRpts \varphi_1\geqRpts\cdots$ be a decreasing sequence in $\Rpts$.
%

Assume $\Bigl(\bigsqcap_{i\in\omega}\bigl(\varphi_i\bigr)\Bigr)(\infty)=C>0$.
Then for all $b\in\mathbb{N}$, there exists $i_b\in\omega$ such that 
$1-\varphi_i([0,b])\geq \frac{C}{2}$ (c.f.\ the proof of Lem.~\ref{lem:orderRposetDNew}).
Therefore for each $b\in\mathbb{N}$, we have:
\begin{align*}
\bigsqcap_{i\in\omega}\ppts(\varphi_i) 
&\geq\ppts(\varphi_{i_b}) \\
&= \varepsilon\cdot\sum_{a=0}^{\infty}a\cdot \varphi_i(a)  & (\text{by def.\ of $\ppts$}) \\
&\geq\varepsilon\cdot b\cdot (1-\varphi([0,b])) \\ 
&\geq\varepsilon\cdot b\cdot (1-\frac{C}{2}))\,.
\end{align*}
%
Hence we have
\begin{displaymath}
{\textstyle
\ppts\Bigl(\bigsqcap_{i\in\omega}\bigl(\varphi_i\bigr)\Bigr)\;=\;\infty\;=\;
\bigsqcap_{i\in\omega}\ppts(\varphi_i)}\,.
\end{displaymath}

Assume 
$\Bigl(\bigsqcap_{i\in\omega}\bigl(\varphi_i\bigr)\Bigr)(\infty)=0$.
Then we have:
\allowdisplaybreaks[4]
\begin{align*}
&\ppts\Bigl(\bigsqcap_{i\in\omega}\bigl(\varphi_i\bigr)\Bigr) \\
&= \varepsilon\cdot\sum_{a=0}^\infty a\cdot\Bigl(\bigsqcap_{i\in\omega}\varphi_i\Bigr)(a) \hspace{-2cm}  & (\text{by def.\ of $\ppts$}) \\
&= \varepsilon\cdot\sum_{a=1}^\infty \Bigl(1-\bigl(\bigsqcap_{i\in\omega}\varphi_i\bigr)([0,a-1])\Bigr) \hspace{-10cm}&\\
& & (\text{as in the proof of monotonicity of $\ppts$}) \\
&= \varepsilon\cdot\sum_{a=1}^\infty \Bigl(1-\inf_{i\in\omega}\bigl(\varphi_i([0,a-1])\bigr)\Bigr) \hspace{-10cm}&\\
& & (\text{see the proof of Lem.~\ref{lem:orderRposetDNew}}) \\
&= \sup_{i\in\omega}\varepsilon\cdot\sum_{a=1}^\infty \Bigl(1-\bigl(\varphi_i([0,a-1])\bigr)\Bigr) \hspace{-10cm}&\\
&= \sup_{i\in\omega}\ppts(\varphi_i) \hspace{-10cm} & (\text{as in the above}) \\
&= \bigsqcap_{i\in\omega}\ppts(\varphi_i) \hspace{-10cm} & (\text{by def.\ of $\leqRptsp$}) \\
\end{align*}
Hence we have $\omega^\op$-continuity.

\subsubsection*{\ref{item:def:rankingdom7Dual}} 
It is already proved in Lem.~\ref{lem:rCorecNewDPTSPTS}.
\end{proof}

\begin{proof}[Proof (Lem.~\ref{lem:prop:rankDomPTSConnectD})] \mbox{}
\subsubsection*{\ref{item:lem:prop:rankDomPTSConnect1}}
For each $\varphi\in \Rpts$, we have:
\begin{align*}
&\qptsp\circ \ppts(\varphi) \\
&= 
\qptsp\bigl(\sum_{a\in \mathbb{N}}a\cdot \varphi(a)\bigr)
& (\text{by def.\ of $\ppts$}) \\
%
&= \begin{cases}
1 & (\sum_{a\in \mathbb{N}}a\cdot \varphi(a)<\infty) \\
0 & (\text{otherwise})
\end{cases} & (\text{by def.\ of $\qptsp$}) \\
&\leq \varphi([0,\infty)) \\
&= \qpts(\varphi) &  (\text{by def.\ of $\qpts$})\,.
\end{align*}
Hence Cond.~\ref{item:lem:prop:rankDomPTSConnect1} is satisfied.

\subsubsection*{\ref{item:lem:prop:rankDomPTSConnect2}}
\begin{displaymath}
 \begin{xy}
 \xymatrix@R=2.0em@C=2.6em{
 {\Fpts X} \ar@{}[dr]|{=} \ar[r]_{\Fpts \uniquefp{c}_{\rpts}} 
 \ar@/^.8em/[rr]^{\Fpts\,\llbracket  \nu r'_{\pts,\varepsilon}\rrbracket_c}
 & {\Fpts \Rpts} \ar[d]_{r_{\pts}} \ar@{}[dr]|{\sqsupseteq} 
 \ar[r]_{\Fpts p}  & {\Fpts \Rptsp} \ar[d]_{r'_{\pts,\varepsilon}} \\
 {X} \ar[u]_{c}  \ar[r]^{\uniquefp{c}_{\rpts}} \ar@/_2.5em/[rr]_{b'} \ar@/_.8em/[rr]_{\llbracket \nu r'_{\pts,\varepsilon} \rrbracket_c} & {\Rpts} \ar[r]^{p} & {\Rptsp} 
 }
 \end{xy}
\end{displaymath}
By Sublem.~\ref{sublem:reversedCondSatisfied}.\ref{asm:behDomain4and3Dual}--\ref{asm:behDomain1and2Dual},
the $\Fpts$-modality $\rptspe:\Fpts\Rptsp\to \Rptsp$ satisfy dual conditions of Asm.~\ref{asm:behDomainrank}.
Hence using the dual statement of Prop.~\ref{prop:sigmaLFP}, we can show that 
$\Phi_{c,\rptspe}:\Setsto{X}{\Rptsp}\to\Setsto{X}{\Rptsp}$ has the greatest fixed point 
$\sem{\nu \rptspe}_c:X\to \Rptsp$ with respect to the pointwise extension of $\leqRptsp$.
By the Knaster-Tarski theorem, we have
\begin{equation}\label{eq:1612011458}
b'\leqRptsp \sem{\nu\rptspe}_c\,.
\end{equation}

By Sublem.~\ref{sublem:reversedCondSatisfied}.\ref{item:def:rankingdom2Dual}--\ref{item:def:rankingdom7Dual},
the triple $(\rpts:\Fpts\Rpts\to \Rpts,\ppts:\Rpts\to \Rptsp,\leqRptsp)$ satisfies the dual conditions of 
the axioms of a ranking domain (Def.~\ref{def:rankingdom}), except that the length of a transfinite sequence in
$\Setsto{X}{(\Rpts)}$
 is restricted to $\omega$.
Note here that $\uniquefp{c}_{\rpts}:X\to \Rpts$ satisfies $\uniquefp{c}_{\rpts}\geqRpts \Phi_{c,\rpts}(\uniquefp{c}_{\rpts})$ by its definition.
Therefore in a similar manner to the proof of Thm.~\ref{thm:soundnessranking},
we can prove 
\begin{equation}\label{eq:1612011457}
\sem{\nu \rptspe}_c \leqRptsp \ppts\circ \uniquefp{c}_{\rpts}\,.
\end{equation}
By (\ref{eq:1612011457}) and (\ref{eq:1612011458}),
we have $b' \leqRptsp \ppts\circ \uniquefp{c}_r$.
\end{proof}

\begin{proof}[Proof (Prop.~\ref{prop:rankDomPTSConnect})]
Immediate from Lem.~\ref{lem:prop:rankDomPTSConnectD}.
\auxproof{
We have:
\begin{align*}
&\exists b':X\to \Rptsp.\; b'\leqRptsp \Phi_{c,\rptspe}(b')\hspace{-.5cm} \\
&\Rightarrow\; \exists b:X\to \Rpts.\; \\
&\qquad\quad b=\Phi_{c,\rpts}(b),\, \;\;\text{and}\;\; \\
&\qquad\quad 
\qptsp\circ b'\leqRpts \ppts\circ \uniquefp{c}_{\rpts}
& (\text{by Lem.~\ref{lem:prop:rankDomPTSConnectD}.\ref{item:lem:prop:rankDomPTSConnect2}})\\
&\Rightarrow\; \exists b:X\to \Rpts.\; \\
&\qquad\quad b=\Phi_{c,\rpts}(b),\, \;\;\text{and}\;\; \\
&\qquad\quad 
\qptsp\circ b'\leq \qpts\circ b
& (\text{by Lem.~\ref{lem:prop:rankDomPTSConnectD}.\ref{item:lem:prop:rankDomPTSConnect1}})\\
&\Rightarrow\; \exists b:X\to \Rpts.\; \\
&\qquad\quad b\leqRpts \Phi_{c,\rpts}(b),\, \;\;\text{and}\;\; \\
&\qquad\quad \forall x\in X.\, (\qptsp\circ b'(x)=1\,\Rightarrow\, \qpts\circ b(x)=1)\,. \hspace{-2.5cm}&
\qed
\end{align*}
}
%
\end{proof}

We can now prove the soundness of additive ranking supermartingale (Thm.~\ref{thm:soundnessRankConv}) 
using our categorical framework as follows.
\begin{align*}
&\exists b'.\,
\text{$b'$ is an $\varepsilon$-additive ranking supermartingale}
\;\text{and}\; b'(x)\neq\infty\\ 
&\overset{\text{Prop.~\ref{prop:rankDomPTSTradD}}}{\Leftrightarrow} 
\exists b'.\,
b'\leqRptsp \Phi_{c,\rptspe}(b')
\text{ and } 
q'_{\pts}\circ b'(x)=1 \\
& \overset{\text{Prop.~\ref{prop:rankDomPTSConnect}}}{\Rightarrow} 
 \exists b.\,\text{$b$ is a ranking arrow wrt.\ $(\rpts,\qpts,\leqRpts)$},\;\text{and}\; \\
& \qquad\qquad\qquad\qquad\qquad\qquad\qquad\qquad\qquad q_{\pts}\circ b(x)=1 \\
& \overset{\text{Thm.~\ref{thm:soundnessranking}}}{\Rightarrow} \llbracket\mu\sigmapts\rrbracket_{c^{G,\Acc}} (x)=1 \\
&\overset{\text{Prop.~\ref{prop:constBehSituationPTS}}}{\Leftrightarrow}  \Reach_{M,\Acc}(x)=1
\end{align*}

\vspace{0.5\baselineskip}
\subsection{\textbf{Soundness of Multiplicative Ranking Supermartingale, Categorically}}\label{subsec:soundMRSapp}

\begin{myproposition}\label{prop:rankDomPTSConnectMult}
We define triples $(\rpts,\qpts,\leqRpts)$ and
$(\rptspa,\qptsp,\leqRptspm)$ as in Prop.~\ref{prop:rankDomPTS} and Prop.~\ref{prop:rankDomPTSTradDMult} respectively.
Let $c:X\to \Fpts X$ be an $\Fpts$-coalgebra. 
Then we have:
\begin{align*}
&\exists b'':X\to \Rptspm.\; b''\leqRptspm \Phi_{c,\rptspa}(b'') \\
&\quad\Rightarrow\; \exists b:X\to \Rpts .\; \\
&\qquad\qquad\qquad b \leqRpts \Phi_{c,\rpts}(b),\;\;\text{and}\; \\
&\qquad\qquad\qquad \forall x\in X.\; (q'_{\pts}\circ b'(x)=1\,\Rightarrow q_{\pts}\circ b(x)=1)\,.
\end{align*}
\end{myproposition}

\noindent This proposition is an immediate corollary of Prop.~\ref{prop:rankDomPTSConnect} and the following lemma.

\begin{mylemma}\label{lem:rankDomPTSConnectMult2}
We define triples $(\rptspe,\qptsp,\leqRptsp)$ and
$(\rptspa,\qptsp,\leqRptsp)$ as in Prop.~\ref{prop:rankDomPTSTradD} and Prop.~\ref{prop:rankDomPTSTradDMult} respectively.
We define a function $\pptsp:\Rptspm\to\Rptsp$ as follows.
\begin{equation*}
\pptsp(a)=\varepsilon\cdot \Biggl(\Bigl(\log_{\frac{1}{\alpha}}\frac{a}{\delta}\Bigr)+1\Biggr)
\end{equation*}
Let $c:X\to \Fpts X$ be an $\Fpts$-coalgebra. 
Then we have:
\begin{align*}
&\exists b'':X\to \Rptspm.\; b''\leqRptspm \Phi_{c,\rptspa}(b'') \\
&\quad \Rightarrow\; \pptsp\circ b'' \leqRptsp \Phi_{c,\rptspe}(\pptsp\circ b'')
\end{align*}
\end{mylemma}


\noindent
\begin{minipage}{0.6\hsize}
\begin{mysublemma}\label{sublem:rankDomPTSConnectMult3}
For the function $\pptsp:\Rptspm\to\Rptsp$ in Lem.~\ref{lem:rankDomPTSConnectMult2},
we have $\pptsp\circ \rptspa\leqRptsp \rptspe\circ \Fpts\pptsp$\,.
\end{mysublemma}
\end{minipage}
\begin{minipage}{0.4\hsize}
\scriptsize
 \begin{xy}
 \xymatrix@R=1.2em@C=1.4em{
 {F \Rptspm} \ar@{}[dr]|{\sqsubseteq} \ar[d]_{\rptspa} \ar@{->}[r]^{F \pptsp} 
 & {F \Rptsp}  \ar[d]_{\rptspe} \\
 {\Rptspm}   \ar@{->}[r]_{\pptsp}
  & {\Rptsp}} 
 \end{xy}
\end{minipage}

\begin{proof}
Let $(\psi,t)\in\Fpts\Rptspm=\dist\Rptspm\times\{0,1\}$.
By the definition of $\leqRptsp$, it suffices to prove 
$\pptsp\circ \rptspa(\psi,t)\geq \rptspe\circ \Fpts\pptsp(\psi,t)$.

If $t=1$, then we have:
\begin{align*}
&\pptsp\circ \rptspa(\psi,t) \\
&= \pptsp(\alpha\delta) & (\text{by def.\ of $\rptspa$})\\
&= 0 & (\text{by def.\ of $\pptsp$}) \\
&= \rptspe(\dist\pptsp(\psi),t) & (\text{by def.\ of $\rptspe$}) \\
&= \rptspe\circ \Fpts\pptsp(\psi,t)\,. &
\end{align*}

Let $t=0$. Note that $\sum_{a\in\supp(\psi)}\psi(a)=1$.
Hence we have:
\begin{align*}
&\pptsp\circ \rptspa(\psi,t) \\
&= \pptsp\Bigl(\frac{1}{\alpha}\cdot\bigl(\sum_{a\in\supp(\psi)}a\cdot \psi(a) \bigr)\Bigr) & (\text{by def.\ of $\rptspa$})\\
&= \varepsilon\cdot\Biggl(\log_{\frac{1}{\alpha}}\frac{1}{\delta}\Bigl(\frac{1}{\alpha}\cdot\bigl(\sum_{a\in\supp(\psi)}a\cdot \psi(a) \bigr)\Bigr)+1\Biggr) \hspace{-2.5cm}& \\
& & (\text{by def.\ of $\pptsp$})\\
&= \varepsilon\cdot\Biggl(\Bigl(\log_{\frac{1}{\alpha}}\bigl(\sum_{a\in\supp(\psi)}\frac{a}{\delta}\cdot \psi(a) \bigr)\Bigr)+2 \Biggr)\hspace{-2cm}& \\
&\geq \varepsilon\cdot\Biggl(\Bigl(\sum_{a\in\supp(\psi)}(\log_{\frac{1}{\alpha}}\frac{a}{\delta})\cdot \psi(a)\Bigr)+2\Biggr)\hspace{-2cm} & \\
& &\hspace{-3.8cm}(\text{$\log_{\frac{1}{\alpha}}(\place)$ is an upward convex function}) \\
&= \Biggl(\sum_{a\in\supp(\psi)}\Bigl(\varepsilon\cdot\bigl((\log_{\frac{1}{\alpha}}\frac{a}{\delta})+1\bigr)\Bigr)\cdot \psi(a)\Biggr)+\varepsilon \hspace{-2.5cm}& \\
&= \rptspe\circ \Fpts\pptsp(\varphi,t) & (\text{by def.\ of $\pptsp$ and $\rptspe$})\,.
\end{align*}
This concludes the proof.
\end{proof}

\begin{proof}[Proof (Lem.~\ref{lem:rankDomPTSConnectMult2})]
\begin{equation*}
 \begin{xy}
 \xymatrix@R=1.6em@C=1.8em{
{FX}  \ar@{}[dr]|{\rotatebox{90}{$\sqsubseteq$}} \ar@{->}[r]^{F b''} &
 {F \Rptspm} \ar@{}[dr]|{\sqsubseteq} \ar[d]|{\smash{\rptspa}\raisebox{-1.5mm}{\vphantom{0.5mm}}} \ar@{->}[r]^{F \pptsp} 
 & {F \Rptsp}  \ar[d]_{\rptspe} \\
 {X} \ar[r]_{b''} \ar[u]_{c} &
 {\Rptspm}   \ar@{->}[r]_{\pptsp}
  & {\Rptsp}} 
 \end{xy}
 \end{equation*}
Let $b'':X\to \Rptspm$ and assume that  $b''\leqRptspm \Phi_{c,\rptspa}(b'')$.
Then we have:
\begin{align*}
&\pptsp\circ b'' \\
&\leqRptsp \pptsp\circ \rptspa\circ \Fpts b''\circ c & (\text{by the assumption})\\
&\leqRptsp \rptspe\circ \Fpts\pptsp\circ \Fpts b''\circ c & (\text{by Sublem.~\ref{sublem:rankDomPTSConnectMult3}})\\
&= \Phi_{c,\rptspe}(\pptsp\circ b'')  & (\text{by definition}) \,.
\end{align*}
This concludes the proof.
\end{proof}

\begin{proof}[Proof (Prop.~\ref{prop:rankDomPTSConnectMult})]
Immediate from Lem.~\ref{lem:rankDomPTSConnectMult2} and Prop.~\ref{prop:rankDomPTSConnect}.
\end{proof}

We can now prove the soundness of multiplicative ranking supermartingale 
using our categorical framework as follows.
\begin{align*}
&\exists b''.\,
\text{$b''$ is an $\alpha$-multiplicative ranking supermartingale},\;\text{and}\; \\
& \qquad\qquad\qquad\qquad\qquad\qquad\qquad\qquad\qquad b''(x)\neq\infty\\ 
&\overset{\text{Prop.~\ref{prop:rankDomPTSTradDMult}}}{\Leftrightarrow} 
\exists b'.\,
b'\leqRptsp \Phi_{c,\rptspe}(b')
\text{ and } 
q'_{\pts}\circ b'(x)=1 \\
& \overset{\text{Prop.~\ref{prop:rankDomPTSConnect}}}{\Rightarrow} 
 \exists b.\,\text{$b$ is a ranking arrow wrt.\ $(\rpts,\qpts,\leqRpts)$},\;\text{and}\; \\
& \qquad\qquad\qquad\qquad\qquad\qquad\qquad\qquad\qquad q_{\pts}\circ b(x)=1 \\
& \overset{\text{Thm.~\ref{thm:soundnessranking}}}{\Rightarrow} \llbracket\mu\sigmapts\rrbracket_{c^{G,\Acc}} (x)=1 \\
&\overset{\text{Prop.~\ref{prop:constBehSituationPTS}}}{\Leftrightarrow}  \Reach_{M,\Acc}(x)=1
\end{align*}

%

\vspace{0.5\baselineskip}
\subsection{\textbf{Proof of Prop.~\ref{prop:weakrankDomPTS}}}
We prove that $\rptswg$ is a corecursive algebra separately.

\begin{mylemma}\label{lem:tfSeqWellDefPTSDWeak}
The algebra $\rptswg:\Fpts \Rptsw\to \Rptsw$ in Prop~\ref{prop:weakrankDomPTS} 
satisfies 
Cond.~\ref{item:def:rankingdom4} in Def.~\ref{def:rankingdom}.
\end{mylemma}

\begin{proof}
%
It is easy to see that Cond.~\ref{item:def:rankingdom40and3} is satisfied.

We prove that Cond.~\ref{item:def:rankingdom41and2} is satisfied.
Let $b_1,b_2:X\to \Rptsw$ and assume that $b_1\leqRptsw b_2$.
Let $x\in X$ and assume that $c(x)=(\delta,t)\in \Fpts X=\dist X\times\{0,1\}$.
%
%
Then we have:
\begin{align*}
&\rptswg\circ \Fpts b_1(\delta,t) \\
&= \begin{cases}
1 & (t=1) \\
\sum_{x\in X} \delta(x)\cdot b_1(x)
& (t=0)
\end{cases} \hspace{-1cm}& \\
& &(\text{by def.\ of $\Fpts$ and $\rptswg$}) \\
&= \begin{cases}
1 & (t=1) \\
\sum_{x\in X} \delta(x)\cdot b_2(x)
& (t=0)
\end{cases} \hspace{-1cm}& (\text{by $b_1\leqRptsw b_2$}) \\
&= \rptswg\circ \Fpts b_2(\delta,t) &(\text{by def.\ of $\rptswg$ and $\Fpts$})\,.
\end{align*}
Therefore Cond.~\ref{item:def:rankingdom41and2} is satisfied.
%
\end{proof}

\begin{mylemma}\label{lem:corecAlgProb}
The algebra $\rptswg:\Fpts\Rptsw\to \Rptsw$ in Prop.~\ref{prop:weakrankDomPTS} is corecursive.
\end{mylemma}

\begin{proof}
Let $c:X\to \Fpts X$ be an $\Fpts$-coalgebra.
For each $x\in X$ such that $c(x)=(\delta,t)\in\Fpts X$, we write $c_1(x)$ and $c_2(x)$ for
$\delta\in\dist\Rptsw$ and $t\in\{0,1\}$ respectively.
We prove that $\Phi_{c,\rptswg}$ has a 
unique fixed point.

We first show that $\Phi_{c,\rptswg}$ has a fixed point.
By Lem.~\ref{lem:tfSeqWellDefPTSDWeak},
the poset $(\Setsto{X}{\Rptsw},\leqRptsw)$ (here $\leqRptsw$ denotes the pointwise extension) 
is a complete lattice.
Therefore we can construct a transfinite sequence 
\begin{multline*}
\botRptsw\leqRptsw \Phi_{c,\rptswg}(\botRptsw) \leqRptsw \cdots 
\leqRptsw \Phi_{c,\rptswg}^{\mathfrak{a}}(\botRptsw)\leqRptsw \cdots
\end{multline*}
as in
Thm.~\ref{thm:KTCC}.2.
By Thm.~\ref{thm:KTCC}.2,
there exists an ordinal $\mathfrak{n}$ such that
$\Phi_{c,\rptswg}^\mathfrak{n}(\botRptsw)$ is the least fixed point of $\Phi_{c,\rptswg}$. 
Let $f=\Phi_{c,\rptswg}^\mathfrak{n}(\botRptsw)$.

It remains to show that $f$ is the unique fixed point of $\Phi_{c,\rptswg}$.
Let $g:X\to \Rptsw$ be a fixed point of $\Phi_{c,\rptswg}$.
As $f$ is the least fixed point, 
we have $f(x)\leqRptsw g(x)$ for each $x\in X$. 
We now define $h:X\to \Rptsw$ by $h(x)=g(x)-f(x)$.
Then we have:
\allowdisplaybreaks[4]
\begin{align*}
&\sup_{x\in X}h(x) \\
&= \sup_{x\in X}\bigl(g(x)-f(x)\bigr)  \hspace{-5.5cm}& \text{(by def.\ of $h$)} \\
&= \sup_{x\in X}\bigl(\Phi_{c,\rptswg}(g)(x)-\Phi_{c,\rptswg}(f)(x)\bigr)  \hspace{-5.5cm}& \\
& &\text{($f$ and $g$ are fixed points)} \\
&\leq \sup_{x\in X}\bigl(\gamma\cdot\sum_{x'\in X}c_1(x)(x')\cdot g(x') \hspace{-5.5cm}\\
&\qquad\qquad\qquad-\gamma\cdot\sum_{x'\in X}c_1(x)(x')\cdot f(x')\bigr)\hspace{-5.5cm} & \\
& \qquad\qquad\text{(by def.\ of $\rptswg$ and that}  \hspace{-5.5cm}\\
& & c_2(x)=1\,\Rightarrow\,\Phi_{c,r}(f)(x)=\Phi_{c,r}(g)(x)=1) \\
%
%
%
&= \gamma\cdot\sup_{x\in X}\sum_{x'\in X}c_1(x)(x')\cdot(g(x')-f(x')) \hspace{-5.5cm} & \\
&= \gamma\cdot\sup_{x\in X}\sum_{x'\in X}c_1(x)(x')\cdot h(x') \hspace{-5.5cm}& (\text{by def.\ of $h$})\\
&\leq \gamma\cdot\sup_{x\in X}\sup_{x'\in X}h(x')  \hspace{-5.5cm}& (\text{by $\sum_{x'\in X}c_1(x)(x')\leq 1$})\\
&= \gamma\cdot \sup_{x\in X}h(x) \hspace{-5.5cm}\,.
\end{align*}
As $0\leq \gamma<1$, we have $\sup_{x\in X}h(x)=0$, and 
this implies 
$f=g$. 
This concludes the proof.
\end{proof}

\begin{proof}[Proof (Prop.~\ref{prop:weakrankDomPTS})]
We prove that $(\rptswg,\qptsw,\leqRptsw)$ satisfies the conditions in Def.~\ref{def:rankingdom}.

\subsubsection*{Cond.~\ref{item:def:rankingdom2}}
Let $(\delta,t)\in \Fpts\Rptsw=\dist \Rptsw\times\{0,1\}$.
Then we have:
\allowdisplaybreaks[4]
\begin{align*}
&\qptsw\circ \rptswg(\delta,t) \\
&= \begin{cases}
\qptsw(1) & (t=1) \\
\qptsw\bigl(\gamma\cdot\sum_{a\in \supp(\delta)}a\cdot \delta(a)\bigr) & (t=0)
\end{cases} 
& (\text{by def.\ of $\rptswg$})\\
&= \begin{cases}
1 & (t=1) \\
\gamma\cdot\sum_{a\in \supp(\delta)}a\cdot \delta(a) & (t=0)
\end{cases}
& (\text{by def.\ of $\qptsw$})\\
&\leq \begin{cases}
1 & (t=1) \\
\sum_{a\in \supp(\delta)}a\cdot \delta(a) & (t=0)
\end{cases}
& (\text{by $\gamma<1$})\\
&=\sigmapts(\delta,t)
& (\text{by def.\ of $\sigmapts$})\\
&=\sigmapts\circ \Fpts\qptsw(\delta,t)
& (\text{by def.\ of $\qptsw$})\,.
\end{align*}
Hence Cond.~\ref{item:def:rankingdom2} is satisfied.
%
 

\subsubsection*{Cond.~\ref{item:def:rankingdom4}}
We have already proved in Lem.~\ref{lem:tfSeqWellDefPTSDWeak}.

\subsubsection*{Cond.~\ref{item:def:rankingdom3}}
Immediate from $\qptsw=\id_{\Rptsw}$.

%
%


\subsubsection*{Cond.~\ref{item:def:rankingdom7}}
It is already proved in Lem.~\ref{lem:corecAlgProb}.

Hence $(\rptswg,\qpts,\leqRptsw)$ is a ranking domain.
\end{proof}

\vspace{0.5\baselineskip}
\subsection{\textbf{Towards Further Examples}}\label{sec:nonTotalRD}
Now that we have a general categorical axiomatization of ranking functions (\S{}\ref{sec:catTheoryRF}), we would like to exploit it in deriving further examples of ``ranking functions'' that are previously unknown, hoping that they will provide novel proof methods for various liveness properties. In the previous section we derived two variations of ranking supermartingales.  Towards further examples, here we indicate a possible direction. 

We can say that not many concrete examples are known of corecursive algebras. Nevertheless, the following ``closure properties'' can be used to derive new examples. 

\begin{mylemma}\label{lem:closurePropOfCorecAlg} 
Let $F\colon \mathbb{C}\to\mathbb{C}$ be a functor.
\begin{enumerate}
%
 
 \item\label{item:lem:closurePropOfCorecAlg3}{} (\cite{adamekHM14corecAlg})
  Consider 
 the well-known construction of 
 the \emph{final sequence}:
 \begin{math}
  1 \stackrel{!}{\leftarrow}
  F1 \stackrel{F!}{\leftarrow}
  F^2 1 \stackrel{F^{2}!}{\leftarrow}
  \cdots
  \stackrel{}{\leftarrow}
  F^{\mathfrak{a}}1
  \stackrel{}{\leftarrow}
  \cdots
 \end{math}, where we define $F^{\mathfrak{a}}1$ for an arbitrary ordinal $\mathfrak{a}$ (we assume enough limits and let $F^{\mathfrak{a}}1=\lim_{\mathfrak{b}<\mathfrak{a}}F^{\mathfrak{b}}1$ for a limit ordinal $\mathfrak{a}$). 
Then,  for each  $\mathfrak{a}$, the algebra $F^{\mathfrak{a}}!\colon F(F^{\mathfrak{a}} 1)\stackrel{}{\to} F^{\mathfrak{a}} 1$ is corecursive.

 \item\label{item:lem:closurePropOfCorecAlg4}{} (dual of~\cite{caprettaUV06recursivecoalgebra}) 
 Let $\tau\colon F\Rightarrow G$ be a natural transformation. Then a corecursive $G$-algebra $s \colon GR \to R$ induces a corecursive $F$-algebra $s\circ \tau_{R}\colon FR\to R$.
 \qed
\end{enumerate}
\end{mylemma}

Lem.~\ref{lem:closurePropOfCorecAlg}.\ref{item:lem:closurePropOfCorecAlg3}--\ref{item:lem:closurePropOfCorecAlg4}
together suggest
the following workflow. For dynamical systems modeled as  $F$-coalgebras, we pick a natural transformation $\tau\colon F\Rightarrow G$ and an ordinal $\mathfrak{a}$---the former \emph{abstracts} (or \emph{collapses}) $F$-behaviors into $G$-behaviors that are supposedly simpler. We then use the set $G^{\mathfrak{a}}1$ of ``$G$-behaviors up-to $\mathfrak{a}$'' as the ranking domain. The set carries a corecursive $G$-algebra by Lem.~\ref{lem:closurePropOfCorecAlg}.\ref{item:lem:closurePropOfCorecAlg3}; and via $\tau$ it carries a corecursive $F$-algebra, too (Lem.~\ref{lem:closurePropOfCorecAlg}.\ref{item:lem:closurePropOfCorecAlg4}). 

An example of such ``behavioral'' ranking domains is in~\S{}\ref{sec:detailBinTree}.
 The set $\Rtree$ consists of: (unlabeled) trees that are possibly countably branching and of finite depth; and the special element $\bot$ that designates non-well-foundedness. It is shown that, for the problem of universal reachability of tree automata, we can indeed use  the set $\Rtree$  to form a (categorical) ranking domain. Here we use the functors $F=\bigl(\,\coprod_{n\in\mathbb{N}\cup\{\omega\}}\Sigma_n\times(\place)^{n}\,\bigr)\times\{0,1\}$ and $G=1+\coprod_{n\in\mathbb{N}\cup\{\omega\}}(\place)^{n}$, and $\tau\colon F\Rightarrow G$ collapses the elements $(t,1)\in FX$---i.e.\ accepting states---to the unique element of $1$ in $GX$. 

In fact the set $\Rtree$ is not precisely the outcome $G^{\omega}1$ of the workflow described above: $G^{\omega}1$ contains all the finite and infinite trees (with suitable branching degrees) as separate elements; but in  $\Rtree$ a single element $\bot$ stands for all the ``non-well-founded'' trees that contains at least one infinite branch. A categorical description of such collapse is our future work; so is a general order with which we can equip $G^{\omega}1$.

\vspace{0.5\baselineskip}
\subsection{\textbf{A Ranking Domain for Tree Automata}}\label{sec:detailBinTree}
\begin{mynotation}\label{nota:trees}
We write $\langle\rangle$ for the empty sequence in $\mathbb{N}^*$.
For $w,w'\in\mathbb{N}^{*}\cup\mathbb{N}^\omega$, we write $w\treeprefix w'$ if $w$ is a prefix of $w'$.
\end{mynotation}

\begin{mydefinition}[unlabeled tree]\label{def:unlabeledtree}
An \emph{unlabeled 
tree} is a set $D\subseteq \mathbb{N}^*$ 
that satisfies the following conditions.
\begin{enumerate}
\item\label{item:def:tree1}
The empty sequence is in $D$, i.e.\ $\empseq\in D$.

\item\label{item:def:tree2}
The set $D$ is prefix-closed, i.e.\ $\forall w\in\mathbb{N}^*.\;\forall i\in\mathbb{N}.\;wi\in D\,\Rightarrow\,w\in D$.

\item\label{item:def:tree3}
The set $D$ is downward-closed, i.e\ $\forall w\in\mathbb{N}^*.\; \forall i\in\mathbb{N}.\;\forall j\leq i.\;
wi\in D\,\Rightarrow wj\in T$.
\end{enumerate}
An unlabeled tree $D$ is said to be \emph{finite-depth} if it satisfies the following additional condition.
\begin{enumerate}
\setcounter{enumi}{3}
\item\label{item:def:tree4}
The set $D$ has no strictly increasing sequence with respect to $\treeprefix$, i.e.\
\begin{displaymath}
\forall v\in\mathbb{N}^\omega.\; |\{w\in D\,\mid\,w\treeprefix v\}|<\infty\,.
\end{displaymath}
\end{enumerate}
We write $\Tree$ (resp.\ $\Treefin$) for the set of all unlabeled trees (resp.\ unlabeled finite-depth trees).

For an unlabeled tree $D$, we define a function $\branch_D:D\to\mathbb{N}\cup\{\infty\}$ as follows.
\begin{equation*}
\branch_D(w)=\max\{i\mid wi\in D\}\,.
\end{equation*}

The \emph{prefix order} over $\Tree$ is a partial order $\treeprefix$ that is defined by
\begin{displaymath}
D_1\treeprefix D_2\;\defarrow\,
\left(\begin{aligned}
&D_1\subseteq D_2,\;\text{and}\\
&\forall w\in D_1.\; 
\branch_{D_1}(w)\in\{0,\branch_{D_2}(w)\}
\end{aligned}
\right)\,.
\end{displaymath}

For a (possibly infinite) family of unlabeled trees $(D_i\in\Tree)_{0\leq i< \mathfrak{a}}$ where $\mathfrak{a}\leq\omega$, 
we define a new tree $\treecomb(D_1,D_2,\ldots)\subseteq\mathbb{N}^*$ as follows.
\begin{displaymath}
\treecomb(D_1,D_2,\ldots)
\;=\;
\{\empseq\}\cup\{iw\mid i<\mathfrak{a}, w\in D_i\} 
\end{displaymath}
\end{mydefinition}

\begin{myremark}\label{rem:finitenessTree}
We note that $D\in\Treefin$ is assumed to be \emph{finite-depth}, but it is not assumed to be \emph{finitely-branching}.
For example, we have
$(\{\empseq\}\cup\{\langle n\rangle \mid n\in\mathbb{N}\})\in \Treefin$
where $\langle n\rangle$ denotes a sequence whose length is $1$.
\end{myremark}

\begin{mydefinition}[ranked alphabet]\label{def:rankedAlph}
A \emph{ranked alphabet} is a pair $\Sigma=(\Sigma,|\place|)$ of a set $\Sigma$ and an \emph{arity function} 
$|\place|:\Sigma\to\mathbb{N}\cup\{\infty\}$.
For each $a\in\Sigma$, $|a|\in\mathbb{N}\cup\{\infty\}$ is called the \emph{arity} of $a$.
For each $n\in\mathbb{N}\cup\{\infty\}$, we write $\Sigma_n$ for $\{a\in\Sigma\mid |a|=n\}\subseteq \Sigma$.
\end{mydefinition}

\begin{mydefinition}[labeled tree]\label{def:labeledtree}
Let $\Sigma$ be a ranked alphabet.
A \emph{$\Sigma$-labeled 
tree} is a pair $T=(D,l)$ of
an unlabeled tree $D\subseteq\Tree$ and a \emph{labeling function} $l:D\to\Sigma$ 
that 
respects the arity, i.e.\
\begin{equation*}
\forall w\in D.\; l(w)=\branch_D(w)\,.
\end{equation*}
%
%
%
%
The unlabeled tree $D$ is called the \emph{domain} of $T$.
%
We write $\TreeSigma$  for the set of all $\Sigma$-labeled trees.
\end{mydefinition}

\begin{mydefinition}[$\FtreeSigma$]\label{def:treeFunc}
Let $\Sigma$ be a ranked alphabet.
We define a functor $\FtreeSigma:\Sets\to\Sets$ 
by 
\begin{displaymath}
\FtreeSigma = \Bigl(\coprod_{n\in\mathbb{N}\cup\{\omega\}}\Sigma_n\times(\place)^{n}\Bigr)\times\{0,1\}\,.
\end{displaymath}
\end{mydefinition}

\begin{mynotation}\label{nota:ajiofdps}
Let 
$c:X\to \FtreeSigma X$ be an $\FtreeSigma$-coalgebra.
For each $x\in X$ where $c(x)=(\xi,t)$, we write $c_1(x)$ and $c_2(x)$ for 
$\xi\in\coprod_{n\in\mathbb{N}\cup\{\omega\}}\Sigma_n\times(\place)^{n}$ and $t\in\{0,1\}$, respectively.
\end{mynotation}

\begin{mydefinition}[run tree]\label{def:runtree}
Let $c:X\to \FtreeSigma X$ be an $\FtreeSigma$-coalgebra. 
We define a new ranked alphabet $X\times\Sigma=(X\times\Sigma,|\place|)$ by 
$|(x,a)|=|a|$.
%
For $x \in X$,
an $(X\times\Sigma)$-labeled tree $(D,l)$ 
is called a \emph{run tree of $c$ from $x$} if it 
satisfies the following conditions (here for each $w\in D$ where $l(w)=(x',a)$, $l_1(w)$ and $l_2(w)$ denote $x'$ and $a$ respectively):
\begin{enumerate}
\item $l_1(\empseq)=x$; and

\item for each $w\in D$, 
\begin{equation*}
c_1(l_1(w))=
\bigl(l_2(w),l_1(w0),l_1(w1),\ldots)\,.
\end{equation*}

\end{enumerate}

A run tree $(D,l)$ is called \emph{accepting} if it has no infinite branch labeled only with non-accepting states, i.e.\
 \begin{align*}
 \neg
 \left(\begin{aligned}
 \exists v\in\mathbb{N}^\omega.\; 
 \forall w\in\mathbb{N}^*\;\text{s.t.}\;w\treeprefix v.\; 
 (w\in D\wedge c_2(l(w))=0)
 \end{aligned}\right)\,.
 \end{align*}
%
\end{mydefinition}

\begin{myproposition}\label{prop:uniqueTree}
For an $\FtreeSigma$-coalgebra $c:X\to\FtreeSigma X$ and $x\in X$,
there exists 
a unique run tree of $c$ from $x$.
\qed
\end{myproposition}

\begin{mydefinition}[$\RunTree_c$ and $\Reach_c$]\label{def:reachTree}
By Prop.~\ref{prop:uniqueTree}, we write $\RunTree_c(x)$ for the unique run tree of $c$ from $x$.
For an $\FtreeSigma$-coalgebra $c:X\to\FtreeSigma X$, we define a function $\Reach_c:X\to \Omegatree$ by
\begin{displaymath}
\Reach_c(x)\;=\;\begin{cases}
1 & (\text{$\RunTree_c(x)$ is accepting}) \\
0 & (\text{otherwise}) \,.
\end{cases}
\end{displaymath}
\end{mydefinition}

\begin{mydefinition}\label{def:TVDsigma}
We define an $\FtreeSigma$-modality $\sigmatree:\FtreeSigma\Omegatree\to\Omegatree$ over 
a truth-value domain $(\Omegatree,\leqOmegatree)$ by 
\begin{displaymath}
\sigmatree\bigl((a,p_1,p_2,\ldots),t\bigr)=
\begin{cases}
1 & (t=1\;\text{or}\;\forall i.\, p_i=1) \\
0 & (\text{otherwise})\,.
\end{cases}
\end{displaymath}
\end{mydefinition}

\begin{myproposition}\label{prop:BinTreeLFP}
The $\FtreeSigma$-modality $\sigmatree:\FtreeSigma\Omegatree\to\Omegatree$ in Def.~\ref{def:TVDsigma} has least fixed points.
For an $\FtreeSigma$-coalgebra $c:X\to\FtreeSigma X$, the (coalgebraic) least fixed property 
$\sem{\mu\sigmatree}_c:X\to\Omegatree$ is given 
by the function $\Reach_c:X\to\Omegatree$ in Def.~\ref{def:reachTree}.
\end{myproposition}

\begin{proof}
We define $f:X\to\Omegatree$ by $f=\Reach_c$.
It suffices to prove that $f$ is the least fixed-point of $\Phi_{c,\sigmatree}:\Setsto{X}{(\Rtree)}\to\Setsto{X}{(\Rtree)}$.

We first show that $f$ is a fixed point of $\Phi_{c,\tree}$.
Let $x\in X$ and assume that $c(x)=((a,x_1,x_2,\ldots),t)\in\FtreeSigma X$.
Then we have:
\begin{align*}
&\sigmatree\circ \FtreeSigma f((a,x_1,x_2,\ldots),t)\\
&= \begin{cases}
1 & (t=1\;\text{or}\;\forall i.\, f(x_i)=1) \\
0 & (\text{otherwise}) 
\end{cases}
& (\text{by def.\ of $\sigmatree$}) \\
&= \begin{cases}
1 & (t=1\;\text{or}\;\forall i.\, \RunTree_c(x_i)\;\text{is accepting}) \\
0 & (\text{otherwise}) 
\end{cases} \hspace{-3cm}&\\
& & (\text{by def.\ of $f$}) \\
&= \begin{cases}
1 & (\RunTree_c(x)\;\text{is accepting}) \\
0 & (\text{otherwise}) 
\end{cases} \hspace{-3cm}&\\
& & (\text{by def.\ of $\RunTree_c$}) \\
&=f(x)\,.
\end{align*}
Hence $f$ is a fixed point of $\Phi_{c,\sigmatree}$.

It remains to show that $f$ is the least fixed point.
Let $f':X\to\Omegatree$ be a fixed point of $\Phi_{c,\sigmatree}$.
It suffices to prove that for each $x\in X$, $f'(x)=0$ implies $f(x)=0$.

Let $x\in X$ and assume $f'(x)=0$.
Let $(D,l)=\RunTree_c(x)$.
For each $i\in\mathbb{N}$, 
we inductively define $x_i\in X$ and $v_i\in\mathbb{N}$ so that $f'(x_i)=0$ for each $i$ as follows.
%
%
\begin{itemize}
\item
For $i=0$, we let $x_0=x$. 
Then
 we have $f'(x_0)=f(x)=0$.

\item
Let $i\geq 0$ and assume that we have defined $x_j$ and $v_{j-1}$ for each $j\leq i$ so that 
$f'(x_i)=0$.
Let $c(x_i)=((a_i,x_{i,1},x_{i,2},\ldots),t_i)$.
As $f'$ is a fixed point of $\Phi_{c,\sigmatree}$, we have: 
\begin{equation}\label{eq:1601071124}
0=f'(x_i) 
=(\sigmatree\circ \FtreeSigma f'\circ c)(x_i)\,.
\end{equation}
By the definition $\sigmatree$, 
this means that there exists $k$ such that $f'(x_{i,k})=0$.
We let $v_i=k$ and $x_{i+1}=x_{i,k}$.
\end{itemize}

We can prove that $x_i$ and $v_i$  satisfy the followings.
\begin{enumerate}
\item\label{item:16010711271} $\forall i.\; (v_0\ldots v_{i-1})\in D$

\item\label{item:16010711272} $\forall i.\; l(v_0\ldots v_{i-1})=x_i$

\item\label{item:16010711273} $\forall i.\; c_2(l(v_0\ldots v_{i-1}))=0$
\end{enumerate}
Here (\ref{item:16010711271}) and (\ref{item:16010711272}) are immediate from the definition.
The equation (\ref{eq:1601071124}) implies that $c_2(l(v_0\ldots v_{i-1}))=c_2(x_i)=t_i=0$ for each $i\in\mathbb{N}$.


This means that $\RunTree_c(x)$ is not accepting (see Def.~\ref{def:runtree}), and therefore we have $f(x)=0$.
This concludes the proof.
\end{proof}

\begin{myproposition}\label{prop:BinTreeAsm}
The $\Ftree$-modality $\sigmatree:\Ftree\Omegatree\to\Omegatree$ in Def.~\ref{def:TVDsigma} satisfies conditions in Asm.~\ref{asm:behDomainrank}.
\qed
\end{myproposition}

%

\begin{myproposition}\label{prop:rankDomBinTree}
Let $\Sigma$ be a ranked alphabet and $\Rtree=\Treefin \cup\{\botRtree\}$.
We define an $\FtreeSigma$-algebra $\rtree:\FtreeSigma(\Rtree)\to\Rtree$,
a function $\qtree:\Rtree\to\Omegatree$ and 
a partial order $\leqRtree$ over $\Rtree$
as follows.
\begin{enumerate}
\item 
\raisebox{7.5mm}{
\begin{minipage}[t]{0.95\hsize}
\begin{multline*}
\!\!\!\!\!\!\rtree\bigl((a,D_1,D_2,\ldots),t\bigr)\\
=\begin{cases} \empseq & (t=1) \\ 
\treecomb(D_1,D_2,\ldots) & (t=0\;\text{and}\;\forall i.\; D_i\neq\botRtree) \\
\botRtree & (t=0\;\text{and}\;\exists i.\; D_i=\botRtree) \end{cases}
\end{multline*}
\end{minipage}
}

\item $\qbtr(D)=\begin{cases} 1 & (D\neq\botRtree) \\ 0 & (D=\botRtree)\,.\end{cases}$

\item $D_1\leqRtree D_2\,\defarrow\, D_1=\botRtree\;\text{or}\;(D_1,D_2\in\Treefin\;\text{and}\;D_1\treepostfix D_2)$ (note the direction).
\end{enumerate}
Then $(\rtree,\qtree,\leqRtree)$ 
is a ranking domain.
\end{myproposition}

\begin{mylemma}\label{lem:corecTree0}
The poset  $(\Rtree,\leqRtree)$ is a complete lattice.
\end{mylemma}

\begin{proof}
Let $K\subseteq\Rtree$.

\myparagraph{$K$ has the least upper bound}
We first prove that $K$ has the least upper bound.
If $K\cap\Treefin=\emptyset$, then $\botRtree$ is the least upper bound of $K$.

Assume $K\cap\Treefin\neq\emptyset$. 
If $K'=K\setminus\{\botRtree\}$ has the least upper bound,
then it is also the least upper bound of $K$.

We define $A\subseteq \mathbb{N}^*$ as follows. 
\begin{equation*}
A=\bigcup\{D'\in\Treefin\mid \forall D\in K'.\; D'\treeprefix D\}\,.
\end{equation*}
As $\{\empseq\}\treeprefix D$ holds for each $D\in K'$ and 
Cond.~\ref{item:def:tree1}--\ref{item:def:tree3}
in Def.~\ref{def:unlabeledtree} are preserved by the union,
$A\in\Tree$.
We prove that $A$ is the least upper bound of $K'$.

We first prove that $A$ is an upper bound.
Let $D\in K'$. We prove $D\leqRtree A$.
To this end, by the definition of $\treeprefix$,
it suffices to  prove $A\subseteq D$ and 
$\branch_A(w)\in\{0,\branch_D(w)\}$ for each $w\in A$.
The former is immediate from the definition of $A$.
The latter is 
satisfied because
we have $\branch_{D'}(w)\in\{0,\branch_D(w)\}$
for each $D'$ such that $D'\treeprefix D$ and $w\in D'$.

Hence $A$ is an upper bound of $K$.
We note that this also proves that $A$ is finite-depth.

It is immediate from the definition of $A$ that $A$ is the least upper bound.

\myparagraph{$K$ has the greatest lower bound}
We prove that $K$ has the greatest lower bound.
If $\botRtree\in K$, then $\bigsqcap K=\botRtree$.
If $K=\emptyset$, then $\bigsqcap K=\{\empseq\}$.

We assume $\botRtree\notin K$ and $K\neq\emptyset$.
For $D_1,D_2\in K$, if there exists $D_3\in K$ such that 
$D_1\treeprefix D_3$ and $D_2\treeprefix D_3$ then by the definition of $\treeprefix$,
we have:
\begin{align*}
\forall w\in D_1\cap D_2.\; &(\branch_{D_1}(w)\neq 0\;\text{and}\;\branch_{D_2}(w)\neq 0)\\
&\qquad\Rightarrow\;
\branch_{D_1}(w)=\branch_{D_2}(w)\,.
\end{align*}

Therefore if there exists $D_1,D_2\in K$ that satisfy
\begin{align}\label{eq:1601071738}
\exists w\in D_1\cap D_2.\;& \branch_{D_1}(w)\neq 0,\;\branch_{D_2}(w)\neq 0\;\notag\\
&\qquad\text{and}\;\branch_{D_1}(w)\neq\branch_{D_2}(w),
\end{align}
then there does not exist $D_3\in K$ such that $D_1\treeprefix D_3$ and $D_2\treeprefix D_3$.
This implies $\bigsqcap K=\botRtree$.

Assume that there does not exist $D_1,D_2\in K$ that satisfy (\ref{eq:1601071738}).
We define $B\in\mathbb{N}^*$ by 
\begin{equation*}
B=\bigcup K\,.
\end{equation*}

It is easy to see that $B\in \Tree$.
It is also easy to see that
if $B$ is not finite-depth then $\botRtree$ is the greatest lower bound of $K$.

Assume $B\in\Treefin$.
We prove that $B$ is the greatest lower bound of $K$.
Let $D\in K$. By the definition of $B$, we have $D\subseteq B$.
As there exist no $D_1,D_2\in K$ that satisfy (\ref{eq:1601071738}), we have $\branch_{D}(w)\in\{0,\branch_{B}(w)\}$ for each $w\in D$.
Hence we have $B\leqRtree D$, and therefore $B$ is a lower bound of $K$.
It is immediate from its definition that $B$ is the greatest lower bound.
\end{proof}

\begin{mylemma}\label{lem:corecTree1}
The triple $(\rtree,\qtree,\leqRtree)$ in Prop.~\ref{prop:rankDomBinTree} satisfies 
Cond.~\ref{item:def:rankingdom4} in Def.~\ref{def:rankingdom}.
\end{mylemma}

\begin{proof}
%
%
Cond.~\ref{item:def:rankingdom40and3} is immediate from Lem.~\ref{lem:corecTree0}.

It is easy to see that 
Cond.~\ref{item:def:rankingdom41and2} is satisfied.
\end{proof}

\begin{mylemma}\label{lem:corecTree2}
The algebra $\rtree:\FtreeSigma\Rtree\to\Rtree$ in Prop.~\ref{prop:rankDomBinTree} is a 
corecursive algebra. 
\end{mylemma}

\begin{proof}
It suffices to show that the function $\Phi_{c,\rtree}:\Setsto{X}{(\Rtree)}\to\Setsto{X}{(\Rtree)}$ 
has a unique fixed point.

By Lem.~\ref{lem:corecTree1}, the set $(\Setsto{X}{(\Rtree)},\leqRtree)$ where 
$\leqRtree$ denotes the pointwise extension 
is a complete lattice.
Therefore by 
Thm.~\ref{thm:KTCC}, the function $\Phi_{c,\rtree}$, which is monotone by Lem.~\ref{lem:corecTree1},
has the least fixed point.

It remains to show that this is the unique fixed point.
Let $f,f':X\to\Rtree$ be fixed points of $\Phi_{c,\rtree}$.

Let $g:X\to\Rtree$ and assume $g(x)=\botRtree$.
Let $c(x)=((a,x_0,x_1,\ldots),t)$.
As $g$ is a fixed point of $\Phi_{c,\rtree}$, 
we have $(\rtree\circ\FtreeSigma g\circ c)(x)=\botRtree$.
Hence we have $t=0$ and that there exists $i$ such that $g(x_i)=\botRtree$.

Therefore we can inductively define a sequence $x^0x^1x^2\ldots\in X^\omega$ so that:
\begin{enumerate}
\item $g(x^i)=\botRtree$;
\item $c_2(x^i)=0$ for each $i\in\omega$; and
\item $c_1(x^i)=(a^i,x^i_1,x^i_2,\ldots)$ then $x^{i+1}=x^i_n$ for some $n$.
\end{enumerate}
%
Conversely, we can prove that existence of a sequence $x^0x^1x^2\ldots\in X^\omega$
that satisfies the conditions above implies $g(x)=\botRtree$.


Therefore we have: 
$f(x)=\botRtree$ if and only if $f'(x)=\botRtree$.


We now prove that for each $n\in\mathbb{N}$, $w\in \mathbb{N}^n$ and $x\in X$ such that $f(x)\neq\botRtree$ and $f'(x)\neq\botRtree$,
$w\in f(x)$ if and only if $w\in f'(x)$
by the induction on $n$.

For $n=0$,  by Def.~\ref{def:unlabeledtree}, we have $\empseq\in f(x)$ and $\empseq\in f'(x)$ for each $x$.

Let $n>0$ and assume that for each $w\in \mathbb{N}^n$ and $x\in X$ such that $f(x)\neq\botRtree$ and $f'(x)\neq\botRtree$, 
we have $w\in f(x)$ if and only if $w\in f'(x)$.
Let $x\in X$ and assume $c(x)=((a,x_0,x_1,\ldots),t)\in(\Sigma\times (\place)^m)\times\{0,1\}$.
For each $i\in\mathbb{N}$, we have:
\begin{align*}
&iw\in f(x) \\
&\Leftrightarrow\; iw\in\Phi_{c,\rtree}(f) & (\text{$f$ is a fixed point}) \\
&\Leftrightarrow\; i< m\;\text{and}\; w\in f(x_i) & (\text{by def.\ of $\rtree$}) \\
&\Leftrightarrow\; i< m\;\text{and}\; w\in f'(x_i) & (\text{by IH}) \\
&\Leftrightarrow\; iw\in\Phi_{c,\rtree}(f') & (\text{by def.\ of $\rtree$}) \\
&\Leftrightarrow\; iw\in f'(x)& (\text{$f'$ is a fixed point}).
\end{align*}

Hence we have $f(x)=f'(x)$ for each $x\in X$.
This concludes the proof.
\end{proof}

\begin{proof}[Proof (Prop.~\ref{prop:rankDomBinTree})]
We prove that $(\rtree,\qtree,\leqRtree)$ satisfies conditions in Def.~\ref{def:rankingdom}.

\subsubsection*{\ref{item:def:rankingdom2}}
Let $((a,D_0,D_1,\ldots),t)\in\FtreeSigma \Rtree$.
Then we have:
\begin{align*}
&(\qtree\circ \rtree)\bigl((a,D_0,D_1,\ldots),t\bigr)=1 \\
&\Leftrightarrow\; \rtree\bigl((a,D_0,D_1,\ldots),t\bigr)\neq\botRtree & (\text{by def.\ of $\qtree$}) \\
&\Leftrightarrow\; t=1\;\text{or}\;\forall i.\; D_i\neq \botRtree & (\text{by def.\ of $\rtree$}) \\
&\Leftrightarrow\; t=1\;\text{or}\;\forall i.\; \qtree(D_i)= 1 & (\text{by def.\ of $\qtree$}) \\
&\Leftrightarrow\; (\sigmatree\circ \FtreeSigma\qtree)\bigl((a,D_0,D_1,\ldots),t\bigr)=1 & (\text{by def.\ of $\sigmatree$})\,.
\end{align*}
Hence we have $\qtree\circ \rtree=\sigma\circ \FtreeSigma\qtree$.

\subsubsection*{\ref{item:def:rankingdom4}}
We have already proved in Lem.~\ref{lem:corecTree1}.

\subsubsection*{\ref{item:def:rankingdom3}}
Monotonicity and strictness are immediate from the definition of $\qtree$.
We prove that $\qtree$ is 
continuous.
Let $K\subseteq \Rtree$.
Then we have:
\begin{align*}
&\bigsqcup_{D\in K}\qtree(D)=0 \\
&\Leftrightarrow\;
\forall D\in K.\; \qtree(D)=0 \\
&\Leftrightarrow\;
K=\{\botRtree\} & (\text{by def.\ of $\qtree$}) \\
&\Leftrightarrow\; \bigsqcup K= \botRtree  \\ 
&\Leftrightarrow\; \qtree(\bigsqcup K)=0 &  (\text{by def.\ of $\qtree$})\,.
\end{align*}
Hence $\qtree$ is 
continuous.

\subsubsection*{\ref{item:def:rankingdom7}}
We have already proved in Lem.~\ref{lem:corecTree2}.
\end{proof}

\begin{myproposition}[completeness]\label{prop:treeCompleteness}
For an arbitrary $\FtreeSigma$-coalgebra $c:X\to\FtreeSigma X$, 
there exists a ranking arrow $b:X\to\Rtree$
with respect to the ranking domain  $(\rtree,\qtree,\leqRtree)$ in Prop.~\ref{prop:rankDomBinTree}
that satisfies the following.
\begin{equation*}
\qtree\circ b\;=\;\sem{\mu\sigmatree}_c
\end{equation*}

\end{myproposition}

\begin{proof}
Immediate from $\qtree\circ \rtree=\sigmatree\circ \FtreeSigma\qtree$ (see the proof of Prop.~\ref{prop:rankDomBinTree}) 
and Prop.~\ref{prop:completeness}.
\end{proof}

\auxproof{
\begin{mylemma}\label{lem:unlabeling}
Let $\Sigma$ be a ranked alphabet and
$c:X\to \FtreeSigma X$ be an $\FtreeSigma$-coalgebra.
We define a ranked alphabet $X\times \Sigma$ as in Def.~\ref{def:runtree}.
For a $(X\times\Sigma)$-labeled tree $T=(D,l)$, 
there exists 
a unique unlabeled tree $\AccTree_c(T)$ that satisfies the following conditions.
\begin{itemize}
\item $\AccTree_c(T)\treeprefix D$
\item $\forall w\in \AccTree_c(T).\; c_2(l_1(w))=0\;\text{and}\;l_2(w)>0 \Rightarrow w0\in D_1$ 
\item $\forall w\in \AccTree_c(T).\; c_2(l_1(w))=1\Rightarrow w0\notin D_1$ 
\end{itemize}
Here $l_1:D\to X$ and $l_2:D\to \Sigma$ are defined as in Def.~\ref{def:runtree}.
\qed
\end{mylemma}

\begin{myproposition}\label{prop:concreteuniquefp}
We assume the conditions in Prop.~\ref{prop:rankDomBinTree}.
As $\rtree:\FtreeSigma\Rtree\to\Rtree$ is a corecursive algebra,
there  exists a unique function
$\uniquefp{c}_{\rtree}:X\to\Rtree$ such that $\uniquefp{c}_{\rtree}=\Phi_{c,\rtree}( \uniquefp{c}_{\rtree})$.
Here the unique function $\uniquefp{c}_{\rtree}$ is given as follows.
%
\begin{displaymath}
\uniquefp{c}_{\rtree}(x)\;=\; 
\begin{cases}
\AccTree_c\bigl(\RunTree_c(x)\bigr) \hspace{-5cm}& \\
&(\text{$\AccTree_c\bigl(\RunTree_c(x)\bigr)$ is finite-depth}) \\
\botRtree & (\text{otherwise})\,.
\end{cases}
\end{displaymath}
\vspace{-1.7\baselineskip}\\
\qed
\end{myproposition}
}


\end{document}